\documentclass[acmsmall,nonacm]{acmart}

\usepackage[utf8]{inputenc}
\usepackage{amsmath,amsthm,graphics}

\usepackage[bibliography=common]{apxproof}
\usepackage{tikz}
\usepackage{bm}
\usepackage{tabularx}
\usepackage{subcaption}

\usetikzlibrary{automata,positioning}

\definecolor{amaranth}{rgb}{0.9, 0.17, 0.31}
\definecolor{americanrose}{rgb}{1.0, 0.01, 0.24}
\definecolor{darkgreen}{rgb}{0.0, 0.2, 0.13}
\definecolor{darkpastelgreen}{rgb}{0.01, 0.75, 0.24}

\usetikzlibrary{calc,decorations.pathmorphing,shapes}

\newcounter{sarrow}

\usepackage{xspace}            %

\newcommand{\mult}{\mathtt{mult}}

\newcommand{\resset}{\mathtt{RES}_{\mathrm{set}}}
\newcommand{\resbag}{\mathtt{RES}_{\mathrm{bag}}}
\newcommand{\resprobset}{\textsc{Res}_{\mathrm{set}}}
\newcommand{\resprobbag}{\textsc{Res}_{\mathrm{bag}}}
	
\newcommand{\datarule}{{\,:\!\!-\,}} %

\newcommand{\calN}{\mathcal{N}}	
\newcommand{\calL}{\mathcal{L}}
\newcommand{\calH}{\mathcal{H}}

\newcommand\red{\mathtt{IF}}

\newcommand{\query}{Q}

\newcommand{\introparagraph}[1]{\textbf{#1.}} %

\newenvironment{psmallmatrix}
  {\left(\begin{smallmatrix}}
  {\end{smallmatrix}\right)}

\usepackage[]{tikz}

\RequirePackage{scalerel}
\RequirePackage{tikz}
\usetikzlibrary{svg.path}

\definecolor{orcidlogocol}{HTML}{A6CE39}
\tikzset{
  orcidlogo/.pic={
    \fill[orcidlogocol] svg{M256,128c0,70.7-57.3,128-128,128C57.3,256,0,198.7,0,128C0,57.3,57.3,0,128,0C198.7,0,256,57.3,256,128z};
    \fill[white] svg{M86.3,186.2H70.9V79.1h15.4v48.4V186.2z}
                 svg{M108.9,79.1h41.6c39.6,0,57,28.3,57,53.6c0,27.5-21.5,53.6-56.8,53.6h-41.8V79.1z M124.3,172.4h24.5c34.9,0,42.9-26.5,42.9-39.7c0-21.5-13.7-39.7-43.7-39.7h-23.7V172.4z}
                 svg{M88.7,56.8c0,5.5-4.5,10.1-10.1,10.1c-5.6,0-10.1-4.6-10.1-10.1c0-5.6,4.5-10.1,10.1-10.1C84.2,46.7,88.7,51.3,88.7,56.8z};
  }
}

\RequirePackage{etoolbox}
\DeclareRobustCommand\orcidicon[1]{\href{https://orcid.org/#1}{\mbox{\scalerel*{
\begin{tikzpicture}[yscale=-1, transform shape]
    \pic{orcidlogo};
\end{tikzpicture}
}{|}}}}

\usepackage[capitalise,nameinlink,noabbrev]{cleveref} %

\newtheoremstyle{mystyle}%
  {}%
  {}%
  {\itshape}%
  {}%
  {\bfseries}%
  {.}%
  { }%
  {}%

\theoremstyle{mystyle}

\newtheoremrep{theorem}{Theorem}[section]       %
\newtheoremrep{proposition}[theorem]{Proposition}
\newtheoremrep{lemma}[theorem]{Lemma}
\newtheoremrep{example}[theorem]{Example}
\newtheoremrep{definition}[theorem]{Definition}
\newtheoremrep{corollary}[theorem]{Corollary}
\newtheoremrep{claim}[theorem]{Claim}
\newtheoremrep{remark}[theorem]{Remark}
\newtheoremrep{conjecture}[theorem]{Conjecture}

\crefname{claim}{Claim}{Claims}
\Crefname{claim}{Claim}{Claims}
\crefname{conjecture}{Conjecture}{Conjectures}
\Crefname{conjecture}{Conjecture}{Conjectures}

\AddToHook{env/proposition/begin}{\crefalias{theorem}{proposition}}
\AddToHook{env/lemma/begin}{\crefalias{theorem}{lemma}}
\AddToHook{env/example/begin}{\crefalias{theorem}{example}}
\AddToHook{env/definition/begin}{\crefalias{theorem}{definition}}
\AddToHook{env/corollary/begin}{\crefalias{theorem}{corollary}}
\AddToHook{env/claim/begin}{\crefalias{theorem}{claim}}
\AddToHook{env/remark/begin}{\crefalias{theorem}{remark}}
\AddToHook{env/conjecture/begin}{\crefalias{theorem}{conjecture}}

\usepackage{crossreftools}
\hypersetup{%
    bookmarksnumbered, bookmarksopen=true, bookmarksopenlevel=1,%
}

\newcommand{\myparagraph}[1]{\paragraph{\textbf{#1}}}

\newcommand{\reversal}[1]{#1^{\mathrm{R}}}

\newcommand{\start}{\mathrm{start}}
\newcommand{\eend}{\mathrm{end}}

\AtEndPreamble{%
    \hypersetup{colorlinks,
      linkcolor=purple,
      citecolor=blue,
      urlcolor=ACMDarkBlue,
      filecolor=ACMDarkBlue,
      pdfcreator = {},
      }}

\newcommand{\ii}{\textup{in}}
\newcommand{\oo}{\textup{out}}

\newcommand{\LL}{\calL}

\newcommand*\circled[1]{\tikz[baseline=(char.base)]{
            \node[shape=circle,fill=black,draw,text=white, inner sep=0.7pt] (char) 
			{\textup{#1}};}
			}

\setcopyright{acmlicensed}
\copyrightyear{2018}
\acmYear{2018}
\acmDOI{XXXXXXX.XXXXXXX}

\title{Resilience for Regular Path Queries: Towards a Complexity
Classification}

\begin{document}

\author{Antoine Amarilli}
\orcid{0000-0002-7977-4441}
\affiliation{%
    \orcidicon{0000-0002-7977-4441}
    \institution{Univ. Lille, Inria, CNRS, Centrale Lille, UMR 9189
    CRIStAL}\city{F-59000 Lille}\country{France}}
\email{antoine.a.amarilli@inria.fr}

\author{Wolfgang	Gatterbauer}
\orcid{0000-0002-9614-0504}
\affiliation{%
    \orcidicon{0000-0002-9614-0504}
	Northeastern University\country{USA}}
\email{w.gatterbauer@northeastern.edu}

\author{Neha Makhija}
\authornote{Part of this research was done while the author was
  at Northeastern University.}
\orcid{0000-0003-0221-6836}
\affiliation{%
    \orcidicon{0000-0003-0221-6836}
	University of Massachusetts, Amherst\country{USA}
  }
\email{nehamakhija@umass.edu}

\author{Mikaël Monet}
\orcid{0000-0002-6158-4607}
\affiliation{%
\orcidicon{0000-0002-6158-4607}
  \institution{Univ. Lille, Inria, CNRS, Centrale Lille, UMR 9189 CRIStAL, F-59000 Lille}\country{France}}
\email{mikael.monet@inria.fr}

\author{Martín Muñoz}
\orcid{0009-0003-3294-6159}
\affiliation{%
\orcidicon{0009-0003-3294-6159}
  \institution{CNRS, Centre de Recherche en Informatique de Lens, UMR 8188}\country{France};
  \institution{Millennium Institute for Foundational Research on Data (IMFD), Santiago}\country{Chile}}
\email{munoz@cril.fr}

\renewcommand{\appendixprelim}{}

\begin{abstract}
  The \emph{resilience problem} for a query and an input set
  or bag database is to compute the minimum number of facts to remove from the
  database to make the query false.
  In this paper, we study how to compute the resilience of 
  Regular Path Queries (RPQs) over graph databases.
  Our goal is to characterize the
  regular languages $L$ for which it is tractable to compute the resilience of
  the existentially-quantified RPQ built from~$L$.
    
    We show that computing the resilience in this sense is tractable (even in combined complexity) for all RPQs defined
    from so-called \emph{local languages}. 
    By contrast, we show hardness in data complexity for RPQs defined from the following language classes
    (after reducing the languages to eliminate redundant words):
    all finite languages featuring a word containing a repeated
    letter,  and
    all languages featuring a specific kind of counterexample to being
    local (which we call
    \emph{four-legged} languages). The latter include in particular all
    languages that are
    not \emph{star-free}. Our results also imply a dichotomy on all languages
    with a so-called
    \emph{neutral letter}.
    We last show tractability for some classes of non-local languages, namely
    the so-called \emph{bipartite chain languages} and \emph{one-dangling
    languages}, and 
    highlight some remaining obstacles towards a full
    dichotomy.
\end{abstract}

\maketitle

\begin{toappendix}
  \crefalias{section}{appendix}
  \crefalias{subsection}{appendix}
\end{toappendix}

\begin{figure}[t]
  {
\crefname{theorem}{\textbf{Thm.}}{\textbf{Thms.}}
\crefname{proposition}{\textbf{Prp.}}{\textbf{Prps.}}
\crefname{lemma}{\textbf{Lem.}}{\textbf{Lems.}}
    \centering
    \begin{tikzpicture}[xscale=1.05]
    \tikzstyle{every node}=[font=\footnotesize]
      \fill[fill=blue!10!white] (-6.3, .3) rectangle (-.6,-5.5);
      \fill[fill=yellow!20!white] (-.6, .3) rectangle (1,-5.5);
      \fill[fill=orange!20!white] (1, .3) rectangle (6,-5.5);

      \node (l1) at (-3.5, -.1) {PTIME};
      \node[text width=2cm,align=center] (l2) at (.2, -.1)
      {Unclassified\\[-.2em]languages};
      \node (l3) at (3.5, -.1) {NP-hard};

      \draw (-6.5, -.5) -- (6.25, -.5);
      \draw (-6.5, -3.5) -- (6.25, -3.5);

      \node[rotate=90] (l1) at (-6.5, -2) {Infinite};
      \node[rotate=90] (l2) at (-6.5, -4.5) {Finite};

      \node(u4) at (.2, -4) {$abc|bcd$};
      \node(u5) at (.2, -5) {$abc|bef$};

      \draw[rounded corners=2pt] (-6.2,-.6) rectangle (-2.5,-4.1); %
      \draw[rounded corners=2pt] (-5.0,-3.6) rectangle (-2.6,-5.4); %
      
      \draw[rounded corners=2pt] (1.1,-.6) rectangle (5.9,-4.1); %
      \draw[rounded corners=2pt] (1.2,-3.6) rectangle (3.5,-5.4); %

      \draw[rounded corners=2pt] (-2.4,-2.8) -- (-2.4,-4.2) -- (-3.6,-4.2) -- (-3.6,-4.6) -- (-2.4,-4.6) -- 
      (-2.4,-5.4) -- (-0.75,-5.4) -- (-0.75,-2.8) -- cycle; %

      \node (abcabd) at (-5.6, -3.85) {$abc|abd$};
      \node (abadcd) at (-3.8, -3.85) {$ab|ad|cd$};
      
      \node (aaaa) at (2.4, -3.85) {$aaaa$};
      \node (axbcxd) at (4.8, -3.85) {$axb|cxd$};

      \node (axxb) at (-4.35, -2) {$ax^*b$};

      \node (local) at (-4.3, -1) {Local languages (\cref{prp:ro-ptime})};
      \node (fourlegged) at (3.5, -1) {Four-legged languages (\cref{prp:four-legged})};

      \node (axxbcxd) at (3.5, -1.75) {$ax^*b|cxd$};
      \draw[rounded corners=2pt] (1.5,-2.25) rectangle (5.5,-3.25); %
      \node (nonstarfree) at (3.5, -2.5) {Non-star-free (\cref{lemma:not aperiodic are hard})};
      \node (badbad) at (3.5, -3) {$b (aa)^* d$};
      
      \node (abbc) at (-3.1, -4.4) {$ab|bc$};
      \node (axbbyc) at (-4.35, -4.4) {$axb|byc$};
      \node[align=center] (bcl) at (-3.8, -5.1) {Bipartite Chain\\[-.2em] Languages
      (\cref{prp:rpn})};

      \node (u1) at (-1.6, -3.2) {$ax^*b | xd$};
      \node(abcbe) at (-1.6, -3.7) {$abc|be$};
      \node(abcdce) at (-1.6, -4) {$abcd|ce$};
      \node (u2) at (-1.6, -4.3) {$abcd|be$};
      \node[text width=2cm, align=center] (submod) at (-1.6, -4.95)
      {One-dangling\\[-.2em]languages\\[-.2em](\cref{prp:submod})};

      \node(abbcca) at (4.8, -4.35) {$ab|bc|ca$ (\cref{prp:abbcca})};
      \node(abcdab2) at (4.8, -4.7) {$abcd|be|ef$};
      \node(abcdab3) at (4.8, -4.95) {$abcd|bef$};
      \node(abcdab4) at (4.8, -5.25) {(\cref{prp:otherhard})};

      \node (aa) at (2.4, -4.3) {$aa$};
      \node (abcacab) at (2.4, -4.55) {$abca|cab$};
      \node[align=center,text width=3cm] (bcl) at (2.4, -5.1) {Finite,
      repeated\\[-.1em] letter
      (\cref{thm:repeated-letter})};

            \node (inf1) at (.2, -1.5) {$a b^*c | ba$};
            \node (inf2) at (.2, -2.5) {$a b^*d | a c^* d | bc$};
    \end{tikzpicture}
  }

    \caption{Overview of our results on the complexity of resilience (all
    results apply both to bag and set semantics), with some examples of
    languages. All languages are infix-free. For legibility, we omit
    \cref{prp:dicho}.
    }
    \label{fig:results}
\end{figure}

\section{Introduction}

This paper studies the \emph{resilience problem}, which asks, for a
Boolean query~$Q$ and a given database~$D$ (in set or bag semantics),
for a minimum-size set of facts to remove from~$D$ so
that~$Q$ is not satisfied. In bag semantics, facts additionally
carry multiplicities, which we view as deletion costs.\footnote{If~$Q$ does not hold on~$D$, then resilience is zero. If~$Q$ holds
on~$D$ and on all subinstances of~$D$, then resilience is infinite.}
The notion of resilience is motivated by the need to identify how
``robust'' a query answer is, especially when the facts of the input
database may be incorrect. Specifically, resilience quantifies how many
facts of~$D$ must be removed before~$Q$ no longer holds, so that higher
resilience values mean that~$Q$ is less sensitive to facts of~$D$ being
removed.
Resilience is also closely connected to database causality and
responsibility measures~\cite{MeliouGHKMS10,DBLP:journals/pvldb/FreireGIM15,DBLP:journals/pvldb/MeliouGMS11},
and more generally to reverse data-management tasks~\cite{DBLP:journals/pvldb/MeliouGS11}, including
view-update and deletion-propagation problems~\cite{buneman2002propagation} with controlled
source side effects~\cite{DBLP:journals/tods/DayalB82} or view side effects~\cite{DBLP:conf/pods/Kimelfeld12}.
It is also related to finding smallest witnesses or
small database examples for explaining wrong queries~\cite{DBLP:conf/icdt/0005S24,DBLP:conf/sigmod/MiaoRY19}.
See~\cite{DBLP:journals/pvldb/MakhijaG25} for a recent unified perspective on several such
explanation tasks as minimum contrastive or abductive explanations.

The resilience problem was introduced for Boolean conjunctive queries (CQs)
and has mostly been studied for such queries
\cite{DBLP:journals/pvldb/FreireGIM15,DBLP:conf/pods/FreireGIM20,bodirsky2024complexity,makhija2024unified}.
The key challenge in this line of work is to understand its \emph{data complexity},
i.e., its complexity for a fixed query~$Q$ as a function of the input database~$D$.
Whenever query evaluation for~$Q$ is in polynomial data complexity, the associated
resilience decision problem is in NP.\footnote{Here and throughout, 
NP-hardness and NP-completeness statements refer to the threshold decision problem: given~$D$ and an integer~$k$, decide whether the resilience is at most~$k$. 
Our tractability results compute the resilience value itself along with a
witness. 
The relation between the decision, value, and witness variants is discussed formally in \cref{sec:prelims}.}
However, the exact complexity landscape is subtle:
resilience is NP-hard for some simple queries, yet tractable for other nontrivial ones
\cite{DBLP:conf/pods/FreireGIM20,makhija2024unified}. 
This motivates the goal of
\emph{characterizing the queries for which resilience can be tractably computed},
ideally via a dichotomy that separates the PTIME cases from the NP-complete ones.
For self-join-free CQs, this boundary is well-understood
\cite{DBLP:journals/pvldb/FreireGIM15}. For general CQs with self-joins,
it remains elusive. 
For self-join-free CQs, dichotomies are known under set semantics~\cite{DBLP:journals/pvldb/FreireGIM15} and bag semantics~\cite{makhija2024unified},
with different tractability boundaries. 
Makhija and Gatterbauer~\cite{makhija2024unified} give a unified ILP formulation, establish cases where its LP relaxation is exact, and develop automated IJP-based hardness certificates. 
A complete classification for arbitrary CQs with self-joins remains open.
We note, however, that for bag semantics,
recent work~\cite{bodirsky_et_al:LIPIcs.STACS.2026.15}
establishes 
a PTIME/NP-complete dichotomy for UCQs over a single binary relation. 
This does not establish the corresponding dichotomy under set semantics.

In this paper, we study resilience from a different perspective: we focus
on graph databases (i.e., databases on an arity-two signature), and on
queries defined by regular languages, namely Boolean Regular Path Queries
(Boolean RPQs), with no constants and with existentially quantified
endpoints. In other words, the query asks whether the database contains a
walk whose edge labels form a word of the regular language, and
resilience asks how many database facts must be removed so that no such
walk remains.

The relationship between RPQs and CQs is subtle.
On finite languages, Boolean RPQs correspond to unions of path-shaped CQs, i.e., to a restricted class of UCQs over binary relations. 
Hence, our results for finite languages can be read as results about path-like UCQs. 
RPQs and UCQs over binary relations are incomparable in expressive power: RPQs can express unbounded reachability (via the Kleene star), 
whereas UCQs can impose branching and cyclic join patterns not expressible by a single RPQ. 
Finite-language RPQs correspond to unions of directed-path CQs.

To our knowledge, the only other work which studies resilience for RPQs
is the recent paper of Bodirsky et al.~\cite{bodirsky2024complexity, bodirsky2025complexity}, in
which the authors use CSP techniques. In particular, they obtain 
a dichotomy on the
resilience problem for Boolean RPQs (and also for the broader class of
so-called \emph{two-way RPQs}, aka 2RPQs). However, their lower bounds
do not apply to databases in set semantics. In addition, their dichotomy
criterion, while decidable, does not give an intuitive
language-theoretic characterization of the tractable cases. We discuss
this in more detail at the end of \cref{sec:prelims}.

Our study of resilience for RPQs is motivated by the above stated database applications,
but it is also motivated by connections to classical network flow
problems. Consider for instance the classical MinCut problem of finding
the minimum-cost cut in a flow network,
which can be solved in polynomial time,
and where we can allow
multiple sources and sinks without loss of generality.
It is easily seen 
that MinCut is
equivalent to the resilience problem in bag semantics for the RPQ
$a x^* b$ on the alphabet $\{a,b,x\}$, with $a$-facts of the database
corresponding to sources, $b$-facts to sinks, and $x$-facts to edges of
the network. Hence, the resilience problem for RPQs corresponds to a
generalization of MinCut, where the edges of the flow network are labeled
and we must cut edges to ensure that no walk in a specific language
remains. In the other direction, for many RPQs, we can devise tractable
algorithms for resilience by reducing to a MinCut problem on a suitably
constructed flow network. However, resilience is NP-hard for some RPQs,
e.g., the finite RPQ $aa$~\cite{DBLP:conf/pods/FreireGIM20}, so these
methods do not always apply. Our results aim at determining which RPQs
have tractable resilience (in particular via such flow-based methods),
and which RPQs have intractable resilience.

\myparagraph{Contributions and outline.}
After giving preliminary definitions and stating the problem in
\cref{sec:prelims}, we start in \Cref{sec:epsro} by giving our main tractability
result (\cref{prp:ro-ptime}): resilience is tractable for RPQs defined from so-called \emph{local
languages}~\cite{mpri,sakarovitch2009elements}. Those languages are intuitively
those that are defined by specifying the possible initial labels of a path,
the possible final
labels of a path, and the consecutive pairs of labels that are allowed.
Examples of local languages are $ax^*c$ or $ab|ad|cd$.
We then show that, for RPQs defined with a local language~$L$, we can compute
resilience by reducing  to the MinCut problem via a
natural product construction. The argument applies to queries both in set and bag semantics,
and it yields an algorithm which is also tractable in combined
complexity. 
  Using the recent almost-linear-time deterministic algorithm for minimum-cost
  flow \cite[Theorem~1.1]{van2023deterministic}, our algorithm
runs in time
$(|A|\times |\Sigma| \times |D|)^{1+o(1)}$ where $A$ is an automaton
representing $L$, where $\Sigma$ is the signature, 
and where $D$ is the input
database;
in other words, the algorithm runs in time
$O((|A|\times |\Sigma| \times |D|)^{1+\epsilon})$
for all $\epsilon>0$.

We then turn to intractability results in
\cref{sec:hard,sec:fourlegged,sec:repeated}.
In all these results, we assume that we focus on the \emph{infix-free}
sublanguage where redundant words have been eliminated, e.g., resilience for $abbc|bb$ and $bb$ is the same
because the corresponding (existentially-quantified) Boolean RPQs are
equivalent.
Note that the infix-free sublanguage of a regular language is easily
seen to also be a regular language.

We define in \cref{sec:hard} the
notion of a \emph{hardness gadget} for a language~$L$, namely, a specific
database to be used as a building block in a reduction.
Our gadgets adapt to RPQs the independent-join-path approach introduced in \cite{DBLP:conf/pods/FreireGIM20}
and developed into a semantic hardness criterion in \cite{makhija2024unified}. 
Our endpoint restriction ensures non-leaking composition, while condensation to an odd path supports the reduction from vertex cover.
When an RPQ~$Q$ has a hardness gadget~$\Gamma$, 
we can reduce the minimum vertex cover problem on
an undirected graph~$G$ to a resilience computation (in set semantics)
for~$Q$ on a
database built using~$\Gamma$, by replacing the edges of~$G$ by a copy
of~$\Gamma$. We illustrate in \cref{sec:hard} how our
gadgets
can be used to show hardness for two specific RPQs: $axb|cxd$
(\cref{prp:axbcxd}), and $aa$ (\cref{prp:aa}).

We then generalize 
these two hardness results in turn to apply to a
broader class of languages. First, in \cref{sec:fourlegged}, we introduce the notion of
a \emph{four-legged language}, as a language featuring a specific kind of
counterexample to being local. More precisely, 
the local languages of \cref{sec:epsro} have the following
known equivalent characterization~\cite{sakarovitch2009elements}: for
any letter $x$ and words $\alpha, \beta, \gamma, \delta$, if a language $L$
contains the words
$\alpha x \beta$ and $\gamma x \delta$ but does not contain the
``cross-product'' word $\alpha x \delta$, then $L$ is not local. A four-legged
language is a non-local language that features a counterexample of this form in
which all the words $\alpha,
\beta, \gamma, \delta$ are required to be non-empty.
This covers, for instance, $axb|cxd$,
but not $ab|bc$.
We show (\cref{prp:four-legged}) that resilience is intractable for any
four-legged language. This result implies hardness for many languages: it
covers all \emph{non-star-free} languages (\cref{lemma:not aperiodic are hard}), a well-known class of infinite
regular languages; and, along with the tractability of local languages, it implies 
a dichotomy for languages having a 
\emph{neutral letter} (\cref{prp:dicho}).

Second, in \cref{sec:repeated}, we show that resilience is hard for any finite
infix-free
language containing a word with a repeated letter (\cref{thm:repeated-letter}).
This generalizes the case of $aa$, and shows that self-joins always imply hardness 
in
the following sense: finite RPQs (those equivalent to UCQs), once made
infix-free, have intractable
resilience as soon as one of their constituent CQs features a self-join.
This contrasts with the setting of CQs, where characterizing the tractable
queries with self-joins is still open~\cite{arxiv:FreireGIM20}.

Note that our two hardness results (\cref{prp:four-legged} and
\cref{thm:repeated-letter}) are incomparable:
languages with a repeated letter are non-local, but they are not all 
four-legged (e.g., $aa$); conversely, some four-legged
languages have no words with a repeated letter (e.g., $axb|cxd$). Assuming finiteness
in \cref{thm:repeated-letter} is also crucial, as $a x^* b$ features
words with repeated letters and still enjoys tractable resilience. The proof of
\cref{thm:repeated-letter} is our main technical achievement: it
relies on \cref{prp:four-legged} and uses different
gadgets
depending on which additional words the language contains.

\Cref{sec:other} then explores how our results do not classify all RPQs. We
mostly focus
on finite RPQs: the remaining open cases are then the RPQs which have no words with
repeated letters, are not local, but are not four-legged. We introduce the class
of chain languages (e.g., $axb|byc$) to cover some of these RPQs,
and show tractability (and reducibility to MinCut) assuming
a bipartiteness condition. We then introduce the class of one-dangling languages
(e.g., $abc|be$), for which we show tractability (via a more elaborate
reducibility to MinCut).
We finish the section by showing that some other unclassified RPQs are intractable.
We conclude the paper in \cref{sec:conclusion}.

\introparagraph{Conference version} 
This paper is an extended version of an earlier conference 
paper~\cite{amarilli2025resilience}.  %
  Compared to the conference version, \emph{all our tractability results in the current paper use a unified reduction to flow}. Some of the new proofs are considerably simpler, 
and we currently do not have any example of a query with tractable resilience for which
tractability does not follow by reducing to an appropriate network flow problem.
Our current results support our hypothesis that flow is all one needs for tractable cases, even those that are currently open.
In more detail, we
have improved \cref{prp:submod} (compared to
\cite[Proposition~7.7]{amarilli2025resilience}         %
by showing a tractability result that
covers a larger family of languages (dubbed \emph{one-dangling languages})
and uses a reduction to a flow
problem. This yields tractability in combined complexity and gives a better exponent than an earlier proof via submodular optimization
from~\cite{amarilli2025resilience}.         %
An overview of our results in the
present paper is in \cref{fig:results}. 
Compared with
\cite[Figure~1]{amarilli2025resilience},          %
the languages $abcd|be$ and $a x^* b |
xd$ have now been classified as tractable thanks to the new \cref{prp:submod}, 
and new unclassified languages 
$a b^*c | ba$ and $a b^*d | a c^* d | bc$
were added.
The present paper also includes the complete proofs for all results, which had
been omitted from the conference article.

\section{Preliminaries and Problem Statement}
\label{sec:prelims}
\myparagraph{Formal languages.}
Let $\Sigma$ be an \emph{alphabet}, consisting of \emph{letters}. A
\emph{word} over $\Sigma$ is a finite
sequence of
elements of $\Sigma$. We write $\epsilon$ for the empty word. 
We write $\Sigma^*$ for the set of words over $\Sigma$, and write $\Sigma^+ \coloneq 
\Sigma^* \setminus \{\epsilon\}$.
We use lowercase
Latin letters such as $a,b,x,y$ to denote letters in $\Sigma$, and
lowercase Greek letters like $\alpha, \beta, \gamma, \tau$ to
denote words in~$\Sigma^*$. The \emph{concatenation} of two words $\alpha,\beta$
is written
$\alpha \beta$. We write $|\alpha|$ for the
length of a word. A word~$\alpha$ is an \emph{infix} of another
word~$\beta$ if $\beta$ is of the form $\delta \alpha \gamma$, and
it is a \emph{strict infix} if moreover $\delta \gamma \neq
\epsilon$. If $\beta = \alpha \gamma$ then $\alpha$ is a
\emph{prefix} of~$\beta$, and a \emph{strict prefix} if
$\gamma\neq\epsilon$. Likewise if $\beta = \delta \alpha$ then
$\alpha$ is a \emph{suffix} of~$\beta$, and a \emph{strict suffix}
if $\delta \neq \epsilon$.

A \emph{language}
$L$ is a (possibly infinite) subset of~$\Sigma^*$. For $L, L'$ two
languages, we write their \emph{concatenation} as $LL' \coloneq  \{\alpha \beta  \mid
\alpha \in L, \beta \in L'\}$.
We mainly focus in this work on
\emph{regular languages}, also known as \emph{rational languages}, which can be represented by \emph{finite state
automata}. We use the following automaton formalisms, which recognize precisely the
regular languages~\cite{sakarovitch2009elements}: \emph{deterministic finite automata} (DFA),
\emph{nondeterministic finite automata} (NFA), and \emph{NFAs with
$\epsilon$-transitions} ($\epsilon$NFA, also called \emph{automata with
spontaneous transitions}).

An $\epsilon$NFA $A = (S,I,F,\Delta)$ 
is a finite set $S$
of \emph{states}, a set $I\subseteq S$ of \emph{initial states}, a
set $F\subseteq S$ of \emph{final states}, and a \emph{transition
relation} $\Delta \subseteq S \times (\Sigma \cup \{\epsilon\})
\times S$. Let $\alpha = a_1\cdots a_m$ be a word. A
\emph{configuration} for~$A$ on~$\alpha$ is an ordered pair $(s,
i)$ for $s \in S$ and $0 \leq i \leq m$. 
A \emph{run} of~$A$
on~$\alpha$ is a sequence of configurations $(s_0, i_0), \ldots,
(s_\ell, i_\ell)$ with $\ell \geq m$
such that $s_0 \in I$ is initial, $i_0 =
0$, $i_\ell = m$, and for each $0 \leq j < \ell$ we have either
$i_{j+1} = i_j$ and $(s_j, \epsilon, s_{j+1}) \in \Delta$ (an
$\epsilon$-transition) or $i_{j+1} = i_j+1$ and $(s_j, a_{i_j+1},
s_{j+1}) \in \Delta$ (a letter transition).
The run
is \emph{accepting} if $s_\ell \in F$. 
The \emph{language of $A$},
denoted $\LL(A)$,
is the set of words having an accepting run: we say that  $A$
\emph{recognizes}
$\LL(A)$.
An example $\epsilon$NFA $A_3$ is depicted in \cref{fig:1c}, and an example
accepting run of~$A_3$ on
the word $ad$ is
  $(s_1, 0), (s_2, 1), (s_4, 1), (s_5, 2)$.
The \emph{size}~$|A|$ of the $\epsilon$NFA $A=(S,I,F,\Delta)$ 
is $|S| + |\Delta|$,
its total number of states and transitions. 

An NFA is an $\epsilon$NFA which does not feature $\epsilon$-transitions, and a
DFA is an NFA having exactly one initial state and  having
at most one transition $(s,a,s')$ in $\Delta$ for each
$s\in S$ and $a \in \Sigma$.

We also use \emph{regular expressions} to denote regular languages:
we omit their formal definition,
see~\cite{sakarovitch2009elements}. 
We write $|$ for the union
operator and $*$ for the Kleene star operator.

\myparagraph{Databases, queries, resilience.}
A \emph{graph database} on the alphabet $\Sigma$, or simply \emph{database}, is
a set of $\Sigma$-labeled \emph{edges} of the form $D\subseteq
V\times \Sigma \times V$, where $V$ is called the \emph{domain} and its elements are called the \emph{nodes} of $D$.
Notice that a graph database can equivalently be seen as a relational database
on an arity-two signature that directly corresponds to the alphabet $\Sigma$.
We also call \emph{facts} the edges $(v,a,v')$ of $D$, and write them as $v \xrightarrow{a} v'$. 
We call $v$ the \emph{tail} of $v \xrightarrow{a} v'$ and
$v'$ its \emph{head}.
A \emph{Boolean query} is a function~$\query$ that takes as input a
database~$D$ and outputs~$\query(D) \in \{0,1\}$. We write
$D\models \query$, and say that $\query$ is \emph{satisfied by
$D$}, when $\query(D)=1$, and write $D\not\models \query$
otherwise. 
In this paper we study the \emph{resilience} of queries:

\begin{definition}[Resilience]
  \label{def:resilience}
  Let $D$ be a database and $\query$ a Boolean query. The
  \emph{resilience of $\query$ on~$D$}, denoted
  $\resset(\query,D)$, is defined\footnote{By convention, if
  $\query$ is satisfied by all subsets of $D$, then the
minimum is defined to be $+\infty$.} as $\resset(\query,D) \coloneq
  \min_{\substack{D'\subseteq D\\\text{s.t. }D \setminus
  D'\not\models \query}} |D'|$.
\end{definition}

In other words, the resilience is the minimum number of facts to
remove from $D$ so that $\query$ is not satisfied. A subset of
facts $D'\subseteq D$ such that $D\setminus D' \not\models \query$
is called a \emph{contingency set of~$D$
for~$\query$}.
Notice that 
if $D\not\models \query$, then
$\emptyset$ is a contingency set and
we have\footnote{We deviate from the original definition
in~\cite{DBLP:journals/pvldb/FreireGIM15} which left resilience
undefined for the case of $D\not\models \query$.}
$\resset(\query,D)=0$.

\myparagraph{Regular Path Queries.}
From a language $L$ over $\Sigma$ we now define a query over graph
databases that tests the existence of walks labeled by words of $L$.
Given a database $D$, a \emph{walk} 
in~$D$ is a (possibly empty) sequence of consecutive
edges
$v_1 \xrightarrow{a_1} v_2$, 
$v_2 \xrightarrow{a_2} v_3$, $\ldots$, $v_k \xrightarrow{a_k} v_{k+1}$
in~$D$. Note that the $a_i$ and~$v_i$ are not
necessarily pairwise distinct.
We say that the walk is \emph{labeled} by the word $\alpha = a_1 \cdots
a_k$, and call it an \emph{$\alpha$-walk}. We also say that the walk
\emph{contains} a word $\beta \in \Sigma^*$ if $\beta$ 
is an infix of~$\alpha$.
An \emph{$L$-walk} in~$D$ is then an
$\alpha$-walk for some $\alpha \in L$.
We define the
Boolean query $\query_L$ as: $\query_L(D) = 1$ iff
there is some $L$-walk in~$D$. In particular, if $\epsilon\in L$
then $\query_L(D) = 1$ on every 
database~$D$ (including  $D = \emptyset$).
A \emph{regular path query} (RPQ)
is then a Boolean query of the form
$\query_L$ for $L$ a regular language.
Note that, in all of the present paper, we evaluate RPQs with walk semantics
(i.e., arbitrary paths with repetitions), as opposed to other works such
as~\cite{martens2019dichotomies} which evaluate them with simple path or trail
semantics: we leave open the question of studying resilience for RPQs in simple
path or trail semantics.

Observe that $\query_L$ does not change if we remove from $L$ all
words $\beta$ such that some strict infix~$\alpha$ of $\beta$ is
also in~$L$: indeed if $D$ contains a $\beta$-walk then it also contains
an $\alpha$-walk. 
Formally, the \emph{infix-free sublanguage
of~$L$} is $\red(L) \coloneq  \{\alpha \in L \mid \text{no strict infix
of } \alpha \text{ is in }L\}$. We then have that~$\query_L$
and~$\query_{\red(L)}$ denote the same
query. The infix-free sublanguage
of $L$ can also
be called the \emph{factor-free sublanguage of $L$} in
\cite{barcelo2019boundedness}, and is called the \emph{reduced sublanguage} in
our earlier conference
paper~\cite{amarilli2025resilience}.  %
It is easy to show that if $L$ is regular then so is $\red(L)$. We say that~$L$
is \emph{infix-free} if
$L = \red(L)$; this is also called being an \emph{infix
code}~\cite{ito1991outfix}.

\myparagraph{Bag semantics.}
We also consider the resilience problems on \emph{bag graph databases},
as opposed to \emph{set graph databases} which are the databases
defined so far. A \emph{bag (graph) database} 
consists of a
database $D$ together with a \emph{multiplicity function}
$\mult:D\to \mathbb{N}_{>0}$; it can equivalently be seen as a multiset of
facts.
In the input, the multiplicities are expressed as integers written in binary.
A Boolean query $Q$ is satisfied by a bag database $(D, \mult)$ iff $D$
satisfies~$Q$, i.e., the multiplicities of facts are not seen by queries.
A contingency set of~$D$ for $\query$ is 
  then again a subset $D' \subseteq D$ such that $D \setminus D' \not\models Q$, and
  the \emph{resilience of $\query$ on~$D$} is: 
$\resbag(\query,D) \coloneq \min_{\substack{D'\subseteq D\\\text{s.t. }D
\setminus D'\not\models \query}} \sum_{e \in D'} \mult(e)$, i.e.,
fact multiplicities represent costs.

\myparagraph{Networks and cuts.}
Our resilience algorithms work by reducing to the problem of finding a
\emph{minimum cut of a flow network}~\cite{Cormen:2009dz}. 
Recall that
a \emph{flow network} 
consists of a directed graph $\calN =
(V,t_{\text{source}},t_{\text{target}},E,c)$ 
with 
$V$ being the vertices, $t_{\text{source}}, t_{\text{target}}\in V$ being special vertices 
called the \emph{source} and the \emph{target} respectively,
$E\subseteq V\times V$, and $c:E \to
\mathbb{R}_+ \cup \{+\infty\}$ giving the \emph{capacity} of each
edge. For $E'\subseteq E$, we define the \emph{cost} of~$E'$ as $c(E') \coloneq 
\sum_{e\in E'} c(e)$,
and denote $\calN \setminus E'$ as the network in which all edges from $E'$ are removed. 
A \emph{cut} is a set of edges
$E'\subseteq E$ such that there is no directed path from $t_{\text{source}}$ to
$t_{\text{target}}$
in $\calN \setminus E'$. The \emph{min-cut problem}, written MinCut, takes as input a
flow network and returns the minimum cost of a cut. 
From the \emph{max-flow min-cut theorem}~\cite{ford1956maximal} and Menger's
theorem~\cite{menger:1927}, we know that the min-cut problem
can be solved in PTIME.
More precisely, 
practical algorithms can run in
$O(|V|^2 \sqrt{|E|})$~\cite{ahuja1997computational},
and recent theoretical algorithms
run in
$|\calN|^{1+o(1)}$~\cite{van2023deterministic}, (i.e., in $O(|\calN|^{1+\epsilon})$ for all $\epsilon > 0$) where $|\calN|
\coloneq
|V| + |E|$.

For all our upper bounds, we work in the RAM
  model on words of logarithmic size and where arithmetic operations take
  constant time~\cite{aho1974design}. In particular, memory cells can only represent integers of
  value polynomial in the input length; this implies that the fact
multiplicities for bag databases are assumed to be bounded by a polynomial and that
capacities in all flow networks that we will construct will be polynomial as
well.

\myparagraph{Problem statement.}
In this paper, we study the \emph{resilience problem} for 
queries~$\query_L$ defined by languages~$L$; most of our results are for 
regular languages.
The \emph{resilience problem} for~$L$, written 
$\resprobset(L)$, is
the following:
given as input a
database $D$ and $k\in \mathbb{N}$, decide if $\resset(\query_L,D)
\leq k$. 
The problem $\resprobbag(L)$ in bag semantics is defined analogously.
  Note that $\resprobset(L)$ 
  easily reduces to
  $\resprobbag(L)$ %
  by giving multiplicity of 1 to each fact; 
  however, in the setting of CQs, some conjunctive queries $\query_0$ are known for which
  $\resprobset(\query_0)$ is in PTIME, but $\resprobbag(\query_0)$ is NP-complete
  \cite[Table~1]{makhija2024unified} (even when the multiplicities of the input
  bag database are written in unary).

We also note that the definition given here corresponds to a \emph{decision
problem} of determining whether the resilience is no greater than the value~$k$.
However, thanks to self-reducibility~\cite[Section~3.3]{Goldreich2010PNP}, the \emph{function
problem} of computing a witness contingency set reduces in polynomial time to
this decision problem: see also \cite[Remark 8.4]{bodirsky2025complexity} for a
self-contained argument. The same holds for the problem of computing the
value of resilience, i.e., the cardinality of such a witness; this can be
performed via self-reducibility, or simply by a binary search over the possible
values. Conversely, the decision problem immediately reduces to any of these two problems.

Our definitions correspond to the \emph{data complexity}
perspective, where we have one computational problem per
language~$L$
and where the complexity is measured as a function
of $|D|$ only. We mostly focus on data complexity, but most of our upper bounds also
apply in \emph{combined complexity}, i.e., when both the database $D$ and an
automaton $A$ representing $L$ 
are given as input.

It is easy to see that, for any regular language~$L$, the problem
$\resprobbag(L)$, hence $\resprobset(L)$, is in NP. Indeed, given a bag
database~$D$, it suffices to guess a subdatabase $D' \subseteq D$,
and verify in PTIME that~$\sum_{f\in D'} \mult(f)$ is small enough
and that~$D'$ is a contingency set, i.e., $D \setminus D'
\not\models \query$. The latter point, i.e., the nonexistence of an
$L$-walk in $D \setminus D'$, can easily be checked in PTIME by a
product construction followed by a reachability test (see, e.g.,
\cite[Lemma 3.1]{DBLP:journals/siamcomp/MendelzonW95}).

However, the NP upper bound is not always tight:
there are languages~$L$ (e.g.,
$L = a|b$)
such that $\resprobbag(L)$ is in PTIME. The goal of our
paper is to characterize the languages $L$ for which
$\resprobset(L)$ and $\resprobbag(L)$ are in PTIME, and give efficient
algorithms for them.

\myparagraph{Connection to VCSPs.}
The recent work
of Bodirsky et al.~\cite{bodirsky2024complexity} uses earlier results on Valued
Constraint Satisfaction Problems (VCSP) to obtain the following dichotomy on $\resprobbag$:

\begin{theorem}[{\cite{bodirsky2024complexity}, see also 
  \cite[Theorem 8.15]{bodirsky2025complexity}}]
  \label{thm:vcsp}
  Let $L$ be a regular language. The problem $\resprobbag(L)$ is either in PTIME or
  it is NP-complete.
\end{theorem}

\noindent
We point out that this result can be extended to so-called \emph{two-way RPQs}, and to the case where some relations are declared \emph{exogenous} (i.e., they cannot be part of a contingency set).
Hence, as discussed in the introduction, \cref{thm:vcsp} gives a complete dichotomy across all regular languages. However, it also has four important
shortcomings:
\circled{\normalsize{1}}~The
VCSP approach gives no explicit characterization of the class $L$ of tractable languages,
or efficient bounds on the complexity of testing whether a given~$L$ is PTIME or
not: the criterion in~\cite{bodirsky2024complexity}, using~\cite{kolmogorov2019testing}, 
can be computationally checked but takes time exponential in a representation of the so-called
\emph{dual} of~$L$.
\circled{\normalsize{2}}~In
tractable cases, the VCSP approach does not give good bounds on the
degree of the polynomial (unlike, e.g., explicit reductions 
to MinCut).
We remark that the first two difficulties 
also occur in works studying other problems where VCSP
implies a general dichotomy but without explicit characterizations or tight 
bounds, for instance in~\cite{hirai2021minimum}.
\circled{\normalsize{3}}~The 
VCSP approach does not shed light on the
\emph{combined complexity}, i.e., it does not identify cases that are
tractable when both the query and
the database are given as input.
\circled{\normalsize{4}}~The 
VCSP approach inherently only applies
to bag semantics and not to set semantics, so there is no known analogous
dichotomy for $\resprobset$.
Our results improve on some of these shortcomings.
Specifically, all our intractability results already apply to the $\resprobset$
problem, 
most of our upper bound results apply to combined
complexity (when the query is given as an automaton along with the database), and
our tractability results
feature a polynomial complexity with a low degree
(and in fact up to $n^{o(1)}$ factors the data complexity is always linear, except
for one result, \cref{prp:rpn}, where it is quadratic).
However, unlike 
\cite{bodirsky2024complexity}, our results do not establish a full dichotomy.

\section{Tractability for Local Languages}
\label{sec:epsro}

We now start the presentation of our results on the resilience
problem by giving our main tractability result in this section.
The result applies to 
\emph{local languages}\footnote{We note that local languages were also
recently observed to correspond to a tractable class of RPQs for which
probabilistic query evaluation is equivalent to source-to-target reliability,
see~\cite{amarilli2025approximating}.
For self-join-free chain CQs, whose label languages are singleton
local languages, an earlier fixed-endpoint correspondence with
source-to-target reliability on layered graphs appears
in~\cite[Proposition~18(a)]{DBLP:journals/vldb/GatterbauerS17}.},
a well-known subclass of the regular
languages~\cite{mpri,sakarovitch2009elements}, whose definition we recall
below.
We show that the resilience problem is tractable for local
languages, also in combined complexity, using a specific notion of automata
and a product construction to reduce resilience to MinCut.
We stress that local languages do not cover all cases of
regular languages with tractable resilience: we will show other
tractable cases in \cref{sec:other}. However, in \cref{sec:fourlegged} we will give
a partial converse: resilience is hard for the so-called \emph{four-legged}
languages, which have a specific form of counterexample to being local.

The remainder of this section is structured as follows: in
\cref{subsec:local-defs} we define what are local languages and provide
two equivalent characterizations that we will use later. We then show in
\cref{subsec:local-ptime} that resilience is PTIME for local
languages.

\subsection{Local Languages}
\label{subsec:local-defs}

We define here local languages and show some useful properties. Intuitively, a
language~$L$ is \emph{local} if membership to~$L$ is described by which
letter pairs can occur consecutively and which letters can start or can end
words~\cite[Chapter III, Section 5.1]{mpri}:

  \begin{definition}[Local language]
\label{def:localdef}
A language $L\subseteq \Sigma^*$ is \emph{local} if there exist two subsets
$P$ and $S$ of $\Sigma$ and a subset $N$ of $\Sigma^2$ such that $L \setminus
\{\epsilon\} = (P \Sigma^* \cap \Sigma^* S) \setminus (\Sigma^* N
\Sigma^*)$. In other words, a non-empty word $w$ is in $L$ if and only if its first letter is in $P$, its last letter is in $S$,
and every infix of $w$ of size two is in $\Sigma^2 \setminus N$.
\end{definition}

Local languages have many equivalent definitions, and we now give two of
them. The first one goes through \emph{local DFAs}:

\begin{definition}[Local DFA]
  \label{def:local}
  A DFA is \emph{local} if, for every letter $a \in \Sigma$, all $a$-transitions
  in the automaton have the same target state.
\end{definition}

Then the following equivalent characterization is known (and easy to check):

\begin{proposition}[\mbox{\cite[Proposition 5.22]{mpri}}]
  A language $L$ is \emph{local} if and only if there exists a local DFA that
  recognizes~$L$.
\end{proposition}

\begin{example}%
  \label{expl:locallang}
  The languages $ax^*b$ and $ab|ad|cd$ are local, with local DFAs
  in
  \cref{fig:1a} and \cref{fig:1b}.
\end{example}

\begin{figure}
  \begin{subfigure}[t]{0.3\columnwidth}
    \begin{tikzpicture}[shorten >=1pt,node distance=1.3cm,on grid,auto,thick]

  \node[state,initial, initial text=]  (s_1)                      {$s_1$};
  \node[state]          (s_2) [right=of s_1] {$s_2$};
  \node[state,accepting](s_3) [right=of s_2] {$s_3$};

  \path[->] (s_1) edge              node        {$a$} (s_2)
            (s_2) edge [loop above] node       {$x$} ()
            (s_2) edge              node {$b$} (s_3);
\end{tikzpicture}
    \caption{Local DFA $A_1$ for 
    $ax^*b$} \label{fig:1a}
  \end{subfigure}%
  \hspace*{\fill}   %
  \begin{subfigure}[t]{0.3\columnwidth}
    \begin{tikzpicture}[shorten >=1pt,node distance=1.2cm and 1.3cm,on grid,auto,thick]

  \node[state,initial, initial text=]  (s_1)                      {$s_1$};
  \node[state]          (s_2) [right=of s_1] {$s_2$};
  \node[state,accepting]          (s_4) [right=of s_2] {$s_3$};
  \node[state]          (s_6) [below=of s_2] {$s_4$};
  \node[state,accepting]          (s_8) [right=of s_6] {$s_5$};

  \path[->] (s_1) edge              node        {$a$} (s_2)
            (s_2) edge              node {$b$} (s_4)
            (s_2) edge              node[above] {~$d$} (s_8)
            (s_1) edge              node[above]        {$c$} (s_6)
            (s_6) edge              node {$d$} (s_8);
\end{tikzpicture}
    \caption{Local DFA $A_2$ for 
  $ab|ad|cd$ 
    } \label{fig:1b}
  \end{subfigure}%
   \hspace*{\fill}   %
  \begin{subfigure}[t]{0.3\columnwidth}
    \begin{tikzpicture}[shorten >=1pt,node distance=1.2cm and 1.3cm,on grid,auto,thick]

  \node[state,initial, initial text=]  (s_1)                      {$s_1$};
  \node[state]          (s_2) [right=of s_1] {$s_2$};
  \node[state,accepting]          (s_4) [right=of s_2] {$s_3$};
  \node[state]          (s_6) [below=of s_2] {$s_4$};
  \node[state,accepting]          (s_8) [right=of s_6] {$s_5$};

  \path[->] (s_1) edge              node        {$a$} (s_2)
            (s_2) edge              node {$b$} (s_4)
            (s_1) edge              node[above]        {$c$} (s_6)
            (s_6) edge              node {$d$} (s_8)
            (s_2) edge              node       {$\epsilon$} (s_6);
\end{tikzpicture}
    \caption{RO-$\epsilon$NFA $A_3$ for 
    $ab|ad|cd$
    } \label{fig:1c}
  \end{subfigure}%
  \caption{Examples of local automata (\cref{expl:locallang}) and RO-$\epsilon$NFAs (see
  \cref{expl:locallang} and \cref{expl:lang}). Initial
  states are indicated by incoming arrows, final states are doubly circled.}
  \label{fig:locallang}
\end{figure}

Next we give a second known equivalent language-theoretic characterization of local
languages, that we call being \emph{letter-Cartesian} and that will also be useful
later:

\begin{definition}[Letter-Cartesian]
  \label{def:letter-Cartesian}
  A language $L$ is \emph{letter-Cartesian}
  if the following implication holds: for every  letter $x \in \Sigma$
  and words $\alpha,
  \beta, \gamma, \delta \in \Sigma^*$, if $\alpha x \beta \in
  L$ and $\gamma x \delta \in L$, then $\alpha x \delta \in L$.
\end{definition}

Note that, by exchanging the roles of $\alpha, \beta$ and $\gamma, \delta$, if
$\alpha x \beta \in L$ and $\gamma x \delta \in L$ then also $\gamma x \beta
\in L$. Hence the name \emph{letter-Cartesian}: for each letter $x \in \Sigma$,
the words of $L$ containing an $x$ can be expressed as a Cartesian product
between what can appear left of an~$x$ and what can appear right of an~$x$.
Formally, for a letter $x\in \Sigma$, define $L_x=\{\alpha \in \Sigma^* \mid \exists \beta \in \Sigma^* \text{ s.t. } \alpha x \beta \in L\}$ and  $R_x=\{\beta \in \Sigma^* \mid \exists \alpha \in \Sigma^* \text{ s.t. } \alpha x \beta \in L\}$.  Then, if $L$ is letter-Cartesian, we have that $L\setminus \{\epsilon\} = \bigcup_{x\in \Sigma} L_x x R_x$.

\begin{example}
  \label{ex:non-local}
  The local languages of \cref{expl:locallang} are
  letter-Cartesian.   
  Take now $L = aa$. With $x\coloneq a$, $\alpha \coloneq  a$, $\beta
  \coloneq  \epsilon$,
  $\gamma \coloneq  \epsilon$, $\delta \coloneq  a$, we see that $L$ is not letter-Cartesian (as $aaa\notin L$),
  and hence, as we will see, it is not local.
\end{example}

It is known that letter-Cartesian languages are precisely the local languages:

\begin{proposition}
  \label{prp:carac}
  A language is local if, and only if,
  it is letter-Cartesian.
\end{proposition}

This result has been claimed in \cite[Chp.\ II, Ex.\ 1.8,
(d)]{sakarovitch2009elements} without proof (as an exercise). For completeness, and
to provide more intuition on local languages to the reader,
we provide a proof here. One direction is easy to show: local languages are
always letter-Cartesian.

  \begin{claim}
    \label{clm:local-is-cart}
    If $L$ is local then it is letter-Cartesian.
  \end{claim}
  \begin{proof}
  Let $A$ be a local DFA recognizing $L$. Let $x \in \Sigma$, and let
  $\alpha,\beta,\gamma,\delta\in \Sigma^*$ such that $\alpha x \beta \in L$ and
  $\gamma x \delta \in L$, and let us show that $\alpha x \delta$ is in $L$ as
  well. Because $L$ is local, the state $q$ that is reached in~$A$ after reading
  $\alpha x$ in the accepting run of $\alpha x \beta$ is the same state that is
  reached after reading $\gamma x$ in the accepting run of $\gamma x \delta$.
  We can thus compose a run from the initial state to $q$ reading $\alpha x$,
  with the run from $q$ to some final state reading $\delta$, to form an
  accepting run for $\alpha x \delta$. This concludes the proof.
  \end{proof}

  We now prove the converse, thus establishing \cref{prp:carac}:
  \begin{claim}
    \label{clm:cart-is-local}
    If $L$ is letter-Cartesian then it is local.
  \end{claim}

To establish this claim, 
we use a construction from \cite[Proposition~2.1]{berstel1996local}.
We define the \emph{local overapproximation} of a
language $L$ as the local DFA constructed as follows:

\begin{definition}
  \label{def:localoverapprox}
    Let $L$ be a 
    language. 
    We start by defining the following:
    \begin{itemize}
      \item The set $\Sigma_\start$ of letters that can start words of~$L$,
        i.e., the letters $a \in \Sigma$ such that there exists a word $a \alpha \in
        L$ for some $\alpha \in \Sigma^*$. 
      \item The set $\Sigma_\eend$ of letters that can end words of~$L$, i.e.,
        letters $a \in \Sigma$ such that $\alpha a \in L$ for some $\alpha
        \in \Sigma^*$. 
      \item The set $\Pi \subseteq \Sigma^2$ of pairs of letters that can occur
        consecutively in a word of~$L$, formally, the pairs $(a,b)$ such
        that there exist $\alpha, \beta \in \Sigma^*$ such that $\alpha a b
        \beta \in L$. 
    \end{itemize}
    The \emph{local overapproximation of $L$} is the DFA $A$ defined in the following way. We have an initial state
    $q_0$, which is final iff $\epsilon \in L$. For every $a \in \Sigma$, we
    have a state $q_a$, which is final iff $a\in \Sigma_\eend$. Then, 
  for every $a\in \Sigma_\start$ we add an $a$-transition from $q_0$ to $q_a$,
  and
  for every $(a,b) \in \Pi$ we add a $b$-transition from $q_a$ to $q_b$.
\end{definition}

The local overapproximation $A$ built from $L$ is local by definition, and the following is
immediate:

\begin{claim}
  \label{clm:reallyimmediate}
  We have $\LL(A) \supseteq L$.
\end{claim}

We now claim that, for letter-Cartesian languages, the local overapproximation
is in fact equal to the original language, namely:

\begin{claim}
  \label{clm:notimmediateatall}
  If $L$ is letter-Cartesian, then $\LL(A) = L$.
\end{claim}

This immediately implies \cref{clm:cart-is-local}: if $L$ is letter-Cartesian,
then it is local, because $L$ is equal to $\LL(A)$ for $A$ the local
overapproximation, and $A$ is a local DFA. So all that remains is to show 
\cref{clm:notimmediateatall}.

  \begin{proof}[Proof of \cref{clm:notimmediateatall}]
    We must show that if $L$ is letter-Cartesian then $\LL(A) = L$. 
    By \cref{clm:reallyimmediate}, it suffices to
    show that $\LL(A) \subseteq L$.

Take $\alpha$
in $\LL(A)$, and let us show that it is in $L$. If $\alpha$ is the empty word then
it is clear by construction that $\alpha$ is also in $L$.
So assume $\alpha = a_1 \cdots a_m$ with $m\geq 1$.
By construction of~$A$, the fact that $\alpha$ is accepted witnesses that:

\begin{itemize}
  \item $a_1 \in \Sigma_\start$, hence there exists $\alpha'' \in \Sigma^*$ such
    that $a_1 \alpha'' \in L$;
  \item $a_m \in \Sigma_\eend$ hence there exists $\alpha''' \in \Sigma^*$ such
    that $\alpha''' a_m \in L$;
  \item For each $1 \leq i < m$ we have $(a_i,a_{i+1}) \in \Pi$, hence there
    exist $\alpha_i', \alpha_i'' \in \Sigma^*$ such that $\alpha_i' a_i a_{i+1} \alpha_i''$ is in~$L$.
\end{itemize}

If $m=1$ then from $a_1 \alpha'' \in L$ and $\alpha''' a_m = \alpha''' a_1 \in
L$, as $L$ is letter-Cartesian on~$a_1$, we deduce that $a_1 = \alpha \in L$.
Hence, in the sequel, we assume that $m>1$.

We now use the fact that $L$ is letter-Cartesian to prove the following
claim (*):
$\alpha_1' \alpha \alpha''_{m-1} \in L$.

To do this, we show by finite induction that $\alpha_1' a_1 \cdots
a_{i+1} \alpha''_i \in L$ for each $1 \leq i < m$. The base case of $\alpha'_1
a_1 a_2 \alpha''_1 \in L$ is by the definition of $\alpha_1'$ and $\alpha_1''$
above. The induction case for $1 \leq i < m-1$ is as follows:
from $\alpha_1' a_1 \cdots a_{i+1} \alpha''_i \in L$ (by induction hypothesis)
and $\alpha_{i+1}' a_{i+1} a_{i+2} \alpha''_{i+1} \in L$ (by definition of
$\alpha_{i+1}', \alpha_{i+1}''$ above), we deduce that 
$\alpha_1' a_1 \cdots a_{i+1} a_{i+2} \alpha''_{i+1} \in L$ because $L$ is
letter-Cartesian on $a_{i+1}$, which is what we wanted to show to
establish the induction case. Hence, the inductive claim holds and for $i=m-1$
we deduce that $\alpha_1' a_1 \cdots a_m \alpha''_{m-1} \in L$, which is claim
(*).
Now, using that $L$ is letter-Cartesian on~$a_1$, from (*) and the word
$a_1 \alpha'' \in L$ identified earlier,
we deduce that
$\alpha \alpha''_{m-1}$ is in $L$.
Symmetrically, from this and the word $\alpha''' a_m \in L$ identified
earlier, using the fact that $L$ is letter-Cartesian on~$a_m$ we
deduce that $\alpha \in L$, which concludes.
\end{proof}

We have concluded the proof of \cref{prp:carac}.

\myparagraph{Testing if a language is local.}
We can use the local overapproximation
(\cref{def:localoverapprox})
to show that it is PTIME to determine whether a language is local.
Indeed, the previous results imply:

\begin{claim}
\label{clm:notsoimmediate}
  Let $L$ be a language, and $A$ its local overapproximation. Then we have $\LL(A) = L$ if and only if $L$ is local.
\end{claim}

\begin{proof}
  If $L$ is local, then $L$ is letter-Cartesian by \cref{prp:carac}
  so by 
  \cref{clm:notimmediateatall} we have
  $\LL(A) = L$. Conversely, if $\LL(A) = L$ then $L$ is recognized by a local
  DFA, so it is local.
\end{proof}

This result gives us a way to check whether the language recognized by an
automaton is local. First, in the case of DFA, we can check this tractably:

\begin{proposition}[Locality testing for DFAs]
  \label{prp:dfaepsro}
  Given a DFA $A$, it is in PTIME to determine if $\LL(A)$ is local.
\end{proposition}
\begin{proof}
  Let $L \coloneq \LL(A)$. Build the local overapproximation $A'$ of~$L$ using
  \cref{def:localoverapprox}, which is clearly in PTIME.
  By \cref{clm:notsoimmediate}, 
  it suffices to check whether $L = \LL(A')$.
  It is well-known that testing equivalence of two DFAs is PTIME, so this concludes the proof.
\end{proof}

Note that \cref{prp:dfaepsro} is about testing if the language $L = \LL(A)$
recognized by~$A$ is local (i.e., some local DFA for $L$ exists); it is not
about testing if the specific DFA $A$ is a local DFA (which is easy to do in linear
time simply by following the definition). It is also not about the complexity of
testing whether the \emph{infix-free sublanguage} of~$\LL(A)$ is local, which can
be the case even if $\LL(A)$ is not local. We can do this by building an automaton $A'$ representing
$\red(L)$ from a DFA $A$ representing~$L$ and then study~$A'$, but the
straightforward construction for~$A'$ entails an exponential blowup
(\cite[Proposition 6]{barcelo2019boundedness}). We leave the precise complexity
of this task open.

By contrast, if the input language~$L$ is given as
an NFA instead of a DFA, then testing if~$L$ is local is
PSPACE-complete~\cite{amarilli2025locality}. %
This holds even under the assumption that $L$ is infix-free, so that the same lower bound
applies to the problem of testing whether the infix-free
sublanguage of an input
NFA is local.

\subsection{Tractability of Resilience for Local Languages}
\label{subsec:local-ptime}

Having introduced local languages and presented their properties, we can now
turn to the main result of this section, namely, that $\resprobbag(L)$ is PTIME when $L$ is local:

\begin{theorem}[Resilience for local languages]
  \label{prp:ro-ptime}
  If $L$ is local then $\resprobbag(L)$ is in PTIME. Moreover, the problem is
  also tractable in combined complexity: given an $\epsilon$NFA $A$ that
  recognizes a local language $L$, and a database $D$, we can 
  compute $\resbag(\query_L,D)$
  in 
  $(|A| \times |D| \times |\Sigma|)^{1+o(1)}$.
\end{theorem}

Note that, in the combined complexity upper bound, the input $A$ is not
necessarily a local DFA:
it suffices to be given
an arbitrary $\epsilon$NFA $A$ recognizing~$L$, together with the promise that
$L$ is local.

Before proving \cref{prp:ro-ptime}, we remark that it implies a
PTIME bound on $\resprobbag(L)$ 
whenever $\red(L)$ is local (as $\resprobbag(L)$ is equivalent to
$\resprobbag(\red(L))$).
For instance, $\resprobbag(L_0)$ for $L_0=\{a,aa\}$ is in PTIME because
$\red(L_0) = \{a\}$ is local, even though $L_0$ is not.
In fact, to see if \cref{prp:ro-ptime} applies to a language~$L$, 
we can always focus on $\red(L)$, thanks to the following lemma (whose proof is
deferred to \cref{apx:qwerty}):

\begin{toappendix}
  \subsection{Proof of \cref{lemma:qwerty}: Infix-Free Sublanguages Preserve Locality}
  \label{apx:qwerty}
\end{toappendix}
\begin{lemmarep}[Infix-free sublanguages preserve locality]
  \label{lemma:qwerty}
  If $L$ is local then so is $\red(L)$.
\end{lemmarep}
\begin{proof}
  By \cref{prp:carac} we have that a language is
  local iff it is letter-Cartesian. So we show the equivalent claim: if $L$ is
  letter-Cartesian then so is $\red(L)$.
  Let $\alpha x \beta$ and $\gamma x \delta$
  be in $\red(L)$, and let us show that $\alpha x \delta \in \red(L)$ as well.
  We have $\alpha x \beta \in L$ and $\gamma x \delta \in L$ and
  $L$ is letter-Cartesian, so we have $\alpha x \delta \in L$. Assume by
  contradiction that $\alpha x \delta \notin \red(L)$. Then $\alpha x \delta$
  has a strict infix, call it $\tau$, that is in $L$. Notice that $\tau$
  cannot be an infix of $\alpha$, because $\alpha x \beta \in \red(L)$ and
  $\tau$ would then be a strict infix of $\alpha x \beta$; and
  similarly $\tau$ cannot be an infix of $\delta$. Therefore $\tau$ must be of
  the form $\tau = \alpha_1 x \delta_1 \in L$ with $\alpha = \alpha_2
  \alpha_1$ and $\delta = \delta_1 \delta_2$. We further have 
  one of $\alpha_2 \neq
  \epsilon$ and $\delta_2 \neq \epsilon$, because $\tau$ is a strict infix of $\alpha x \delta$.
  Assume
  that $\alpha_2 \neq \epsilon$ (the case of $\delta_2 \neq \epsilon$ is
  symmetrical). Then we have $\alpha_1 x \delta_1 \in L$, and $\alpha x \beta \in
  L$, so since $L$ is letter-Cartesian we have $\alpha_1 x
  \beta \in L$. But as $\alpha_2 \neq \epsilon$ we know that $\alpha_1 x
  \beta$ is a strict infix of~$\alpha x \beta$, contradicting the fact that $\alpha x \beta \in \red(L)$. This concludes
  the proof.
\end{proof}

The proof of \cref{prp:ro-ptime} is not difficult and is a simple reduction to
the MinCut problem. To present it,
we need a variant of
local automata, which we name \emph{Read-Once (RO) automata}.
Intuitively, RO automata strengthen the requirement of local automata by requiring
that for each letter $a \in \Sigma$ there is \emph{at most one transition} labeled
with~$a$; but in exchange RO automata are not required to be deterministic and allow for
(arbitrarily many) $\epsilon$-transitions. 
Formally:
 
\begin{definition}[RO-$\epsilon$NFA]
  A \emph{read-once $\epsilon$NFA} (RO-$\epsilon$NFA for short) is an $\epsilon$NFA $A = (S,I,F,\Delta)$ such that for every letter $a\in \Sigma$
  there is at most one transition of the form $(s,a,s')$ in
  $\Delta$.
\end{definition}

\begin{example}[\cref{expl:locallang} continued]
  \label{expl:lang}
  The local DFA $A_1$ in \cref{fig:1a} for $ax^*b$ is an
  RO-$\epsilon$NFA, but the local DFA $A_2$ in \cref{fig:1b} for $ab|ad|cd$ is
  not. An RO-$\epsilon$NFA for this language is $A_3$ 
  (\cref{fig:1c}).
\end{example}

We show that RO-$\epsilon$NFAs recognize
precisely the local languages (and the presence of $\epsilon$-transitions is
essential for this, see \cref{apx:roexpr}):

\begin{lemma}[RO-$\epsilon$NFA = local languages]
  \label{lem:localepsro}
  A language $L$ is local iff it is recognized by an RO-$\epsilon$NFA. 
  Moreover, given an $\epsilon$NFA $A$ recognizing a local language $L =
  \LL(A)$, we can
  build 
  in $O(|A|\times|\Sigma|)$ an RO-$\epsilon$NFA recognizing $L$. 
\end{lemma}
\begin{proof}
  We first prove the first sentence of the claim.

    We start by showing that a local DFA can be transformed into an equivalent
    RO-$\epsilon$NFA. Let $A = (S, I, F, \Delta)$ be the local DFA. We build
    from it an RO-$\epsilon$NFA $A' = (S', I, F, \Delta')$ defined as follows.
    We let $S' \coloneq S \cup \{s'_a \mid a \in \Sigma\}$ where the states of
    the form $s_a'$ are fresh. For every letter $a \in \Sigma$ that has a
    transition in~$A$, letting $s_a$ be the common target state of all
    $a$-transitions in~$A$, we create one $a$-transition $(s_a', a, s_a) \in
    \Delta'$, and add $\epsilon$-transitions $(s, \epsilon, s_a') \in \Delta'$
    from every state $s \in S$ that has an outgoing $a$-transition in~$A$. The
    resulting automaton $A'$ is an RO-$\epsilon$NFA by construction, and there
    is a clear bijection between the runs of~$A$ and those of~$A'$.
    
    Next, we show that an RO-$\epsilon$NFA can be transformed into an equivalent
    local DFA. Let $A$ be the RO-$\epsilon$NFA. First add a super-initial
    state that has $\epsilon$-transitions to all initial states of~$A$, and make
    it the only initial state of~$A$. This does not change the language
    recognized by~$A$, and ensures that~$A$ has only one initial state.
    Then, we remove
    $\epsilon$-transitions in PTIME using standard techniques: first ensure
    that every state having an $\epsilon$-path to a final state is final; then
    for every letter $a \in \Sigma$ occurring in~$A$, for every state $s$
    having an $\epsilon$-path to the starting state $s_a'$ of the unique
    $a$-transition $(s_a', a, s_a)$ of~$A$, add an $a$-transition from $s$
    to~$s_a$; finally remove all $\epsilon$-transitions. It is clear that the
    resulting automaton $A'$ is a local DFA that recognizes the same language,
    which concludes the proof of the first sentence of the claim.

  The second sentence follows from the construction given in
  \cref{clm:cart-is-local}, noticing that the sets $\Sigma_\start$,
  $\Sigma_\eend$ and $\Pi$ can be computed in time 
  $O(|A|\times |\Sigma|)$.
    Indeed, first assume that the input $\epsilon$NFA is trimmed (see
    \cref{def:trim}) using \cref{clm:trim}. The set $\Sigma_\start$ can then be computed in time 
    $O(|A|)$ by performing a depth-first search
    from the initial states of $A$, following $\epsilon$-transitions and finding which letters it can reach; and the same holds for $\Sigma_\eend$. To compute $\Pi$ we do the following: for every letter $a \in \Sigma$, we do a depth-first search (with complexity $O(|A|)$) starting from the targets of $a$-transitions and following $\epsilon$-transitions to find which letters $b$ this can reach, adding $ab$ in $\Pi$. We thus obtain a total complexity of 
    $O(|A|\times |\Sigma|)$.
\end{proof}

The point of RO-$\epsilon$NFAs is that we can use them to solve resilience
with a product construction, and conclude the proof of our main tractability
result in this section (\cref{prp:ro-ptime}):

\begin{proof}[Proof of \cref{prp:ro-ptime}]
  Given an $\epsilon$NFA $A'$ for the local language $L$, we compute an RO-$\epsilon$NFA 
  $A = (S,I,F,\Delta)$ for~$L$ in $O(|A'| \times |\Sigma|)$ using
  \cref{lem:localepsro}.
  We check in time $O(|A|)$ if $\epsilon \in L$: if so then resilience
  is~$+\infty$, so in the sequel we assume that $L$ does not contain~$\epsilon$.
  Let $D$ be the input bag database with fact multiplicities given by $\mult$,
  and let $V$ be the domain of~$D$.
  Define a network~$\calN_{D,A}$ as follows:
  \begin{itemize}
    \item The set of vertices is $(V \times S) \cup \{t_\text{source},t_{\text{target}}\}$,
      with the fresh vertices $t_{\text{source}}$ and $t_{\text{target}}$ being
      respectively the source and target of the network.
    \item There is an edge with capacity $\mult(v \xrightarrow{a} v')$ from $(v,s)$ 
      to $(v',s')$ for all $v,v' \in V$ and $s,s'\in S$ and $a\in
      \Sigma$ such that $v \xrightarrow{a} v' \in D$ 
      and $(s,a,s')\in \Delta$
    \item There is an edge with capacity $+\infty$ from $(v,s)$
      to $(v,s')$ for all $v \in V$ and $s,s' \in S$ such that $(s,\epsilon,s')\in
      \Delta$.
    \item There is an edge with capacity $+\infty$ from $t_{\text{source}}$ to every vertex
      of the form $(v,s)$ with $s\in I$, and from every vertex of the form $(v,s)$
      with $s\in F$ to~$t_{\text{target}}$.
  \end{itemize}
  Clearly $\calN_{D,A}$ is built in $O(|A|\times |D|)$
  from $A$ and $D$.
  Moreover, there is a one-to-one correspondence between the
  contingency sets of~$D$ for~$\query_L$ and the 
  cuts of $\calN_{D,A}$ that have finite cost: intuitively, selecting facts of~$D$ 
  in the
  contingency set corresponds to cutting edges in~$\calN_{D,A}$, and
  ensuring that the remaining facts do not satisfy~$\query_L$ corresponds to
  ensuring that no source to target paths remain.
  Crucially, as~$A$ is read-once,
  the correspondence between facts of $D$ and edges of
  $\calN_{D,A}$ with finite capacity is one-to-one.
  Further, the sum of multiplicities of the facts of a contingency set is equal to the cost of the
  corresponding cut, so minimum cuts correspond to minimum contingency sets.
  We can thus solve MinCut on $\calN_{D,A}$
  in $|\calN_{D,A}|^{1+o(1)}$~\cite[Theorem~1.1]{van2023deterministic}
to compute the
  answer to $\resprobbag(L)$ in~$D$.
\end{proof}

\begin{toappendix}
  \subsection{Trimming $\epsilon$NFAs}

In this appendix, we give some standard definitions on automata~\cite[Section
1.3.1]{sakarovitch2009elements}:

  \begin{definition}
    \label{def:trim}
    Let $A = (S,I,F,\Delta)$ be an $\epsilon$NFA.
    We say that a state $s \in S$ is \emph{accessible} if there is a path
    in~$A$ from a state of~$I$ to~$s$ using transitions of~$\Delta$. We say
    that $s$ is \emph{co-accessible} if there is a path in~$A$ from~$s$ to a
    state of~$F$ using transitions of~$\Delta$. We say that $s$ is
    \emph{useful} if it is both accessible and coaccessible. We say that $A$ is
    \emph{trimmed} (or \emph{trim}~\cite{sakarovitch2009elements}) if every state of~$A$ is useful.
  \end{definition}

We give a self-contained proof of the well-known fact on automata~\cite{sakarovitch2009elements} that they can be trimmed in linear time:

  \begin{claim}
    \label{clm:trim}
    Given an $\epsilon$NFA $A$, we can compute in time $O(|A|)$ an $\epsilon$NFA
    $A'$ which is trimmed and such that $\LL(A) = \LL(A')$.
  \end{claim}

  \begin{proof}
    We first do one graph traversal in~$A$ from~$I$ to compute the accessible
    states in $O(|A|)$. Then, we do one graph traversal from~$F$ following the transitions
    of~$\Delta$ backwards to compute the co-accessible states in~$O(|A|)$. Then,
    we obtain $A'$ by simply removing all states of~$A$ that are not useful, and
    removing all transitions that involve these states. The computation of
    $A'$ from~$A$ is in time $O(|A|)$ overall.

    It is clear that $A'$ is
    then trimmed, because every remaining state $s$ is a state which was useful
    in~$A$, and the paths witnessing it still witness that $s$ is useful
    in~$A'$. Further, $A'$ recognizes the same language as $A$: the construction
    clearly ensures $\LL(A') \subseteq \LL(A)$, and conversely there are no
    accepting runs
    in~$A$ involving states that are not useful, so every accepting run in~$A$
    is also an accepting run of~$A'$. This concludes the proof.
  \end{proof}

\end{toappendix}

\begin{toappendix}
  \subsection{Expressiveness of RO-$\epsilon$NFAs Without $\epsilon$-Transitions}
\label{apx:roexpr}

We explain here why RO-$\epsilon$NFAs without $\epsilon$-transitions are strictly
less expressive than RO-$\epsilon$NFAs with $\epsilon$-transitions: this
contrasts with the setting of NFAs, in which $\epsilon$NFAs and NFAs have the
same expressive power (i.e., they both recognize all regular languages).

Let us define an \emph{RO-DFA} as an RO-$\epsilon$NFA without
$\epsilon$-transitions; indeed, note that an RO-$\epsilon$NFA without
$\epsilon$-transitions is automatically deterministic thanks to the read-once
condition. RO-DFAs are also a subclass of local DFAs (and a strictly less
expressive one, as we now show). Indeed:

\begin{lemma}
  There is no RO-DFA recognizing the language $ab|ad|cd$ from \cref{expl:lang}.
\end{lemma}

\begin{proof}
Let $A$ be an RO-DFA which accepts $ab$, $ad$, and $cd$. Then the target
state of the $a$-transition must be the same as the starting state of the
$b$-transition (because $A$ accepts $ab$), and it must also be the same as 
the starting state of the $d$-transition (because $A$ accepts $ad$), and the
latter state must also be the same as the target state of the
$c$-transition (because $A$ accepts $cd$). This implies that $A$ also accepts $cb$.
\end{proof}

We leave open the question of characterizing the precise subset of the local
languages that are accepted by RO-DFA~\cite{cstheoryepsro}.

\end{toappendix}

\section{Hardness Techniques and First Results}
\label{sec:hard}
In this section, we move on to cases where computing resilience is intractable. 
From now on we only study languages
that are infix-free, because showing hardness for an infix-free language $L$ is
sufficient to imply the same for any language $L'$ such that $\red(L') = L$.

We start the section by showing the NP-hardness of $\resprobset(aa)$, by an explicit
reduction from the minimum vertex cover problem. 
We then generalize this reduction and
define the notion of \emph{(hardness) gadgets} as databases that can be used
to encode graphs and reduce from minimum vertex cover.
We show that whenever a language~$L$ admits a gadget 
then $\resprobset(L)$ is
NP-hard.
We finally illustrate these techniques to show the
hardness of resilience for the specific language
$axb|cxd$.

The notion of gadgets introduced in this section will be used in the next two 
sections to show the hardness of entire language families. Specifically, we will
generalize in \cref{sec:fourlegged} the hardness of $\resprobset(axb|cxd)$ to that
of the so-called \emph{four-legged languages}; and generalize in \cref{sec:repeated} the hardness
of $\resprobset(aa)$ to that of all finite languages featuring a word with a
repeated letter.

\subsection{Hardness of $aa$}
\label{sec:hardness:aa}

We first present the hardness proof for the language $aa$.
This result is already known from the setting of CQs~\cite[Prop.~10]{DBLP:conf/pods/FreireGIM20},
by a direct reduction from 3-SAT.
We give a self-contained proof by reduction from the minimum vertex cover (VC) problem,
inspired by the edge-gadget framework of~\cite{makhija2024unified}.
This proof serves as a first illustration of the reductions used throughout the section:
we encode a graph given as input to VC by replacing each edge with a fixed database gadget,
so that hitting 
all query matches (i.e., selecting a set $H$ of edges such that every match of the query contains an edge of $H$) corresponds to choosing a vertex cover of the original graph.

The query $Q_{aa}$ asks whether the input database contains a directed walk of length two whose
two facts are labeled by $a$. Thus, a contingency set for $Q_{aa}$ is a set of $a$-facts that hits all
directed $aa$-walks. At a high level, the reduction starts from an undirected VC instance, chooses
an arbitrary orientation of its edges, and replaces each directed edge by a copy of a fixed
\emph{(oriented) database gadget}. 
The gadget has \emph{two distinguished endpoint facts}, one for each
endpoint of the encoded VC edge. 
All gadget copies incident to the same VC vertex share the
corresponding endpoint fact.

The gadget need not be symmetric, since any orientation of the VC instance suffices. What matters
is composability: identifying copies at shared endpoint facts must not create unintended directed
$aa$-walks across different copies, which would be spurious matches. In the terminology
of~\cite{makhija2024unified}, the \emph{composition of gadgets} must be \emph{non-leaking}. In our setting,
this is enforced syntactically by the choice of endpoints: endpoint elements never occur as heads of
internal facts, so a directed walk cannot leave the internal part of one gadget copy and enter another.

Inside one completed copy of the gadget, the graph whose vertices are facts and whose edges are
$aa$-matches is an \emph{odd path} between the two endpoint facts. After identifying copies according
to the VC instance, the graph of $aa$-matches is an odd subdivision of the original undirected graph.
Hitting 
all $aa$-matches is therefore the same as finding a vertex cover of this odd
subdivision; and, as we will recall in
\cref{prp:subdivision},
the minimum vertex cover of an odd
subdivision determines the minimum vertex cover of the original graph.

To summarize, in the language of \emph{independent join
paths}~\cite{makhija2024unified}, 
the endpoint facts provide the two alternatives for covering 
edges (the semantic \emph{OR-property}), while the syntactic endpoint condition guarantees non-leaking
composition (the \emph{composability} property). Thus, the gadgets below can be seen as directed-walk analogues of independent join
paths. Let us now turn to the concrete case of the query $aa$ 
before 
we present more general hardness gadgets in the next subsection.

\begin{proposition}
  \label{prp:aa}
  $\resprobset(aa)$ is NP-hard.
\end{proposition}

In other words, given a
directed graph~$G$ and integer $k \in \mathbb{N}$, it is NP-hard to decide if we can remove $k$
edges of~$G$ such that $G$ no longer contains any walk of length 2.
We show \cref{prp:aa} by reducing from the \emph{minimum vertex cover problem}. 
Recall that a \emph{vertex cover} of an undirected graph $G = (V, E)$ is a
subset of vertices $V' \subseteq V$ such that every edge of $E$ is incident to a
vertex of~$V'$, i.e., for each $e \in E$, we have $e \cap V' \neq \emptyset$.
The \emph{vertex cover number} 
of~$G$ is the minimum cardinality of a vertex
cover of~$G$.
Given an undirected graph $G$ and
an integer $k \in \mathbb{N}$, 
the \emph{minimum vertex cover problem} asks if the vertex cover number of $G$ is at most~$k$.
It is an NP-complete problem~\cite[Main Theorem]{karp1972reductibility}.

Before showing the reduction, we show that the vertex cover number 
of an undirected graph~$G$ can be computed from the vertex
cover number of any odd-length subdivision of~$G$. 
Formally, for $\ell \in \mathbb{N}$, $\ell \geq 1$, the \emph{$\ell$-subdivision} of~$G$ 
is the undirected graph $G'$
(unique up to isomorphism) that we obtain
by replacing each edge of~$G$ by a path of length $\ell$ (i.e., with
$\ell-1$
fresh internal vertices, that are disjoint across edges).
We claim:

\begin{proposition}[Hardness of vertex cover on subdivisions]
  \label{prp:subdivision}
  Let $\ell \in \mathbb{N}$ be an odd integer. Then for any undirected graph $G =
  (V, E)$, letting $G'$ be the 
  $\ell$-subdivision of~$G$, then the vertex cover number of~$G'$ is 
  $k + m(\ell-1)/2$, 
  where $k$ is the vertex cover number of~$G$ and $m$ is the
  number of edges of~$G$.
\end{proposition}

\begin{proof}
  It suffices to show that $G$ has a vertex cover of cardinality at most~$k$
  if and only if $G'$ has a vertex cover of cardinality at most $k+m(\ell-1)/2$.

  We first show the “only if” direction. Let $G = (V, E)$ and let $C$ be a vertex
  cover of~$V$ of cardinality at most~$k$. We pick a subset $C'$ of the vertices of $G'$
  by taking the vertices of $C$ and picking vertices of the subdivisions in the
  following way. For each edge $\{u,v\}$ of~$G$, letting $u - v_1
  - \cdots v_{\ell-1} - v$ be the corresponding path in~$G'$,
  we distinguish three cases:
  \begin{itemize}
    \item If $u \in C$ and $v\notin C$ then we add to $C'$ the vertices $v_2,
      v_4, \ldots,
      v_{\ell-1}$
    \item If $v \in C$ and $u \notin C$ then we add to $C'$ the vertices $v_1,
      v_3, \ldots,
      v_{\ell-2}$.
    \item If $u \in C$ and $v \in C$ we do one of the above (it does not matter
      which).
  \end{itemize}
  The case where $u \notin C$ and $v \notin C$ is impossible because $C$ is a
  vertex cover of~$G$.

  Now, for each edge of~$G$ we select $(\ell-1)/2$ vertices in~$C'$, so in total $C'$
  has the prescribed cardinality. To see that it is a vertex cover, note that
  all edges of~$G'$ occur in a path of the form $u - v_1
  - \cdots v_{\ell-1} - v$, and the vertices of~$C'$
  selected above form a vertex cover of the edges of such paths. Hence, $C'$ is
  a vertex cover of~$G'$.

  We next show the second direction, which is somewhat more challenging. Given a
  vertex cover $C'$ of~$G'$ of cardinality at most $k + m(\ell-1)/2$, we will
  first modify $C'$ without increasing its cardinality to guarantee the
  following property (*): for each edge $(u,v) \in E$, considering the
  corresponding path $u - v_1
  - \cdots v_{\ell-1} - v$ of~$G'$, then there are
  precisely $(\ell-1)/2$ of the intermediate vertices of the path that are
  in~$C'$, formally $|C' \cap \{v_1, \ldots, v_{\ell-1}\}| = (\ell-1)/2$.

  For this, we process each path $u = v_0 - v_1
  - \cdots v_{\ell-1} - v_\ell = v$, and consider each vertex $v_1,
  \ldots, v_{\ell-1}$ in order. If for $i>1$ we have that $v_i \in C'$ and the
  previous vertex $v_{i-1}$ was already in~$C'$, then we remove $v_i$ from $C'$
  and add instead $v_{i+1}$ to~$C'$. Note that the cardinality of $C'$ does not
  increase, and may strictly decrease if $v_{i+1}$ was already in~$C'$. At the
  end of this process, given that we start with a vertex cover $C'$ of~$G'$, we see
  that $C'$ contains precisely $(\ell-1)/2$ vertices of the path: if $v_1 \in
  C'$ then we keep the $v_i$'s for which $i$ is odd, and otherwise the $v_i$'s
  for which $i$ is even.

  We have processed $C'$ to still be a vertex cover, and have cardinality
  $m(\ell-1)/2 + k'$ for some $k' \leq k$, with $k' = |C' \cap V|$. Let us define
  $C \coloneq  C' \cap V$. It has cardinality at most~$k$, and we claim that~$C$ is a
  vertex cover. Indeed, assuming by contradiction that some edge $(u,v)$ of~$G$
  is not covered by~$C$, so that $u \notin C$ and $v \notin C$,
  and considering the path $u = v_0 - v_1
  - \cdots v_{\ell-1} - v_\ell = v$, we know that $C'$
  (after the processing of the previous paragraph) was a vertex cover of~$G'$,
  but if $u \notin C$ and $v \notin C$ then it cannot be the case that precisely
  $(\ell-1)/2$ of the $v_i$'s are in~$C'$ and cover all the edges of the path, a
  contradiction.

  We have thus shown both directions of the claimed equivalence, which concludes
  the proof.
\end{proof}

Let us now show the hardness of resilience for the language $aa$:

\begin{proof}[Proof of \cref{prp:aa}]
    \begin{figure}
      \centering
      \begin{subfigure}{.4\linewidth}
        \centering
        \begin{tikzpicture}[xscale=1,yscale=.7]
          \color{orange}          
          \node (v) at (1, 0) {$v$};
          \node (u) at (.33, 1) {$u$};
          \node (w) at (1.67, 1) {$w$};
          \node (v2) at (1, -2.5) {$v$};
          \node (u2) at (.33, -1.5) {$u$};
          \node (w2) at (1.67, -1.5) {$w$};
          \color{black}
          \draw (u) -- (v);
          \draw (v) -- (w);
          \draw (w) -- (u);
          \draw[->] (u2) -- (v2);
          \draw[->] (v2) -- (w2);
          \draw[->] (w2) -- (u2);
        \end{tikzpicture}
          \caption{Undirected graph $G$ (top) and arbitrary orientation $G'$ of~$G$ (bottom)}
          \label{fig:coding1}
      \end{subfigure}
      \hfill
  \null\begin{subfigure}{.13\linewidth}
    \centering
    \includegraphics[scale=0.35]{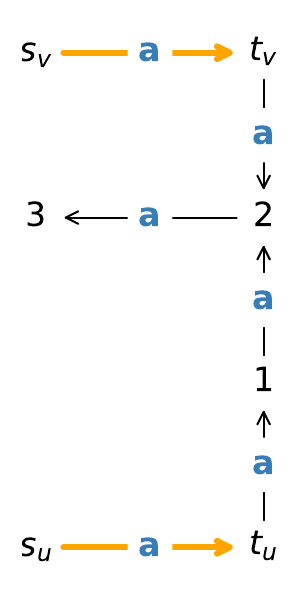}
    \caption{Completed gadget for $aa$
    }
    \label{fig:aagadget}
  \end{subfigure}
  \hfill
  \begin{subfigure}{.19\linewidth}
    \centering
  \includegraphics[scale=0.35]{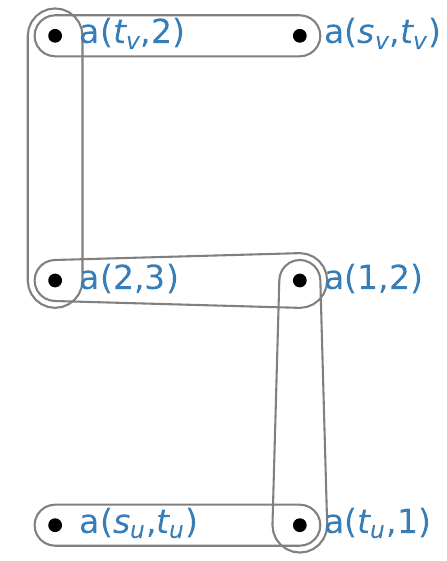}
    \caption{Graph of $aa$-matches of \cref{fig:aagadget}
    }
    \label{fig:aamatches}
  \end{subfigure}
      \hfill
      \begin{subfigure}{.35\linewidth}
        \centering
      \includegraphics[scale=0.4]{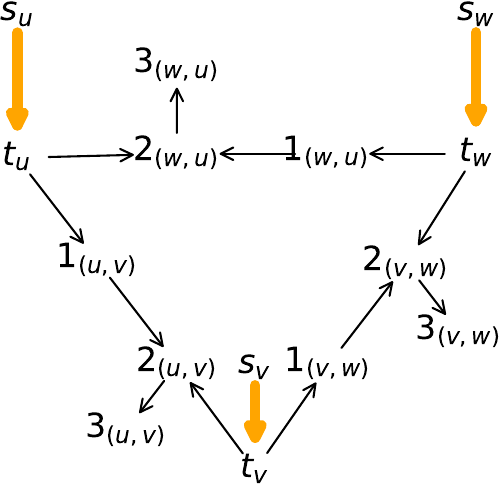} 
      \caption{Database $\Xi$ constructed to encode $G'$.
        All edges
        are labeled by~$a$.}

      \label{fig:aacoding2}
      \end{subfigure}
      \begin{subfigure}{.58\linewidth}
      \centering
        \includegraphics[scale=0.44]{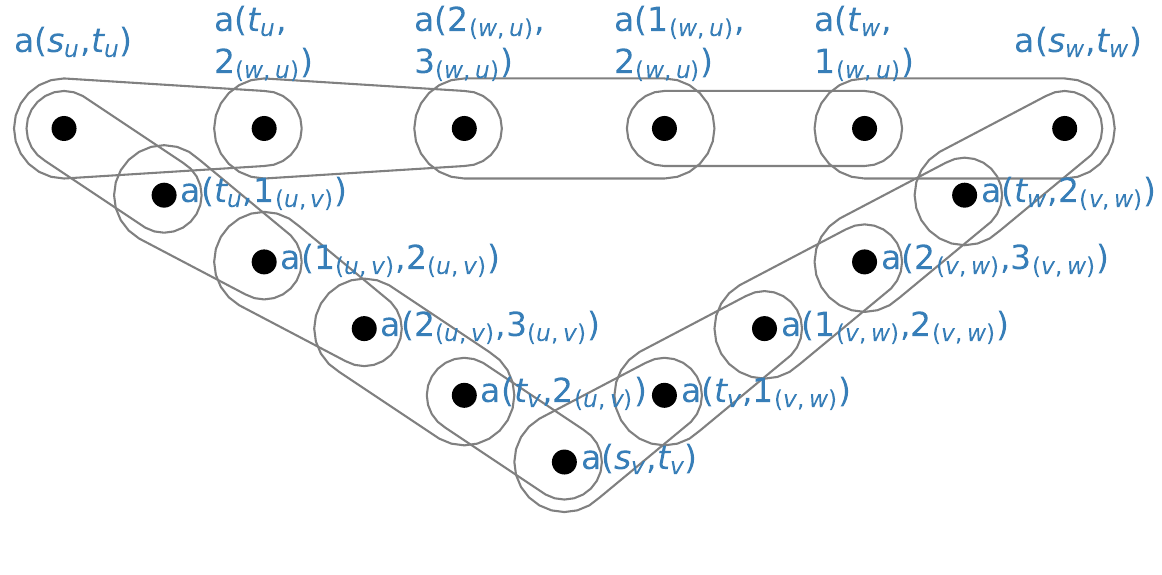}
        \caption{Graph of $aa$-matches of \cref{fig:aacoding2}
        }
        \label{fig:aatrianglematches}
      \end{subfigure}
      \caption{Illustration of the proof of \cref{prp:aa}.
      }
      \label{fig:aacoding}
    \end{figure}

We reduce from the minimum vertex cover problem. Let $G$ be an undirected graph and $k
  \in \mathbb{N}$ be an integer. We perform the following construction, 
  illustrated in \cref{fig:aacoding}.
  We start by picking an arbitrary orientation for each edge of~$G$, giving a
  directed graph $G' = (V, E)$: see \cref{fig:coding1}.
  We then construct a database $\Xi$ in the following way.
  First, for each vertex $u \in V$
  we create a fact $s_u \xrightarrow{a} t_u$. 
  Second, for each edge $e = (u,
  v)$ of~$E$
  we create the following set $D_e$ of facts:
  $t_u \xrightarrow{a} t_{e,1}$,
  $t_{e,1} \xrightarrow{a} t_{e,2}$,
  $t_{e,2} \xrightarrow{a} t_{e,3}$,
  $t_{v} \xrightarrow{a} t_{e,2}$. 
  Together with $s_u \xrightarrow{a} t_u$ and $s_v \xrightarrow{a} t_v$, the
  facts of~$D_e$ form an instance 
  that we call the \emph{completed gadget} $D_e'$
  for~$e$: we illustrate
  $D_e'$ in \cref{fig:aagadget}, abbreviating $t_{e,i}$
  as $i$ for brevity.
  We introduce the notations $F_{e,\ii} = s_u \xrightarrow{a} t_u$ 
  and $F_{e,\oo} = s_v
  \xrightarrow{a} t_v$ to refer to the two facts of $D_e' \setminus D_e$, called
  the 
  \emph{endpoint facts}: 
  they are 
  shown in bold orange in \cref{fig:aagadget}.
  The overall database $\Xi$ built from the orientation~$G'$ of~$G$ is 
  illustrated in \cref{fig:aacoding2}, where for brevity we 
  abbreviate $t_{e,i}$ as $i_e$.
  Building $\Xi$  from~$G$
is in $O(|G|)$ (as $D_e$ has constant size). 

  We then claim that the resilience of $aa$ on~$\Xi$ is equal to the vertex cover
  number of the 5-subdivision of~$G$. 
  Indeed,
  define the \emph{graph of $aa$-matches} on a database~$D$ that has no self-loop (i.e., no fact of the form $u\xrightarrow{a}u$)
  as the undirected graph having one vertex for each fact
  of~$D$
  and one edge for each pair of facts of~$D$ that forms an $aa$-walk.
  It is then clear that $\resset(aa,\Xi)$ is equal to the vertex cover number of the
  graph of $aa$-matches of $\Xi$ (noticing that by construction, $\Xi$ has no self-loops).
  Notice then that the graph of $aa$-matches of each~$D_e'$
  is a path of length 5 from vertex 
  $F_{e,\ii}$ to vertex $F_{e,\oo}$ 
  (see \cref{fig:aamatches}).
  Moreover, due to the shape of the gadgets and to how we glue them, facts of different
  gadget copies that are not endpoint facts cannot occur together as an edge of
the graph of $aa$-matches (in the terminology mentioned above, this corresponds
to the gadget composition being non-leaking).
  Thus, the graph of $aa$-matches of~$\Xi$ (shown in \cref{fig:aatrianglematches}) is (isomorphic to) 
the
  5-subdivision of~$G$,
   and we conclude by \cref{prp:subdivision}.
\end{proof}

\subsection{Hardness Gadgets}
To show the hardness of resilience for other languages than $aa$,
we will build \emph{hardness gadgets}, or \emph{gadgets} for
brevity.
Gadgets are used as building blocks to reduce from
the vertex cover problem, generalizing the idea of \cref{fig:aagadget}
from \cref{sec:hardness:aa}.
We first define \emph{pre-gadgets} as databases with two domain elements to
which endpoint facts will be attached. We then define the \emph{encoding} $\Xi$ of
a directed graph $G$ by gluing together 
copies of the pre-gadgets following~$G$, as 
in the proof of \cref{prp:aa}. We last define \emph{gadgets} as pre-gadgets
whose hypergraph of matches 
of the query~$\query$ forms an odd path. 
The vertex cover problem
on~$G$ then reduces to the resilience of~$\query$ on~$\Xi$.

\begin{definition}[Pre-gadget]%
  \label{def:pregadget}
  A \emph{pre-gadget} 
  is a  tuple $\Gamma = (D, t_\ii, t_\oo, a)$
  consisting of a database $D$, two distinguished domain elements $t_\ii \neq t_\oo$
  of~$D$
  respectively called the \emph{in-element} and \emph{out-element},
  and a label $a \in \Sigma$. We require that $t_\ii$ and
  $t_\oo$ never occur as heads of facts of~$D$, i.e., there are no facts of the form $u \xrightarrow{x} t_\ii$
  or $u \xrightarrow{x} t_\oo$ in $D$.
  The \emph{completion}
  of~$\Gamma$ is the database $D'$ obtained from $D$ by adding two fresh distinct
  domain elements $s_\ii$ and $s_\oo$ and the two facts 
  $s_\ii \xrightarrow{a} t_\ii$ and 
  $s_\oo \xrightarrow{a} t_\oo$ using letter $a$.
\end{definition}

\begin{example}
  \cref{fig:aagadget} gives an example
  of the completion of a pre-gadget. The
  in-element and out-element of the pre-gadget are respectively $t_u$ and $t_v$, and
  the facts $s_u \xrightarrow{a} t_u$ and $s_v \xrightarrow{a} t_v$ are the
  facts added in the completion (so the pre-gadget has 4 edges and the
  completion has 6).
\end{example}

We next explain how we use gadgets to encode directed graphs (as in
\cref{fig:aacoding}):
\begin{definition}[Encoding a graph with a pre-gadget]
  \label{coding}
  Given a directed graph $G$
  and a pre-gadget $\Gamma = (D, t_\ii, t_\oo, a)$,
  the \emph{encoding} of~$G$ with~$\Gamma$ is the database $\Xi$ obtained
  from~$G$ as follows:
  \begin{enumerate}
  \item For each node $u$ in $G$, create a fact $s_u \xrightarrow{a} t_u$
    in~$\Xi$, where $s_u$ and $t_u$ are fresh domain elements.
  \item For each edge $e = (u,v)$ in $G$, add to~$\Xi$ a %
    copy $D_e$ of $D$ defined by the following isomorphism:
      \begin{itemize}
        \item domain element $t_\ii$ is renamed to~$t_u$ and domain
          element $t_\oo$ is renamed to~$t_v$
        \item all other elements are renamed to fresh elements
          called the \emph{internal elements} of~$D_e$.
      \end{itemize}
  \end{enumerate}
\end{definition}

Let us observe the following property about the encoding, which is the reason why
we require that in-elements and out-elements never occur as the head of
facts in pre-gadgets (before completion).

\begin{claim}[Walks in encodings]
  \label{composability}
  Let $G$ be a directed graph, %
  let $\Gamma = (D, t_\ii, t_\oo, a)$ be a pre-gadget, and let $\Xi$ be the encoding of $G$ with $\Gamma$.
  Then
  no walk of $\Xi$ passes through two internal elements 
  belonging to different copies of $D$. 
\end{claim}
\begin{proof}
  Consider a walk in~$\Xi$ and consider the first element used in the walk which
  is an internal element of a copy $D_e$ of~$D$. The only elements shared
  between $D_e$ and the rest of~$\Xi$ are the images of the elements $t_\ii$ and
  $t_\oo$, i.e., elements of the form $t_v$ for $v$ a vertex of~$G$. However,
  by our requirement on~$D$,
  there is no fact in~$D$ having $t_\ii$ or $t_\oo$ as head. Hence, the
  remainder of the walk can never reach the images of~$t_\ii$ and~$t_\oo$, so it
  can never reach an internal element of another copy.
\end{proof}

For now, our definition of pre-gadgets is independent of the query. Let us now
formalize what we require from query matches in gadgets, in order to reduce
from the vertex cover problem. We do this by introducing the \emph{hypergraph of
matches}, 
generalizing the graph of matches from the proof of \cref{prp:aa},
along with a \emph{condensation} operation to simplify it.

\subsection{Condensed Hypergraph of Matches}
\label{subsec:condensed}
Let $L$ be a language and $D$ a database.
A \emph{match} 
(or witness)
of~$L$ on~$D$ is a set of edges
corresponding to an 
$L$-walk (as defined in \cref{sec:prelims}): formally, each $L$-walk $e_1,
\ldots, e_m$ defines a match which is the set
$\{e_1, \ldots, e_m\}$, noting that several $L$-walks may define the same
match. Then:

\begin{definition}[Hypergraph of matches]
  The \emph{hypergraph of matches} of~$D$ and~$L$, denoted $\calH_{L,D}$, 
  is the hypergraph whose nodes are the
  facts of~$D$ and whose hyperedges are the matches of~$L$ on~$D$.
\end{definition}

Recall that a \emph{hitting set} (also called \emph{hypergraph
transversal}~\cite{gottlob2004hypergraph})
of a hypergraph $\calH = (V,E)$ is
a subset $V'\subseteq V$ such that $V' \cap e \neq \emptyset$ for
every $e\in E$. 
Thus, hitting sets are a generalization of vertex covers to hypergraphs.
It is clear from the definition of resilience that $\resset(Q_L,D)$
is equal to the minimum size of a hitting set of
the hypergraph of matches $\calH_{L,D}$.  We now introduce two
rewriting rules, called the \emph{condensation rules}, that can be applied to any hypergraph $\calH = (V,E)$ to simplify it
without affecting the minimum size of hitting sets.
For $v\in V$, write $E(v) = \{e\in E \mid v \in e\}$. 
The condensation rules are:

\begin{description}
  \item[Edge-domination.] If there are $e,e' \in E$ with $e
    \neq e'$ and $e\subseteq e'$, 
    we can remove $e'$ to obtain $\calH' =
    (V,E\setminus \{e'\})$.  Intuitively, hitting sets must intersect $e$ so they 
    automatically intersect~$e'$.
  \item[Node-domination.] If there are $v,v' \in V$ with $v \neq v'$ such that 
    $E(v) \subseteq E(v')$, 
    we can remove $v$ to obtain $\calH'=
    (V\setminus \{v\}, \{e\setminus \{v\} \mid e \in E\})$. Intuitively, we can
    always replace $v$ by~$v'$ 
    in hitting sets.
\end{description}

We say that $\calH'$ is a \emph{condensed hypergraph} of
$\calH$ if $\calH'$ can be obtained from $\calH$ by 
applying some sequence of condensation rules
(note that we do \emph{not} require that $\calH'$ cannot be further simplified
using the rules).
It can be remarked that the condensation rules are \emph{confluent} in the sense
that we always get isomorphic hypergraphs if we apply them to an input
hypergraph until it is
no longer possible, no matter the order in which we apply them:
see \cite{amarilli2025confluence}. %
However, we will not need confluence for our purposes in the present paper.
One can easily see that the rules preserve the minimum size of a hitting set:

\begin{claim}[Condensation preserves minimum size of hitting sets]
  \label{simplifications}
  Let $\calH'$ be 
  a condensed hypergraph of~$\calH$.
  Then the minimum size of a hitting set of~$\calH$ and $\calH'$ are the same.
\end{claim}

\begin{proof}
  We show the claim by induction on the length of the sequence of applications
  of the condensation rules: the
  base case, corresponding to the empty sequence, is trivial. 
  
  For the induction
  case, let us first assume that  $\calH'$ is obtained from~$\calH = (V,
  E)$ by applying
  the edge-domination rule. We then immediately see 
 that any $V' \subseteq V$ is a hitting set of
    $\calH$ iff it is a hitting set of $\calH'$, so the minimal size of a
    hitting set in $\calH$ and~$\calH'$ is the same.

    Second, let us assume that $\calH' = (V', E')$ is obtained from~$\calH = (V,
    E)$ by
    applying the node-domination rule on two vertices $v$ and $v'$ with $E(v)
    \subseteq E(v')$.
    Write $V' = V \setminus \{v\}$ and $E' = \{e \setminus \{v\} \mid e \in
    E\}$.
    We show both directions of the claim. First, 
    any hitting set of $\calH'$ is also a hitting set of
    $\calH$, so the minimum size of a hitting set of~$\calH$ is at most the
    minimum
    size of a hitting set of~$\calH'$. Second, let us consider 
    a hitting set $S'$ of $\calH$. If
    $v\notin S'$ then $S'$ is a hitting set of $\calH'$. Otherwise, if $v
    \in S'$, then $(S' \setminus \{v\}) \cup \{v'\}$ is a
    hitting set of $\calH'$ of cardinality at most~$|S'|$. In both cases
    $\calH'$ has a hitting set of cardinality at most~$|S'|$, so that the
    minimum size of a hitting set of~$\calH'$ is at most the minimum size of a
    hitting set of of~$\calH$. Hence, the equality is proven, which concludes
    the induction step and finishes the proof.
\end{proof}

By \cref{simplifications}, we know that $\resset(Q_L,D)$ is equal to
the minimum size of a hitting set of %
any condensed hypergraph of~$\calH_{L,D}$, so we can always apply condensation rules in hardness proofs.
We can now define the requirement for a pre-gadget $\Gamma = (D, t_\ii, t_\oo,
a)$ to be a gadget:\footnote{\emph{Independent Join Paths} were introduced in \cite{DBLP:conf/pods/FreireGIM20} 
and subsequently given a semantic formulation and a sufficient hardness criterion in \cite{makhija2024unified}. Our gadgets adapt this approach to RPQs.
} we 
require that the hypergraph of matches $\calH_{L,D'}$ of~$L$ on the completion
$D'$ of~$\Gamma$ can be condensed into an \emph{odd path}. Formally,
let $\calH = (V, E)$ be a hypergraph and $u, v \in V$. We call $\calH$ an
  \emph{odd path from~$u$ to~$v$} if all hyperedges of~$\calH$ have size 2 and the
  corresponding graph is an odd path from~$u$ to~$v$; i.e.,
  there is a
sequence $u = w_1, \ldots, w_{2k} = v$ of pairwise distinct vertices such that 
  $E = \{\{w_i, w_{i+1}\} \mid 1 \leq i < 2k\}$.

\begin{definition}[Gadget]
  \label{def:gadget}

  A \emph{gadget}
  for language~$L$ is a pre-gadget
  $\Gamma = (D, t_\ii, t_\oo, a)$
  that satisfies the following condition:
  letting $D'$ be the completion of~$\Gamma$
  (see \cref{def:pregadget})
  where we add two $a$-facts $F_\ii$ and $F_\oo$ 
  from fresh elements to $t_\ii$
  and $t_\oo$ respectively, there exists a condensation of the hypergraph
  of matches $\calH_{L,D'}$ of~$L$ on~$D'$ which is an odd path from~$F_\ii$ to~$F_\oo$.
\end{definition}

\begin{example}
  The hardness pre-gadget whose completion is pictured in \cref{fig:aagadget} is
  a gadget for~$aa$, as can be seen on the hypergraph of matches
  of \cref{fig:aamatches} (with no condensation rules applied).
\end{example}

We claim that the existence of a gadget implies that the resilience problem is
hard, namely:

\begin{proposition}[Gadgets imply hardness]
  \label{prp:gadgethard}
  Let $L$ be an infix-free language. If there is a gadget~$\Gamma$ for~$L$, then
  $\resprobset(L)$ is NP-hard. %
\end{proposition}

All our NP-hardness results in the paper will be proved using
\cref{prp:gadgethard}, by building gadgets for the queries that we consider. We
note that we build the gadgets by hand, 
so it is possible that some queries may have shorter gadgets than the ones
that we exhibit.

We prove~\cref{prp:gadgethard} in the remaining of
\cref{subsec:condensed}, by reducing from
the vertex cover problem. Let $G'$ be an undirected graph, and let $G$ be an
arbitrary orientation of $G'$. We use the
gadget to encode the input graph. The idea is that the definitions that we impose on
gadgets are designed to ensure that the encoding of $G$ has a
condensed hypergraph of matches $\calH'$ which is an odd-length
subdivision of the original graph $G'$, so that by
\cref{prp:subdivision} and \cref{simplifications} computing a minimum vertex cover on~$G'$
reduces to computing a minimum hitting set of $\calH'$. Let us prove this formally.

  We call an \emph{odd subdivision} of an undirected graph an $\ell$-subdivision for
  some odd $\ell \in \mathbb{N}$.
  The technical claim used in the proof of \cref{prp:gadgethard} is the
following:

\begin{claim}
  \label{clm:gadgethypergraph}
  Let $G$ be a directed graph, let $\Gamma = (D, t_\ii, t_\oo, a)$ be a gadget for $L$, and let
  $\Xi$ be the encoding of~$G$ using~$\Gamma$. Then there exists a  condensed
  hypergraph of $\calH_{L,\Xi}$ that has only hyperedges of size $2$ and which is (when seen as an undirected graph) an odd subdivision of~$G'$, the undirected graph corresponding to $G$.
\end{claim}

\begin{proof}
  Thanks to \cref{composability}, one can check that the hypergraph
  of matches $\calH_{L,\Xi}$ is precisely the union of the
  hypergraphs of matches of the various copies $D_e$ of $D$, but where
  in- and out-elements have been merged as in
  \cref{coding}. Then, considering the sequence of condensation rules to be done to obtain the condensed
  hypergraph of matches of the completion $D'$ of $D$ that is an odd path as in
  \cref{def:gadget}, we can do this sequence of simplifications on
  every copy $D_e$ of the gadget $D$ (noticing that by
  \cref{def:gadget} the source and target facts are not affected by
  the condensation rules), thus indeed obtaining a condensed
  hypergraph of matches of $\calH_{L,\Xi}$ that is an odd subdivision of
  $G'$. 
\end{proof}

We can now prove \cref{prp:gadgethard} from \cref{clm:gadgethypergraph}:

\begin{proof}[Proof of \cref{prp:gadgethard}]
  We reduce from the vertex cover problem. Let $\Gamma$ be a gadget for
  $L$. Let $G'$ be an undirected graph, and pick an arbitrary
  orientation $G$ of $G'$. Construct the encoding $\Xi$ of $G$ by $\Gamma$:
  clearly this can be done in linear time. Use the oracle to
  compute the resilience of $\query_L$ on $\Xi$, written $r$, which is
  equal to the minimum size of a hitting set in $\calH_{L,\Xi}$. By
  \cref{clm:gadgethypergraph}, $\calH_{L,\Xi}$ has a condensed
  hypergraph~$\calH'$ that is an odd subdivision of $G'$ and moreover
  by \cref{simplifications} we have that the size of a minimum
  hitting set of $\calH'$ is also equal to $r$. We can then conclude
  using \cref{prp:subdivision} (noticing that the length of the
  odd subdivision is a constant that depends only on~$L$).
\end{proof}

Hence, from now on, all hardness results in the paper will be shown by
simply exhibiting hardness gadgets for the
languages in question. 

\myparagraph{Automatic verification and visualization of gadgets} 
We wrote an implementation~\cite{gadgetrepo} to
automatically verify some correctness properties of all the gadgets that we present in this paper.
Namely,
our implementation takes as input a description of a gadget along with a query.
The query can be either explicit (e.g., $aa$, or $axb|cxd$), or it can be
implicitly defined via requirements (e.g., words that are known to satisfy the
query, or words for which it is known that no infix satisfies the query).
The implementation then checks that
the completion of the gadget satisfies the requirements outlined in
\cref{def:pregadget}. 
This is used to construct the condensed hypergraph
of matches and check that it is indeed an odd path, to satisfy
\cref{def:gadget}.
Further, the code exhaustively enumerates the directed paths in
the gadget and checks which infixes of these paths are known to
satisfy the query.
This is needed to show if a gadget is indeed correct for all queries that satisfy an implicitly defined criteria. We use this method to show hardness of classes of queries in \cref{sec:fourlegged,sec:repeated}.
The implementation of~\cite{gadgetrepo} is also
used to prepare the gadget drawings and condensed hypergraph of
matches presented in the paper.

\subsection{Hardness for $axb|cxd$}

We illustrate the use of gadgets by 
building one for 
the language $axb|cxd$:

\begin{proposition}
  \label{prp:axbcxd}
  $\resprobset(axb|cxd)$ is NP-hard.
\end{proposition}

\begin{proof}
By \cref{prp:gadgethard}, it is enough to exhibit a gadget for
$axb|cxd$. We claim that the pre-gadget $\Gamma$ whose completion
$D'$ is depicted in \cref{fig:axbcxdgadget} is a gadget. First,
$\Gamma$ is clearly a pre-gadget, so we only need to show that it
is a gadget, i.e., that it satisfies \cref{def:gadget}. 

We show in \cref{fig:axbcxdallmatches} the hypergraph of matches
$\calH_{L,D'}$ of the gadget whose completion $D'$ is depicted in
  \cref{fig:axbcxdgadget}. We now explain why $\calH_{L,D'}$ can be
condensed to the hypergraph of \cref{fig:axbcxdmatches}, which is
an odd path.
In this case, it suffices to apply node-domination rules to
$\calH_{L,D'}$ to obtain $\calH'$ by eliminating all vertices
corresponding to facts that are not endpoint facts and occur in
only one hyperedge. For instance, we can first remove
$t_{out} \xrightarrow{x} 15$ because it is dominated by $15
\xrightarrow{b} 16$, then remove $12 \xrightarrow{x} 15 $ because
it is dominated by $13 \xrightarrow{a} 12$, and so on, until we
obtain $\calH'$. 

Thus, $\Gamma$ is indeed a gadget, which concludes the proof.
\end{proof}

\begin{figure}
  \begin{subfigure}{.25\linewidth}
    \centering
    \includegraphics[scale=0.4]{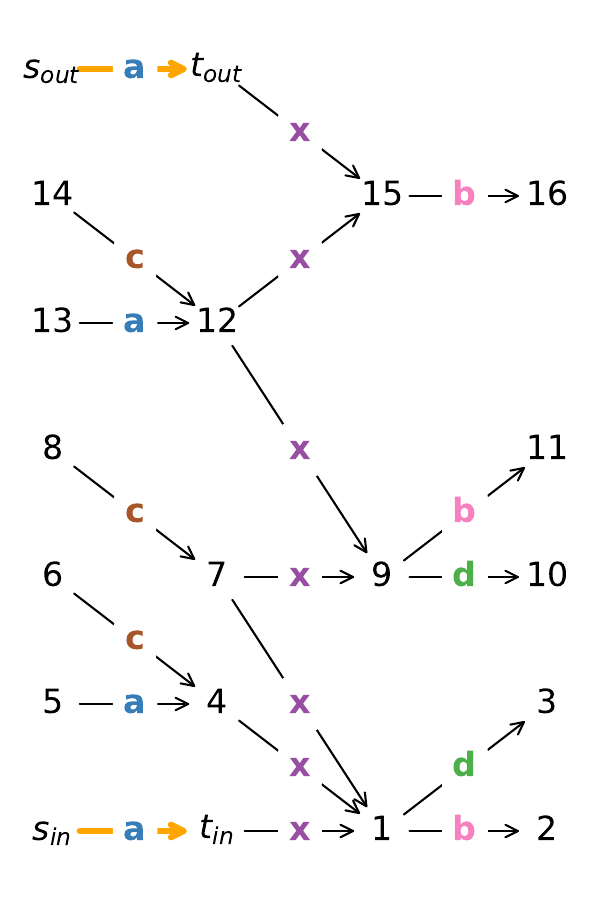}
    \caption{Completed gadget for $axb|cxd$
    }
    \label{fig:axbcxdgadget}
  \end{subfigure}
  \hfill
  \begin{subfigure}{.37\linewidth}
      \centering
    \includegraphics[scale=0.35]{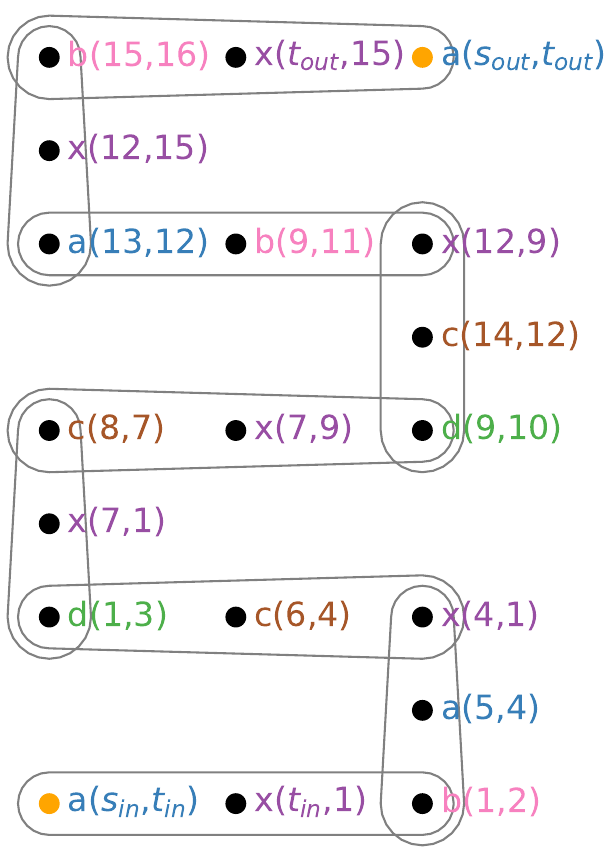} 
    \caption{Hypergraph of matches of $axb|cxd$ on the completed gadget in \cref{fig:axbcxdgadget}}
    \label{fig:axbcxdallmatches}
  \end{subfigure}
  \hfill
  \begin{subfigure}{.26\linewidth}
    \centering
  \includegraphics[scale=0.35]{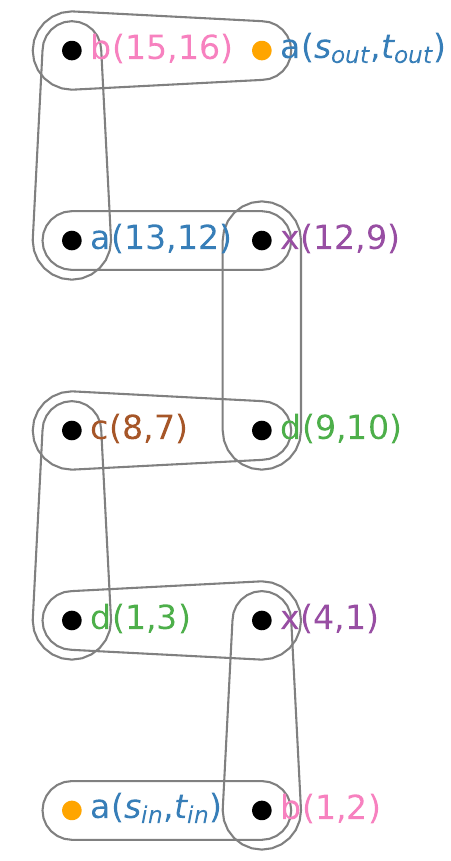}
    \caption{Condensed hypergraph of matches of \cref{fig:axbcxdallmatches}}
    \label{fig:axbcxdmatches}
  \end{subfigure}
  \caption{Illustration of hardness gadgets and corresponding condensed
  hypergraphs of matches for the proof of \cref{prp:axbcxd}.
  }
\end{figure}

\section{Hardness of Four-Legged Languages}
\label{sec:fourlegged}
In the previous section, we have presented 
hardness gadgets,
and 
shown the hardness of resilience
for the two specific languages $aa$ 
and
$axb|cxd$. 
In this section we generalize the latter result 
to a subclass of the 
infix-free non-local languages, called the \emph{four-legged languages}:
these 
languages intuitively feature a specific kind of counterexample to locality.
We show that every four-legged language admits a gadget, so that resilience for such
languages is always NP-hard by \cref{prp:gadgethard}.
We give two consequences of this result: it implies that resilience is hard for all
infix-free regular languages that are \emph{non-star-free}; and it also
implies a dichotomy on resilience for languages featuring a
\emph{neutral letter}.

\subsection{Four-Legged Languages}
\label{subsec:fourlegged}

Recall the notion of local languages (\cref{def:localdef}) from \cref{sec:epsro},
which hinges on recognizability
by certain automata
(namely, local DFAs, see \cref{def:local}), and which was equivalent to the notion of being
letter-Cartesian (\cref{prp:carac}).

If a language~$L$ is not letter-Cartesian, by definition there are
$x\in \Sigma$ and $\alpha,\beta,\gamma,\delta \in \Sigma^*$ such that
$\alpha x \beta \in L$ and $\gamma x \delta \in L$ but $\alpha x \delta \notin
L$.
We then define the \emph{four-legged languages} as those infix-free languages  that
feature such a violation 
in which none of the words $\alpha$, $\beta$, $\gamma$, and $\delta$ is empty. 
Formally:

\begin{definition}[Four-legged]
  A language $L$ is \emph{four-legged} if it is infix-free and
  has a letter
  $x \in \Sigma$ (the \emph{body}) and four non-empty words
  $\alpha, \beta, \gamma, \delta \in \Sigma^+$ (the \emph{legs}) with
  $\alpha x \beta \in L$ and $\gamma x \delta \in L$ but $\alpha x
  \delta \notin L$.
\end{definition}

\begin{example}
  \label{expl:four-legged}
  The languages $axb|cxd$ (from \cref{prp:axbcxd}) and $axb|cxd|cxb$ are non-local and
  four-legged. %
  The languages $aa$ (from \cref{prp:aa}) and
  $ab|bc$
  are non-local but
  not four-legged.
\end{example}

Note that four-legged languages are never local, because they are
not letter-Cartesian. The main result of this section is that $\resprobset(L)$ is always NP-hard when $L$ is four-legged;
formally:

\begin{theorem}[Hardness of four-legged languages]
  \label{prp:four-legged}
  If $L$ is four-legged then 
  $\resprobset(L)$ is NP-hard.
\end{theorem}

We prove \cref{prp:four-legged} in the remainder of \cref{subsec:fourlegged}.
To this end, we first show an intermediate lemma: for 
languages that are four-legged, we can assume without loss of
generality that the legs are \emph{stable}, i.e., no infix of $\alpha x \delta$
is in $L$:

\begin{definition}
  \label{def:stable}
  Let $L$ be a four-legged language, and let $x \in \Sigma$ be a body and let
  $\alpha, \beta, \gamma, \delta \in \Sigma^+$ be legs that witness that $L$ is four-legged, i.e., $\alpha x \beta \in L$
  and $\gamma x \delta \in L$ but  $\alpha x \delta \notin L$. We say that the
  legs $\alpha,\beta,\gamma, \delta$ are \emph{stable} if, in addition, there is
  no infix of $\alpha x \delta$ which is in~$L$.
\end{definition}

\begin{lemma}
  \label{lem:stable}
  Any infix-free language which is four-legged with body~$x \in \Sigma$ also has
  stable legs with body~$x$.
\end{lemma}
\begin{proof}
  Let $L$ be an infix-free language which is four-legged, and let $x\in \Sigma$
  be a body and $\alpha',\beta',\gamma',\delta' \in \Sigma^+$ be legs that
  witness it, i.e., $\alpha' x \beta' \in L$ and $\gamma' x \delta' \in L$ but
  $\eta' \coloneq  \alpha' x \delta' \notin L$.

  If $\alpha',\beta',\gamma',\delta'$ are already stable (i.e., no
  infix of $\eta'$ is in~$L$) then we are done.
  Otherwise, as $\eta' \notin L$, it means that there is a strict
  infix $\eta$ of $\eta'$ that is in $L$. We know that $\eta$ must have a non-trivial
  intersection with $\alpha'$, as otherwise $\eta \in L$ would be a strict infix of
  $\gamma' x \delta' \in L$, contradicting the fact that $L$ is infix-free.
  Similarly, $\eta$ has a non-trivial intersection with~$\delta'$, as otherwise
  $\eta \in L$ would be a strict infix of $\alpha' x \beta' \in L$.
  So let us write $\alpha' = \alpha_2
  \alpha_1$ and $\delta' = \delta_1 \delta_2$, with $\alpha_1$ and
  $\delta_1$ being non-empty, so that $\eta = \alpha_1 x \delta_1$ is a strict infix of
  $\alpha' x \delta'$ that is in $L$. Note that $\alpha_2$ and $\delta_2$ cannot
  both be empty, as otherwise we have $\eta = \eta'$ but we assumed that $\eta$ is a
  strict infix of~$\eta'$.

  We distinguish two cases:

  \begin{itemize}
    \item We have $\delta_2 \neq \epsilon$. Then we claim that
      $\alpha \coloneq  \gamma'$, $\beta \coloneq  \delta'$, $\gamma \coloneq  \alpha_1$
      and $\delta \coloneq \delta_1$ are stable legs. Indeed, observe
      first that none of $\alpha,\beta,\gamma,\delta$ is empty.
      Moreover we have $\alpha x \beta = \gamma' x \delta' \in L$
      and $\gamma x \delta = \alpha_1 x \delta_1 \in L$. Last, no
      infix of $\alpha x \delta = \gamma' x \delta_1$ can be in $L$
      because $\gamma' x \delta_1 \delta_2 = \gamma' x \delta'$ is
      in $L$ and $\delta_2 \neq \epsilon$ so any infix of $\alpha x
      \delta$ is a strict infix of a word of~$L$; but $L$ is
      infix-free so no strict infix of a word of~$L$ is a word of~$L$,
      hence no infix of $\alpha x \delta$ is a word of~$L$.
    \item We have $\alpha_2 \neq \epsilon$. Then we do the
      symmetric choice of taking $\alpha \coloneq  \alpha_1, \beta \coloneq 
      \delta_1$, $\gamma \coloneq  \alpha'$, and $\delta \coloneq  \beta'$ and
      claim that they are stable legs. Instead, none of $\alpha,
      \beta,\gamma, \delta$ is empty. Further, $\alpha x \beta =
      \alpha_1 x \delta_1$ and $\gamma x \delta = \alpha' x \beta'$
      are in~$L$. Moreover we have $\alpha x \delta = \alpha_1 x
      \beta'$ and all its infixes are infixes of $\alpha' x \beta'$
      which are strict because $\alpha_2 \neq \epsilon$; as
      $\alpha' x \beta' \in L$ and $L$ is infix-free, no infix of
      $\alpha x \delta$ is a word of~$L$.
  \end{itemize}

  Hence, in either case, we can define stable legs for~$L$ with body~$x$, which
  concludes the proof.
\end{proof}

\begin{figure}
  \begin{subfigure}[b]{0.45\linewidth}
    \centering
    \includegraphics[scale=0.4]{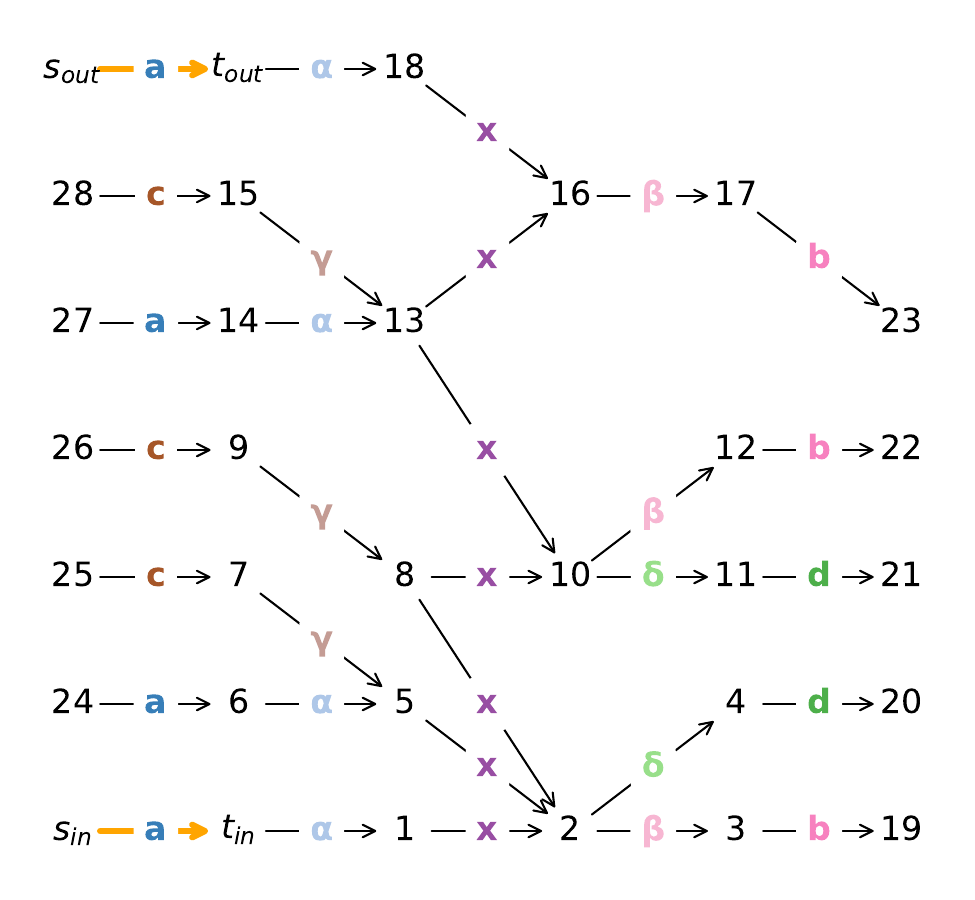}
  \end{subfigure}
  \begin{subfigure}[b]{0.4\linewidth}
    \centering
    \includegraphics[scale=0.4]{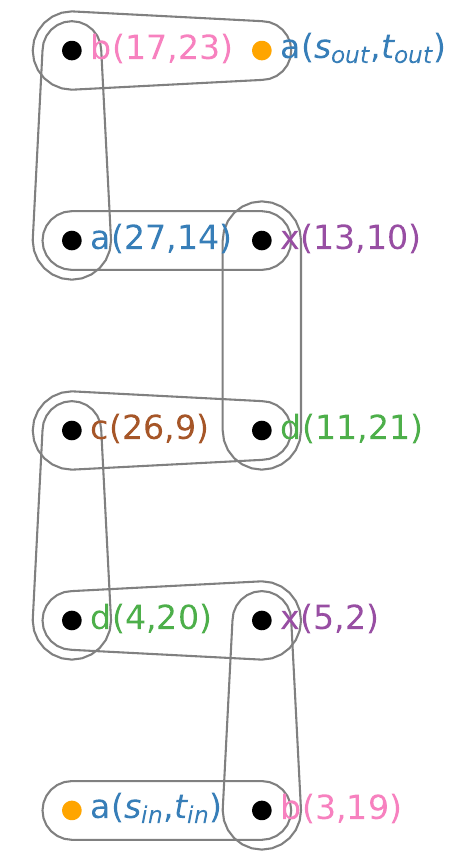}
  \end{subfigure}
  \caption{Completed gadget for Case 1 of \cref{prp:four-legged}.}
  \label{fig:case1}
\end{figure}

\begin{figure}
  \begin{subfigure}[b]{0.55\linewidth}
    \centering
    \includegraphics[scale=0.4]{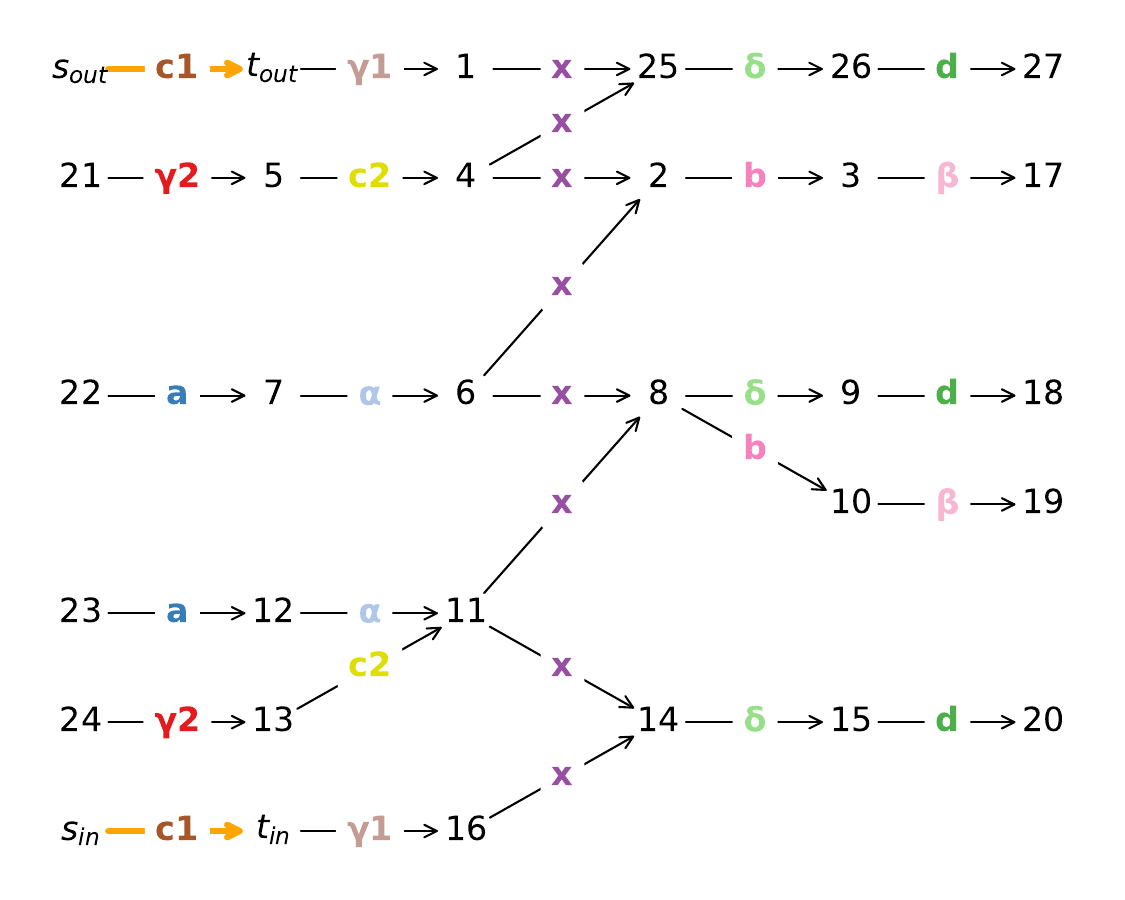}
  \end{subfigure}
  \begin{subfigure}[b]{0.25\linewidth}
    \centering
    \includegraphics[scale=0.4]{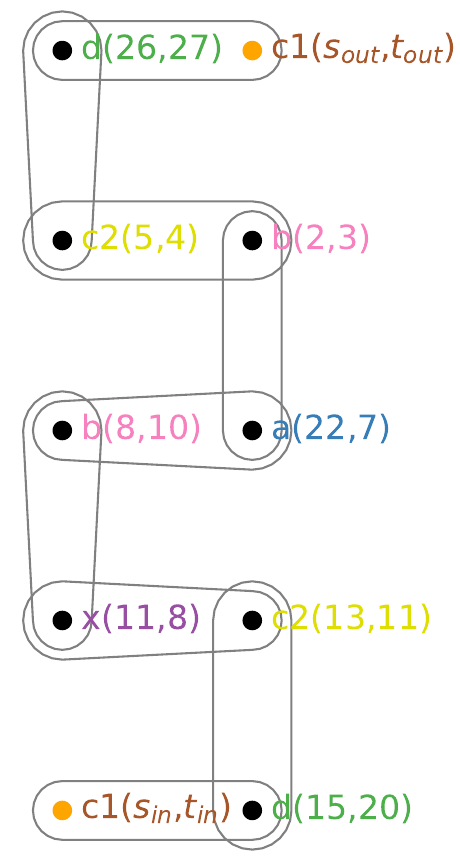}
  \end{subfigure}
  \caption{         Completed gadget for Case 2 of \cref{prp:four-legged}.}
  \label{fig:case2}
\end{figure}

We can now prove that resilience is hard for four-legged languages
(\cref{prp:four-legged}):

\begin{proof}[Proof of \cref{prp:four-legged}]
  Let $L$ be an infix-free and four-legged language. Using
  \cref{lem:stable}, let $x \in \Sigma$ be a body and
  $\alpha',\beta',\gamma',\delta' \in \Sigma^+$ be four stable
  legs, i.e., we have that $\alpha' x \beta' \in L$ and $\gamma' x
  \delta' \in L$ but no infix of $\alpha' x \delta'$ is in~$L$.
  We then distinguish two cases:

  \begin{description}
    \item[Case 1.] The first case is when, furthermore, no infix
      of $\gamma' x \beta'$ is in $L$. Then in this case, let us
      write $\alpha' = a \alpha$ (i.e., $a$ is the first letter of
      $\alpha'$, which is well defined since $\alpha' \neq
      \epsilon$) and likewise write $\beta' = \beta b$ and $\gamma'
      = c \gamma$ and $\delta' = \delta d$ (note that some of the
      letters $a$, $b$, $c$, and $d$ may be equal). In this case we
      claim that the (completed) gadget depicted in \cref{fig:case1} is a valid
      gadget.\footnote{We point out that this gadget is essentially the generalization of the gadget in \cref{fig:axbcxdgadget} for $L=axb|cxd$.} 
      Edges in the gadgets labeled with words $\alpha, \beta, \gamma, \delta \in \Sigma^*$ represent paths labeled with the letters in the words.
      If any of the words is empty, then the head node of the edge is merged with the tail node.
      Such a drawing allows us to show hardness of any language that can be obtained by replacing the $\alpha, \beta, \gamma, \delta$ with any words in $\Sigma^*$.
      This is possible since, for any language that meets the conditions of 
      case~1, a condensed hypergraph of matches can be obtained that does not
      contain 
      any vertices corresponding to a word (i.e., a Greek symbol). 
      We can then check exhaustively that all directed walks in the completed gadget are either accounted for by the condensed hypergraph of matches or are not in the language.
      Thus, we know that for any such language, a condensed hypergraph of matches will be an odd path as shown in \cref{fig:case1}, and hence the gadget is a valid gadget. 

\item[Case 2.] The second case is when there are some infixes of
  $\gamma' x \beta'$ that are in $L$. Observe that such infixes
  must intersect non-trivially both $\gamma'$ and $\beta'$. Indeed,
  infixes of~$\gamma' x \beta'$ which do not intersect $\gamma'$
  are strict infixes of $\alpha' x \beta' \in L$ so they are not
  in~$L$ because $L$ is infix-free, and likewise infixes of $\gamma' x
  \beta'$ which do not intersect $\beta'$ are strict infixes of
  $\gamma' x \delta' \in L$ so again they are not in~$L$. Let us
  write $\alpha' = a \alpha$ (i.e., $a$ is the first letter of
  $\alpha'$) and $\beta' = b \beta$ ($b$ is the first letter of
  $\beta'$) and $\gamma' = c_1 \gamma_1$ ($c_1$ is the first letter
  of $\gamma'$) and $\gamma' = \gamma_2 c_2$ ($c_2$ is the last
  letter of $\gamma'$) and $\delta' = \delta d$ ($d$ is the last
  letter of $\delta'$). 
  Then we claim that the (completed) gadget depicted in
  \cref{fig:case2} is a valid gadget. 
  Indeed, we can again check exhaustively all directed walks in the gadget and
      observe that the following holds for all of them: they
  (1) either contain or are subsumed by $\alpha' x \beta'$ or $\gamma' x
      \delta'$ matches; (2) are an infix of $\alpha' x \delta'$ (and hence can
      never be a match); or (3) contain $c_2 x b$ as an infix. 
  We notice that any infix of $\gamma' x \beta'$ that is a match must contain $c_2 x b$ as an infix, and we know that such a match exists. 
  For any such match, the other letters of the match can be eliminated in the condensed
      hypergraph of matches, as they are subsumed (via the edge-domination rule)
      by the $c_2 x b$: thus the
      node-domination rule 
      can remove the other letters of the $\gamma' x \beta'$-infix match.
  Thus, it suffices to consider $\alpha' x \beta'$, $\gamma' x \delta'$, and $c_2 x b$ as the only matches in the condensed hypergraph of matches.
  When we do so, we can see that there exists a condensed hypergraph of matches that is an odd path as shown in \cref{fig:case2}, and hence the gadget is valid. Note that some facts of the database are removed from the condensed hypergraph of matches but are necessary for it to be correct, e.g., those corresponding to $23 \xrightarrow{a} 12$ or $12 \xrightarrow{\alpha} 11$.
\end{description}
This concludes the proof of \cref{prp:four-legged}.
\end{proof}

Note that \cref{prp:four-legged} is not restricted to regular languages and 
applies to arbitrary languages (even when resilience is not in NP).

\subsection{Implications of Hardness for Four-Legged Languages}
\label{subsec:implications}

Next, we spell out two consequences of \cref{prp:four-legged}. The first one
concerns \emph{non-star-free languages}, the second concerns languages with a
\emph{neutral letter}.

\begin{toappendix}
  \label{apx:starfreered}
  We formally show the following result stated in the main text:

\begin{claim}
  \label{clm:starfreered}
  Let $L$ be a star-free regular language. Then $\red(L)$ is star-free.
\end{claim}

Of course, the converse statement is not true, e.g., $(aa)^*$ is
not star-free but its infix-free language $\epsilon$ is star-free.
Let us prove the result:

\begin{proof}[Proof of \cref{clm:starfreered}]
  We use an alternative characterization of star-free languages~\cite[Chapter VI]{mpri}:
  they are those 
  languages that can be constructed inductively by applying the operations of
  concatenation and the Boolean operations (including complementation) on base
  languages consisting of the empty language, the singleton language containing
  only the empty word, and the singleton languages formed of
  single-letter words for each letter of the alphabet~$\Sigma$.

  We then use the fact that, for any language $L$, we can express $\red(L)$ as
  follows:
  \[
    \red(L) = L \setminus (\Sigma^+ L \Sigma^* \cup \Sigma^* L \Sigma^+)
  \]
  Note that $\Sigma^*$ is the complement of the empty set, and $\Sigma^+$ is the
  concatenation of $\Sigma^*$ with the union of all letters of~$\Sigma$, so both
  are star-free languages. Hence, if $L$ is a star-free language, then thanks
  to the fact that star-free languages are closed under concatenation and
  Boolean operations, so is $\red(L)$, which concludes the proof.
\end{proof}

\end{toappendix}

\myparagraph{Hardness for non-star-free languages.}
We say that a regular language $L$ is
\emph{star-free} if there exists a $k>0$ such that, for every $\rho,
\sigma, \tau \in \Sigma^*$, for every $m \geq k$, we have $\rho \sigma^k \tau \in L$
iff $\rho \sigma^m \tau \in L$. Star-free languages 
have many equivalent characterizations (see, e.g., \cite[Chapter VI]{mpri}).
Further, if a
language~$L$ is star-free then so is $\red(L)$ (see \cref{apx:starfreered}). We can show
that star-freeness of the infix-free sublanguage is a necessary condition 
to enjoy tractable resilience, because non-star-free infix-free regular languages are four-legged:

\begin{lemma}[Four-leggedness of non-star-free languages]
  \label{lemma:not aperiodic are hard}
  Let $L$ be a regular language which is infix-free and non-star-free.
  Then $L$ is four-legged.
\end{lemma}
\begin{proof}
  Let $A$ be a DFA which recognizes $L$, and 
  let $k> 0$ be greater than the number of states of~$A$.
  From the fact that $L$ is not star-free, applying the definition to this
  choice of~$k$, we obtain
  $\rho, \sigma, \tau \in \Sigma^*$ and $m \geq k$ such that exactly one of $\rho \sigma^k \tau$
  and $\rho \sigma^m \tau$
  is in the language. 
  In particular this implies that $\sigma \neq \epsilon$; let us
  write $\sigma = x \sigma'$.

  As $k$ is greater than the number of states
  of~$A$, by the pigeonhole principle, there must be $0 \leq i <
  j \leq k$ such that the partial runs of $\rho \sigma^i$ and $\rho \sigma^j$ finish at
  the same state when read by~$A$. Let 
  $\omega \coloneq  j-i$. We know that we have $\rho
  \sigma^k \tau \in L$ iff $\rho \sigma^{k + \omega} \tau \in L$, because
  exchanging the prefixes $\rho \sigma^i$ and $\rho \sigma^j$ does not change the state that we reach.
  Iterating the argument, we get the following fact, dubbed (*): for any $p
  \geq 0$ we have $\rho \sigma^k \tau \in L$ iff $\rho \sigma^{k+p \omega} \tau \in L$. For similar
  reasons, we get the fact dubbed (**): for any $p \geq 0$ we have
  $\rho \sigma^m \tau \in L$ iff $\rho \sigma^{m+p \omega} \tau \in L$.

  Let $(\phi,\psi) = (k,m)$ if $\rho \sigma^k \tau \in L$ and $\rho \sigma^m
  \tau \notin L$, and
  $(\phi,\psi) = (m,k)$ if it is the opposite.
  Recall that $m, k > 0$ so 
  $\phi> 0$ and $\psi> 0$.
  Further, by (*) and (**) we know the
  following fact, dubbed~($\dagger$): we have $\rho
  \sigma^{\phi + p\omega} \tau \in L$ for all $p \geq 0$, and $\rho \sigma^{\psi +
  p\omega} \tau
  \notin L$ for all $p \geq 0$. 
  Let $p$ be large
  enough so that $\phi + p\omega - 1 > \psi$. Now let us take:
  \begin{itemize}
    \item $\alpha = \rho \sigma^{2\omega-1}$,
    \item $\beta = \sigma' \sigma^\phi \tau$,
    \item $\gamma = \rho \sigma^{\phi + p\omega - 1 - \psi}$ and
    \item $\delta = \sigma' \sigma^\psi \tau$.
  \end{itemize}
  Note that all of $\alpha, \beta, \gamma, \delta$ are
  non-empty because they all contain $\sigma \neq
  \epsilon$ at a power which is greater than
  0. We claim that they witness that $L$ is four-legged with body~$x$. Indeed:
  \begin{itemize}
    \item $\alpha x \beta = \rho \sigma^{2\omega-1} x \sigma' \sigma^\phi \tau =
      \rho \sigma^{2\omega+\phi} \tau$
    \item $\gamma x \delta = \rho \sigma^{\phi + p\omega - 1 - \psi} x \sigma'
      \sigma^\psi \tau =
      \rho \sigma^{\phi +
      p\omega} \tau$
    \item $\alpha x \delta = \rho \sigma^{2\omega-1} x \sigma'  \sigma^\psi
      \tau = \rho \sigma^{2\omega+\psi} \tau$.
  \end{itemize}
  So by ($\dagger$) we have that:
  \begin{itemize}
    \item $\alpha x \beta \in L$ because $\rho \sigma^\phi \tau \in L$
    \item $\gamma x \delta \in L$ because $\rho \sigma^\phi \tau \in L$
    \item $\alpha x \delta \notin L$ because $\rho \sigma^\psi \tau \notin L$
  \end{itemize}
  This establishes that $L$ is four-legged. As it is infix-free, we can conclude by
  \cref{prp:four-legged}.
\end{proof}

Using \cref{prp:four-legged}, we deduce:

\begin{corollary}[Hardness of non-star-free languages]
  \label{lemma:blabla}
  Let $L$ be a regular language. If $\red(L)$ is non-star-free, then $\resprobset(L)$ is NP-hard.
\end{corollary}

\myparagraph{Dichotomy for languages with a neutral letter.} 
A second consequence of \cref{prp:four-legged} concerns languages $L$ with a 
\emph{neutral} letter, i.e.,
a letter $e\in \Sigma$ such that inserting or deleting~$e$ anywhere 
does not change the membership of words 
to~$L$.
Formally, $e$ is neutral for~$L$ if for every $\alpha,\beta\in \Sigma^*$ we have
$\alpha\beta\in L$ iff $\alpha e \beta\in L$. 
(Note that the existence of a neutral letter for~$L$ does not mean that the
infix-free sublanguage $\red(L)$ also has a neutral letter, and in fact it is easy to
observe that non-empty
infix-free languages can never have a neutral letter.)
In formal language theory, one  often
assumes the existence of a neutral letter as a technical assumption on
languages~\cite{koucky2005bounded,barrington2005first}.
Under this assumption, we can show:

\begin{proposition}[Dichotomy for languages with neutral letters]
  \label{prp:dicho}
  Let $L$ be a language with a neutral letter. 
    If the language $\red(L)$ is local then $\resprobbag(L)$ is PTIME,
    otherwise $\resprobset(L)$
  is NP-hard.
\end{proposition}

Note that we can read this result as implying that, in the presence of a neutral letter, there
are not many tractable cases left. Indeed, intuitively, when considering a language
$L$ with neutral letter whose infix-free sublanguage is local, then membership
to~$L$ must be guided only by the first and last non-neutral letters and the set
of non-neutral letters that occur in the word.

We prove~\cref{prp:dicho} in the remainder of \cref{subsec:implications}. 
  To prove this result, we first show that, for languages featuring a neutral letter, whenever the
infix-free language is non-local then it must either be four-legged or contain a
word with a repeated letter, specifically a word of the form $aa$ for $a\in \Sigma$:

\begin{lemma}
\label{lem:neutral}
Let $L$ be a language with a neutral letter such that $\red(L)$ is not
local. Then one of the following two cases holds (possibly both):
\begin{itemize} 
  \item the language $\red(L)$ is four-legged; 
  \item there is $x\in \Sigma$ such that $xx \in \red(L)$.
\end{itemize}
\end{lemma}

To prove this lemma, let us prove the following intermediate result:

  \begin{claim}
    \label{clm:neutral}
    Let $L$ be a language with a neutral letter $e$ and let
    $L'=\red(L)$. Then for any non-empty words $\theta, \eta \in
    \Sigma^+$, if $\theta \eta \in L'$ then $\theta e \eta \in L'$.
  \end{claim}

  \begin{proof}
    We have $\theta \eta \in L'$ so $\theta \eta \in L$, and $e$ is
    neutral for~$L$ so we have $\theta e \eta \in L$. Let us
    proceed by contradiction and assume that $\theta e \eta \notin
    L'$. This means that there is a strict infix $\tau$ of $\theta
    e \eta$ which is in~$L$. 
    We first claim that $\tau$ must include the~$e$. Indeed, if it
    does not, then there are two cases:
        \begin{itemize}
        \item $\tau$ is an infix of $\theta$. Then it is a strict
          infix of $\theta \eta \in L'$, which contradicts the fact
          that $L'$ is infix-free, so this case is impossible.
          \item $\tau$ is an infix of $\eta$. This is impossible
            for the same reason as the previous case.
        \end{itemize}
    Now, we know that $\tau$ is a strict infix of $\theta e \eta$
    which includes the~$e$. Let $\tau'$ be obtained from~$\tau$ by
    removing~$e$. We have $\tau' \in L$ because $\tau \in L$
    and $e$ is neutral for~$L$. What is more, it is easily seen
    that $\tau'$ is an infix of $\theta \eta$, and it is a strict
    infix because $\tau$ includes the~$e$ so from $|\tau| < |\theta
    e \eta|$ (because $\tau$ is a strict infix) we deduce $|\tau'|
    = |\tau| - 1 < |\theta e \eta| - 1 = |\theta \eta|$. So $\tau'
    \in L$ is a strict infix of $\theta \eta \in L'$, which
    contradicts the assumption that~$L'$ is infix-free.
  \end{proof}

  We can now prove \cref{lem:neutral}:
  \begin{proof}[Proof of \cref{lem:neutral}]
  Let $L' = \red(L)$ and $e\in \Sigma$ a neutral letter of $L$. By
\cref{prp:carac},~$L'$ is not letter-Cartesian, so there is $x\in \Sigma$ and
$\alpha',\beta',\gamma',\delta'\in \Sigma^*$ such that $\alpha' x \beta' \in L'$ and $\gamma' x \delta' \in L'$ and $\alpha' x \delta' \notin L'$.

  Now, because $L'$ is infix-free, it is easy to check 
  that the following
  claim, dubbed ($\star$), must hold:
  \begin{equation}
    \tag{$\star$}
    \text{we have~} \alpha'
    \beta' \neq \epsilon \text{~and~} \gamma' \delta' \neq \epsilon \text{~and~}
    \alpha' \gamma' \neq \epsilon \text{~and~} \beta' \delta' \neq
  \epsilon.
\end{equation}

Indeed:

\begin{itemize}
  \item If $\alpha' \beta' = \epsilon$ then $\alpha' x \beta' \in L'$
    is equal to~$x$ which is an infix of $\gamma' x \delta' \in
    L'$, so by infix-freeness it must be the case that $\gamma'
    \delta' = \epsilon$, but now $x = \alpha' x \delta' = x$ and we
    assumed $ x= \alpha' x \delta' \notin L'$,
    which is a contradiction.
  \item If $\gamma' \delta' = \epsilon$, the reasoning is similar:
    in this case $\gamma' x \delta'$ is equal to  $x$, which is an
    infix of $\alpha' x \beta'$, so by infix-freeness we must have
    $\alpha' \beta' = \epsilon$, but then $x = \alpha' x \beta' \in
    L'$ and $x = \alpha' x \delta' \notin L'$ is a contradiction.
\end{itemize}
  Thus, we have established ($\star$).

We now start proving the statement of the lemma, namely, that either $L'$
is four-legged or that $L'$ contains $xx$.
Let us first show 
  that (call this $\dagger$) if both $\alpha'$ and $\delta'$ are empty then 
  the
  statement of the lemma holds.
  Indeed, in that case we have $x \beta' \in L'$ and $\gamma' x \in
  L'$. By ($\star$) (first two points) we have that $\beta'$ and $\gamma'$ are
    both non-empty. 
  By \cref{clm:neutral} we deduce from $x \beta' \in L'$ that $x e  \beta' \in
  L'$, 
  and
  we deduce from $\gamma' x \in L'$ that $\gamma' e x \in L'$.
  Now, take $e$ as a candidate choice of body, and $\alpha = x$, $\beta = \beta'$, $\gamma = \gamma'$ and $\delta =
x$, as a candidate choice of legs: all of these are non-empty.
  We have $\alpha e \beta \in L'$ and $\gamma e \delta \in L'$. So,
  considering the word $\alpha e \delta = xex$, there are two
  possibilities:
  either $xex \notin L'$ or $xex \in L'$. In the first case, 
  we conclude with our choice of body and legs that $L'$ is four-legged.
  In the second case,
  we have $xx \in L$ because $e$ is neutral for~$L$. Further, we have $xx \in L'$
  because the strict infixes of~$xx$ are the empty word and $x$ which are both
  strict infixes of the word $x \beta' \in L'$: 
  as $L'$ is
  infix-free we conclude that $xx \in L'$.

The same reasoning shows that (call this $\dagger'$) if both $\beta'$ and $\gamma'$ are empty then
the statement of the lemma holds.
Indeed, in this case we have $\alpha' x \in L'$ and $x \delta' \in L'$, and we
have $\alpha'$ and $\delta'$ both non-empty by ($\star$) (first two points). Applying
\cref{clm:neutral}  to each gives us that $\alpha' e x \in L'$ and $x e \delta'
\in L'$, and again we conclude that $L'$ is four-legged or that $xex \in L'$
from which we show $xx \in L'$ in the same way.

We have ruled out the case of $\alpha'$ and $\delta'$ being both empty, and of
$\beta'$ and $\gamma'$ being both empty.
Combining with $(\star)$, one gets that either $\alpha'$ and
$\gamma'$ are both non-empty, or $\beta'$ and $\delta'$ are both non-empty.
Indeed, assume that the first term in the disjunction is false, i.e., that it
is not the case that $\alpha'$ and $\gamma'$ are both non-empty.
If we have $\alpha' = \epsilon$, then $\delta'$ is non-empty from ($\dagger$), and $\beta'$ is non-empty by $(\star)$ (first point). 
If we have $\gamma' = \epsilon$, then $\beta'$ is non-empty by ($\dagger'$), and $\delta'$ is non-empty by $(\star)$ (second
point).

We finish the proof by distinguishing these two cases:
\begin{itemize}
  \item Assume that $\alpha' \neq \epsilon$ and $\gamma' \neq \epsilon$.
    We have $\alpha' x \beta' \in L'$ so by \cref{clm:neutral} we have
    $\alpha' e x \beta' \in L'$, and we have
    $\gamma' x \delta' \in L'$ so by \cref{clm:neutral} we have
    $\gamma' e x \delta' \in L'$.
    We then claim that $\alpha' e x \delta' \notin L'$, implying that $L'$
is four-legged (taking $e$ as the body, $\alpha= \alpha'$, $\beta = x \beta'$, $\gamma =
\gamma'$ and $\delta = x\delta'$, which are all not empty). Indeed, assume by
    way of contradiction that $\alpha' e x \delta' \in L'$. We would then have
$\alpha' e x \delta' \in L$, so $\alpha' x \delta' \in L$ (as $e$ is neutral).
    As $\alpha' x \delta' \notin L'$ by assumption, there
exists a strict infix $\tau$ of $\alpha' x \delta'$ that is in $L'$. Now,
    adding $e$ to $\tau$ would give a strict infix $\tau'$ of $\alpha' e x
    \delta'$, with $\tau' \in L$ because $\tau \in L$ and $e$ is neutral
    for~$L$, 
a contradiction because $L'$ is infix-free.
\item Assume that $\beta' \neq \epsilon$ and $\delta' \neq \epsilon$. The
  reasoning is similar: using \cref{clm:neutral} we have $\alpha' x e \beta' \in
    L'$ and $\gamma' x e \delta' \in L'$
  and we obtain that $L'$ is four-legged by taking $e$ as the body, $\alpha =
\alpha' x$, $\beta = \beta'$, $\gamma = \gamma' x$ and $\delta = \delta'$.
\end{itemize}
This concludes the proof.
\end{proof}

For completeness we give here examples of languages for which only one of
  the two cases in the statement of \cref{lem:neutral} holds. Let $L_1 = e^*be^*ce^* |
  e^*de^*fe^*$ (for which $e$ is neutral), with $\red(L_1) = be^*c | de^*f$.
  Then $\red(L_1)$ is four-legged
  (taking $x=e$, $\alpha = b$, $\beta = c$,
  $\gamma = d$, $\delta = f$), hence $\red(L_1)$ is not local, but it does not have a word of the
  form $xx$ for $x\in \Sigma$.  Now take 
  $L_2 = e^* (a|c) e^* (a|d) e^*$, 
  with $\red(L_2)= (a|c) e^* (a|d)$.
  Then $\red(L_2)$ is not local (taking $x=a$, $\alpha = c$, $\beta =
  \epsilon$, $\gamma = \epsilon$, $\delta = d$), it contains the word $aa$, but
  one can check that it is not four-legged.\\

We can then use \Cref{lem:neutral} to prove our complexity characterization
of resilience for languages featuring a neutral letter (\cref{prp:dicho}):

\begin{proof}[Proof of \cref{prp:dicho}]
If $\red(L)$ is local then $\resprobbag(L)$ is PTIME by \cref{prp:ro-ptime}.
Otherwise, by \cref{lem:neutral} either $\red(L)$ is four-legged, in which case
$\resprobset(L)$ is NP-hard by \cref{prp:four-legged},
  or $aa\in
\red(L)$ for $a\in \Sigma$.
  In this case one can easily adapt the proof of
  \cref{prp:aa}: one can check that the reduction only builds instances using
  the letter $a$, and the restriction of $\red(L)$ to words using only the
  letter~$a$ must be exactly $aa$ given that $\red(L)$ is infix-free. This
  concludes the proof.
\end{proof}

The dichotomy above relies on the presence of a neutral letter, so it does
not apply to all languages. In particular, it does not apply at all to finite
languages. This motivates further study of finite languages, and in fact all the
languages studied in the remainder of the paper will be finite languages.

\section{Hardness of Finite Languages Having Repeated Letters}
\label{sec:repeated}
We now conclude the presentation of our main hardness results by showing that
resilience is hard for a class of languages that generalizes $aa$
(\cref{prp:aa}). We say that a
word $\alpha$ has a \emph{repeated letter} if we have $\alpha = \beta a \gamma a
\delta$ for some $a \in
\Sigma$ and $\beta, \gamma, \delta \in \Sigma^*$.
We show that, for infix-free finite languages,
resilience is NP-hard whenever the language contains a word with a repeated letter:

\begin{theorem}[Hardness with repeated letters]
  \label{thm:repeated-letter}
  Let $L$ be an infix-free finite language that contains a word with a repeated
  letter. Then $\resprobset(L)$ is NP-complete.
\end{theorem}

Before proving~\cref{thm:repeated-letter}, we discuss its relation
with existing work and with the other results we have obtained so far.

\myparagraph{Connections to UCQs and self-joins.}
Observe that
finite languages $L$ correspond to the case of RPQs that are in fact unions of
conjunctive queries (UCQs). The UCQs obtained from RPQs have a specific form:
they are posed on an arity-two signature corresponding to the
alphabet~$\Sigma$, and their constituent conjunctive queries (CQs) are directed
paths corresponding to words of the language~$L$. 
The language being infix-free corresponds then to the UCQ being \emph{irredundant}.
In this sense, repeated
letters in words correspond to self-joins in CQs. Thus,
\cref{thm:repeated-letter} implies that, for those (irredundant) UCQs that correspond to finite (infix-free) RPQs,
the resilience problem is NP-hard unless every CQ is self-join-free.%
\footnote{This
contrasts with the setting of general UCQs, where resilience is sometimes
tractable even in the presence of
self-joins.
Specifically, in that setting,
the following Boolean CQ 
over a signature of arity at most two has tractable resilience under set semantics~\cite[Proposition~12]{DBLP:conf/pods/FreireGIM20}:
$Q \datarule A(x), R(x,y),
R(z,y), C(z)$. See also \cite[Appendix~C.1, Example~61]{arxiv:FreireGIM20}.}
Note that the converse does not hold, e.g.,
  the UCQ corresponding to the language $axb|cxd$ is hard, yet each of its
  constituent CQs is self-join-free. Thus our results do not imply a dichotomy on
  UCQs with directed-path constituent CQs, with some such queries being left unclassified by our work (e.g., the query corresponding to language $abc|bcd$). Our results however imply a dichotomy for
  \emph{single directed-path CQs}, because a singleton word without repeated
  letters is local and hence tractable, while a singleton word with a repeated
  letter is hard by Theorem~\ref{thm:repeated-letter}. Last, we discuss
  the recent work~\cite{bodirsky_et_al:LIPIcs.STACS.2026.15}
that establishes a dichotomy for UCQs that use a single binary relation: on the one
hand, these queries are more general than the unions of directed path CQs
(since CQs are not restricted to be directed paths), but on the other hand they 
do not allow multiple binary relations.

\myparagraph{Connections to other language classes.}
We now explain how 
\cref{thm:repeated-letter}
relates to the languages discussed so far. We first perform 
a sanity check: such languages are never local.

\begin{lemma}
  \label{lem:nonlocal}
  Let $L$ be a 
  finite language containing a word with a
  repeated letter. Then $L$ is not local.
\end{lemma}

\begin{proof}
  As $L$ is finite and contains a word with a repeated letter, we can take a word with a repeated letter whose length is
  maximum among words of $L$ with a repeated letter. 
  Let $\alpha = \beta a \gamma a \delta \in L$ be such a word.
  Assume by
  contradiction that $L$ is local.
  Then, by \cref{prp:carac}, $L$ is letter-Cartesian.
  Applying the definition to~$\alpha$ and~$\alpha$, we must have that the
  following word is also in~$L$:
  \[\beta a \gamma a \gamma a \delta \in L.\]
  This word is a word of~$L$ which contains a repeated letter and is strictly
  longer than~$\alpha$. 
  This contradicts the maximality of~$\alpha$, and
  concludes the proof.
\end{proof}

Equivalently, every local language containing a word with a repeated letter must be infinite.
Note that the finiteness hypothesis is essential for \cref{lem:nonlocal}, e.g., 
$L = ax^*b$
contains words with repeated letters but it is local (see
\cref{expl:locallang}). 
Further notice that there is no converse to \cref{lem:nonlocal}: 
there are finite infix-free non-local
languages that feature no words with repeated letters and for which resilience
is hard. This is for instance the
case of the four-legged language $axb|cxd$  of \cref{prp:axbcxd}.
Thus,
the languages covered by
\cref{thm:repeated-letter} and
 \cref{prp:four-legged}
are incomparable: they each generalize \cref{prp:aa} and
\cref{prp:axbcxd}, respectively, without implying the other result. (However, the proof of 
\cref{thm:repeated-letter} relies on \cref{prp:four-legged}.)
These two results also do not cover all intractable cases, as we will
see in the next sections.\\

We prove \cref{thm:repeated-letter} in the remainder of this section.
Because the proof is technical, we sketch it here
before giving the details.
We first make a preliminary observation: the tractability of resilience is
preserved under the 
\emph{reversal operation}
on regular languages. We then
start from a word featuring a repeated letter, and pick it to optimize certain
criteria (maximizing the gap between the repeated letters, and then maximizing
the total length of the word): we call this a \emph{maximal-gap} word.
Note that the
existence of such words is
the only place where we use the fact that $L$ is finite; indeed, our auxiliary hardness
results in the sequel will be stated in full generality without the assumption
that $L$ is finite.

We then consider such a maximal-gap word $\beta a \gamma a \delta$, and consider
various cases. 
We first deal separately with the case where $\beta$ and $\delta$ are both
non-empty, as we can then show rather easily that~$L$ is four-legged. Once this
case is dealt with,
we use the 
reversal operation
to assume without loss of
generality that it is $\beta$ which is empty.
We then consider the word $\gamma a \gamma$
and study whether it has infixes in~$L$. If it does not, we easily build a
gadget. If it does, we can now show that $\delta$ must be empty or otherwise the
language is four-legged. We then study the choice of infix $\alpha'$ of $\gamma a \gamma$
and distinguish two cases by considering the suffix of the left $\gamma$ and the
prefix of the right $\gamma$: either they overlap or they do not
overlap. We build
suitable gadgets for each case.

\myparagraph{The reversal operation.}

We first make a straightforward observation on the 
\emph{reversal} of
languages. The \emph{reversal}
of a word $\alpha = a_1 \cdots a_n \in \Sigma^*$ is
$\reversal{\alpha} = a_n
\cdots a_1 \in \Sigma^*$ (in particular
$\reversal{\epsilon} = \epsilon$).
The \emph{reversal}
of a language $L \subseteq \Sigma^*$ is $\reversal{L} \coloneq 
\{\reversal{\alpha} \mid \alpha \in L\}$.
Observe that $\reversal{(\reversal{L})} = L$.
It is immediate by symmetry that a language $L$ is infix-free iff
$\reversal{L}$ is
infix-free. Further, it is also immediate that a language and its 
reversal
have the same
complexity for resilience:

\begin{proposition}
  \label{prp:reversal}
  Let $L$ be any language. There are PTIME reductions between
  $\resprobset(\reversal{L})$ and $\resprobset(L)$, and between
  $\resprobbag(\reversal{L})$ and $\resprobbag(L)$.
\end{proposition}

The point of this claim will be to eliminate symmetric cases, up to replacing the
language by its 
reversal.
We only use the claim about set semantics in the sequel (given that we will show
hardness of resilience in set semantics, from which hardness for bag
semantics follows): the claim for bag semantics in \cref{prp:reversal} is given just for
completeness.

\begin{proof}[Proof of \cref{prp:reversal}]
  We give the proof for set semantics, the proof for bag semantics is the
  same mutatis mutandis.
  Given a database instance $D$ for $\resprobset(\reversal{L})$,
  simply create a database instance $\reversal{D}$ by reversing the direction of
  all arrows of~$D$. It is then obvious that a subset $D' \subseteq D$
  satisfies $\query_{\reversal{L}}$ iff the corresponding subset
  $\reversal{(D')}$
  of~$\reversal{D}$, which has the same cardinality, satisfies $\query_{L}$.
  Hence, the answer to $\resprobset(\reversal{L})$ on~$D$ and to $\resprobset(L)$ on
  $\reversal{D}$ are identical, giving the forward PTIME reduction. The
  converse reduction is obtained by applying the previous argument to
  reduce
  $\reversal{(\reversal{L})} = L$ to $\reversal{L}$.
\end{proof}

\myparagraph{Maximal-gap words}
Let $L$ be an infix-free finite language that contains a word with a repeated
letter, and let us show that $\resprobset(L)$ is NP-hard, to establish
\cref{thm:repeated-letter}.
We will pick a word with a repeated letter for the rest of the proof, but we
will choose such a word carefully. Here is
the precise notion we will need:

\begin{definition}
  \label{def:maximalgap}
  For $L$ a 
  finite language, we say that a word $\alpha = \beta a \gamma
  a \delta$ that features a repeated letter is a \emph{maximal-gap} word if:
  \begin{itemize}
    \item The gap between the repeated letters is maximal, i.e., there is no word $\alpha' = \beta' a' \gamma' a' \delta'$
in~$L$ such that $|\gamma'| > |\gamma|$
    \item Among such words, the length of~$\alpha$ is maximal, i.e., there is no word $\alpha' = \beta' a' \gamma' a' \delta'$
in~$L$ such that $|\gamma'| = |\gamma|$ and $|\alpha'| > |\alpha|$.
  \end{itemize}
\end{definition}

It is clear that any finite language that contains a word with a repeated
letter must have maximal-gap words. Hence, in the rest of the proof, we fix
$\alpha = \beta a \gamma a \delta$ to be an arbitrary such word. We will do case
disjunctions on this word in the rest of the proof.

\myparagraph{Reducing to the case $\beta = \epsilon$.}
Let us first assume that $\beta$ and $\delta$ are both non-empty. In this
case we conclude because $L$ is four-legged. Indeed:

\begin{claim}
  \label{clm:betadelta}
  Let $L'$ be an infix-free language containing a maximal-gap word
  $\beta a \gamma a \delta$ such that $\beta \neq \epsilon$ and $\delta \neq
  \epsilon$. Then $L'$ is four-legged.
\end{claim}

\begin{proof}
  Consider the two decompositions of the word $\alpha$ as $\alpha = (\beta a \gamma)
  a \delta$ and $\alpha = \beta a (\gamma a \delta)$. Note that $\beta a
  \gamma$, $\delta$, $\beta$, and $\gamma a \delta$ are all non-empty, so they
  can be used as legs.
  So consider the word $\alpha' = (\beta a \gamma) a (\gamma a \delta)$. The
  word~$\alpha'$
  contains two $a$'s separated by $\gamma a \gamma$, which is longer
  than~$\gamma$. Hence, $\alpha'$ is not 
  in~$L'$, as it would contradict the fact that $\alpha$ is maximal-gap (the
  first criterion of the definition of maximal-gap suffices).
  So we have a witness of the fact that $L'$ is four-legged.
\end{proof}

If $\beta$ and $\delta$ are both non-empty, then applying \cref{clm:betadelta} and \cref{prp:four-legged}
with $L' \coloneq  L$ concludes. 
Hence, in the rest of the argument, we assume that one of $\beta$ and
$\delta$ is empty. Further, without loss of generality, up to replacing $L$ by
$\reversal{L}$ using \cref{prp:reversal}, we assume that $\beta = \epsilon$.
Hence, in the rest of the proof, we
can assume that $\beta = \epsilon$: 
we have a maximal-gap word $\alpha = a \gamma a \delta \in L$ and this word is maximal-gap.

\myparagraph{Reducing to the case where some infix of $\gamma a \gamma$
is in~$L$.}
Continuing the proof, we will now consider the word $\gamma a \gamma$ and study
whether some infixes of this word are in~$L$. (Intuitively, when designing a
gadget for~$L$, such paths will appear, and so the construction of the gadget
will be different depending on whether these paths constitute query matches or
not.) Let us first assume that no infix of $\gamma a \gamma$ is in~$L$: this
case can easily be ruled out by designing a gadget for~$L$ based on the word
$\alpha = a \gamma a \delta$:

\begin{lemma}
  \label{lem:aatype}
  Let $L'$ be an infix-free language containing a word $a \gamma a \delta$ where
  no infix of~$\gamma a \gamma$ is in~$L'$. Then $\resprobset(L')$ is NP-hard.
\end{lemma}

\begin{figure}
  \null\hfill
  \begin{subfigure}[b]{0.4\linewidth}
    \centering
    \includegraphics[scale=0.4]{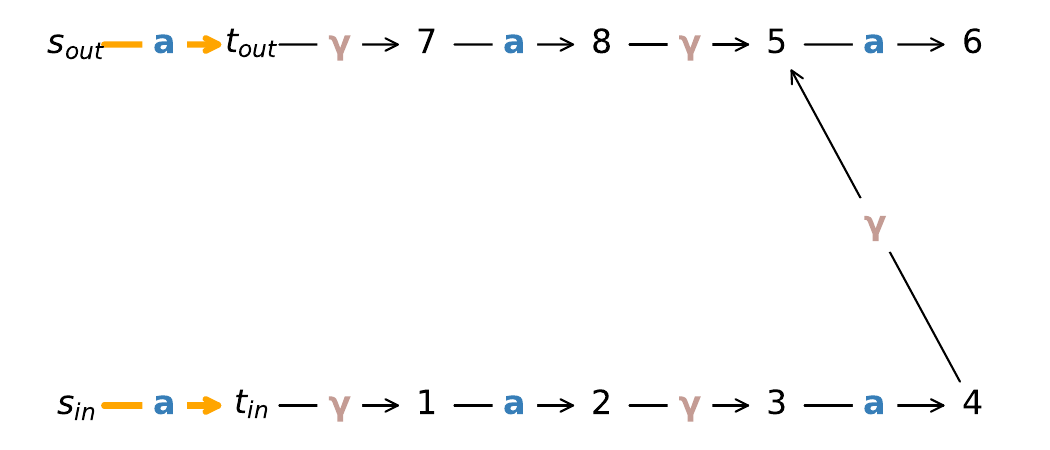}
  \end{subfigure}
  \hfill
  \begin{subfigure}[b]{0.4\linewidth}
    \centering
    \includegraphics[scale=0.4]{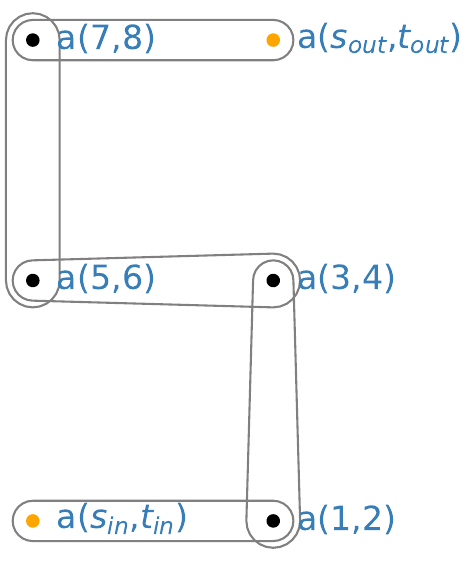}
  \end{subfigure}\hfill\null
  \caption{Completed gadget for the proof of \cref{lem:aatype} in the case
  when $\delta=\epsilon$ (left), and condensed hypergraph of matches (right).}
  \label{fig:aatype-nobeta}
\end{figure}
\begin{figure}
  \null\hfill
  \begin{subfigure}[b]{0.4\linewidth}
    \centering
    \includegraphics[scale=0.4]{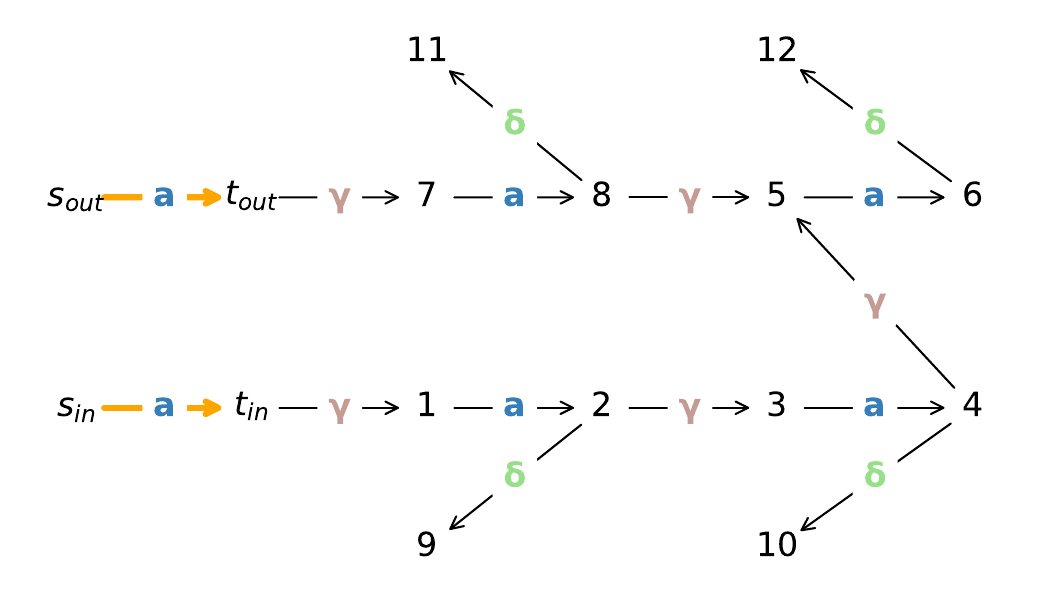}
  \end{subfigure}
  \hfill
  \begin{subfigure}[b]{0.4\linewidth}
    \centering
    \includegraphics[scale=0.4]{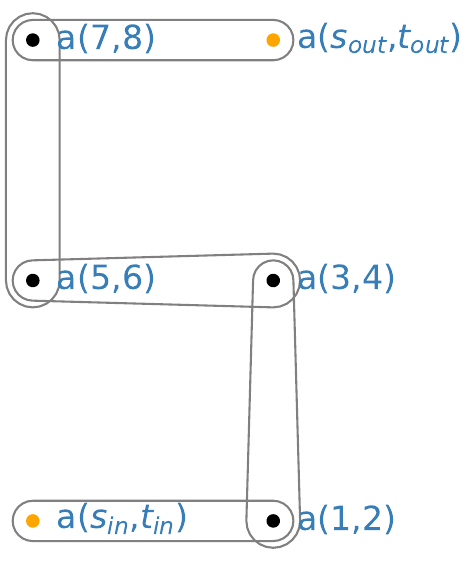}
  \end{subfigure}\hfill\null
  \vspace{-2mm}
  \caption{Completed gadget for \cref{lem:aatype} in the case where $\delta 
  \neq \epsilon$ (left), and condensed hypergraph of matches (right).
  }
  \label{fig:aatype-beta}
\end{figure}

\begin{proof}
  We first present the case where $\delta=\epsilon$, and then present the case
  where $\delta \neq
  \epsilon$.

  In the case where $\delta=\epsilon$, we use the completed gadget pictured on
  \cref{fig:aatype-nobeta} (left), where the edge labeled $\gamma$ stands for
  a (possibly empty) path.
  Note that, no matter what $\gamma$ stands for, there will be no facts having
  $t_u$ and $t_v$ as their head (except the two facts of the gadget
  completion, pictured in orange), so indeed it is a pre-gadget.
  The condensed hypergraph of matches is shown
  on \cref{fig:aatype-nobeta} (right), and it is an odd path, so indeed the
  instance pictured on \cref{fig:aatype-nobeta} (left) is (the completion of) a
  gadget and we can conclude by \cref{prp:gadgethard}.
  Note how we use
  the assumption that the matches on the completed gadget are the ones
  indicated: directed walks including two $a$'s must include $a\gamma a$ as an
  infix and so they are dominated, and directed walks including only one
  $a$ cannot be a match of the query. 

  In the case where $\delta\neq\epsilon$, we use the
  completed gadget pictured on \cref{fig:aatype-beta} (left), again with the
  edge labeled $\gamma$ standing for a possibly empty path. Again, no matter
  what $\gamma$ stands for, there will be no facts having $t_u$ and $t_v$ as
  head (except for the facts added in the completion), so again we have a
  pre-gadget. 
  The condensed
  hypergraph of matches is shown on \cref{fig:aatype-beta} (right). Now, to
  show that it is actually like this, first observe that the facts of the
  $\delta$-walks are
  always dominated by the preceding $a$: there is no match of the query within
  $\delta$ (as it is a strict infix of $a \gamma a \delta$), so all matches
  using $\delta$ must use the preceding $a$, and we can remove the facts of
  $\delta$-walks in the condensation via the node-domination rule. We are then
  essentially back to the preceding case. More precisely, again, 
  no directed walk containing
  only one copy of the $a$-edges of the gadget can be a word of the language,
  because these paths are either infixes of $\gamma a \delta$ (hence strict
  infixes of $a \gamma a \delta \in L$) or they are infixes of $\gamma a \gamma$
  and so they are not in the language by assumption. Further, for any two
  consecutive $a$'s in the gadget, there is a $a \gamma a \delta$ match
  containing them.
  What is more, after observing that $\delta$ facts are eliminated by 
  node-domination, we note that all query matches that are infixes of a word of the
  form $a (\gamma a)^*$ must contain two consecutive $a$'s, so they are subsumed
  by the matches corresponding to $a \gamma a \delta$, which after 
  node-domination contain just two consecutive $a$'s and the $\gamma$ between them.
  (The facts corresponding to $\gamma$ are then eliminated by node-domination.)

  Hence, we indeed have a condensed hypergraph of matches corresponding to
  \cref{fig:aatype-beta} (right), 
  and we can conclude by \cref{prp:gadgethard}.
\end{proof}

Applying this lemma with $L' \coloneq  L$ concludes.
So in the rest of the proof we can assume that there is an infix of $\gamma a
\gamma$ which is in~$L$.

\myparagraph{Reducing to the case $\delta = \epsilon$.}
The maximal-gap word $\alpha = a \gamma a \delta$, 
together with the infix
of $\gamma a \gamma$ which is in~$L$, now allows us to show that $L$ is
four-legged unless $\delta = \epsilon$. To show this, we need to discuss the
structure of the infix of~$\gamma a \gamma$:

\begin{claim}
  \label{clm:cacinfix}
  Let $L'$ be an infix-free language containing the word $\alpha = a \gamma a
  \delta$
  and containing an infix of $\gamma a \gamma$. Then $L'$ contains a word of the
  form $\alpha' = \gamma_1 a \gamma_2$ where $\gamma_1$ is a non-empty suffix of~$\gamma$
  and $\gamma_2$ is a non-empty prefix of~$\gamma$.
\end{claim}

\begin{proof}
  Consider the infix of $\gamma a \gamma$ which exists by hypothesis.
This infix must include the middle $a$, as otherwise
it is an infix of~$\gamma$, which is a strict infix of~$\alpha \in L'$, and 
this contradicts the assumption that $L'$ is infix-free. 

So write the infix as
$\alpha' = \gamma_1 a \gamma_2$ with $\gamma_1$ being a (not necessarily strict) suffix
of~$\gamma$ and $\gamma_2$ being a (not necessarily strict) prefix of~$\gamma$.
All that remains is to show that $\gamma_1$ and $\gamma_2$ are non-empty.

  First, we must have $\gamma_1 \neq \epsilon$,
  otherwise $\alpha'$ is an infix of $a
\gamma$ which is a strict infix of $a \gamma a \delta \in L'$, and this contradicts
the assumption that $L'$ is infix-free.
  
  Second, we must have $\gamma_2 \neq \epsilon$
otherwise $\alpha'$ is an infix of $\gamma a$ which is a strict infix of
$a \gamma a \delta$ and we conclude in the same way.
\end{proof}

Hence, applying \cref{clm:cacinfix} to $L' \coloneq  L$, we know that $L$
contains a word $\alpha' = \gamma_1 a
\gamma_2$ with $\gamma_1$ and $\gamma_2$ being respectively a non-empty suffix
of~$\gamma$ and a non-empty prefix of~$\gamma$.
We now claim that, if $\delta \neq \epsilon$, then $L$ is four-legged, allowing
us to conclude:

\begin{claim}
  \label{clm:beta}
  Let $L'$ be an infix-free language containing a word
  $a \gamma a \delta$ such that $\delta \neq \epsilon$ and a word $\gamma_1 a
  \gamma_2$ such that $\gamma_1 \neq \epsilon$, $\gamma_2 \neq \epsilon$,
  and $\gamma_1$ is a suffix of~$\gamma$.
  Then $L'$ is four-legged.
\end{claim}

\begin{proof}
  Indeed, consider
  $\alpha' = \gamma_1 a \gamma_2$ and
  $\alpha = (a \gamma) a \delta$. The words $\delta$, $a \gamma$, $\gamma_1$, and $\gamma_2$
are all non-empty, so they can be used as legs.
  So consider the word $\alpha'' = \gamma_1 a \delta$
  and show that $\alpha'' \notin L'$, so that $L'$ is four-legged.

  Indeed, the word $\alpha''$ is an infix of $\gamma a \delta$ (because
  $\gamma_1$ is a suffix of~$\gamma$), which is a strict infix
  of $a \gamma a \delta$ which is in~$L'$. Thus, as $L'$ is infix-free, $\alpha'' \notin
  L'$, so we have a witness of the fact that $L'$ is four-legged.
\end{proof}

Applying \cref{clm:beta} with $L' \coloneq  L$, we can conclude unless $\delta =
\epsilon$, which we assume in the rest of the proof.

\myparagraph{Distinguishing two cases depending on whether $\gamma_1$ and
$\gamma_2$ overlap.}
To summarize,
we now know that $L$ contains a maximal-gap word $\alpha = a \gamma a
\in L$ and a word $\alpha' = \gamma_1 a \gamma_2 \in L$ for $\gamma_1$ a non-empty suffix of
$\gamma$ and $\gamma_2$ a non-empty prefix of~$\gamma$. We distinguish two
cases, intuitively based on whether $\gamma_1$ and $\gamma_2$ cover $\gamma$ in
a way that overlaps or not:

\begin{description}
  \item[Non-overlapping case.]
    The non-overlapping case is when $|\gamma_1| + |\gamma_2| \leq |\gamma|$,
    so that we can write $\gamma
    = \gamma_2 \eta \gamma_1$ for some $\eta \in \Sigma^*$.
  \item[Overlapping case.]
    The overlapping case is when we have $|\gamma_1| +
    |\gamma_2| > |\gamma|$, so that we can write $\gamma = \eta'' \eta \eta'$
    such
    that $\gamma_1 = \eta \eta'$ and $\gamma_2 = \eta'' \eta$, with 
    the ``overlap'' $\eta$ being non-empty.
\end{description}

\myparagraph{Showing hardness in the overlapping case.}
Let us first deal with the overlapping case. We
have $\alpha' = \gamma_1 a \gamma_2 \in L$, meaning $\alpha' = \eta \eta' a \eta'' \eta
\in L$, and we know that the ``overlap'' $\eta$ is non-empty. We then show:

\begin{claim}
  \label{clm:etaeta}
  If the words $\eta'$ and $\eta''$ are not both empty, then $L$ is four-legged.
\end{claim}

\begin{proof}
  Let us first assume that $\eta'$ is non-empty and show that $L$ is
  four-legged in this case. Let us write $\eta' = x \sigma$.
  Substituting this into $\alpha'$ and $\alpha$,
  we have $\alpha' = \eta x (\sigma a \eta'' \eta) \in L$ and
  $\alpha = (a \eta'' \eta) x (\sigma a) \in L$. Now, taking $x$ as body, let
  us consider the cross-product word $\alpha'' = \eta x \sigma a = \eta \eta'
  a$. The
  word $\alpha''$ is a  strict infix of $\alpha = a \eta'' \eta \eta' a
  \in L$. Hence, as $L$ is infix-free, $\alpha'' \notin L$, so that we can conclude
  that $L$ is four-legged unless $\eta'$ is empty.

  Let us substitute $\eta' = \epsilon$ in $\alpha$ and $\alpha'$.
  We know that $\alpha' = \eta a \eta'' \eta \in L$ and $\alpha = a \eta'' \eta a \in
  L$. Let us now assume that $\eta''$ is non-empty and show that $L$ is
  four-legged in that case as well. Let us write $\eta'' = x \sigma$. We have 
  $\alpha = a x (\sigma \eta a) \in L$ and $\alpha' = (\eta a) x (\sigma \eta)
  \in L$. Now, taking $x$ as body again, let us consider the cross-product word
  $\alpha'' = a x \sigma \eta = a \eta'' \eta$. The word $\alpha''$ is a strict
  infix of $\alpha = a \eta'' \eta a
  \in L$. Hence, as $L$ is infix-free, $\alpha'' \notin L$, so we can conclude
  again that $L$ is four-legged unless $\eta''$ is empty.

  We conclude the claim: $L$ is four-legged unless $\eta' = \eta'' = \epsilon$.
\end{proof}

From \cref{clm:etaeta}, we can conclude unless $\eta' = \eta'' = \epsilon$, so
we assume this in the rest of the proof of hardness for the overlapping case. 
We now know that $\alpha' = \eta a \eta \in L$ and $\alpha = a \eta a \in
  L$.

  Let us now re-use the fact that $\alpha$ was chosen to be a maximal-gap word
  of~$L$. Now, as $\eta$ is non-empty, observe that $\alpha' = \eta a \eta$
  contains repeated letters across the two copies of~$\eta$. These letters are
  spaced at least by the length of~$\eta$ (and possibly more, if $\eta$ contains
  repeated letters -- even though this cannot happen by the maximal-gap
  condition).
  In any case, $\alpha'$ achieves a gap between two repeated letters which is at
  least that of~$\alpha$: so $\alpha$ and $\alpha'$ are tied for the first
  criterion in the definition of maximal-gap. Hence, by the second criterion,
  since $\alpha$ is maximal-gap, we know that $\alpha'$ cannot be strictly
  longer than~$\alpha$. So from $|\alpha'| \leq |\alpha|$, as $|\alpha'| =
  2|\eta|+1$ and $|\alpha| = |\eta|+2$, we conclude $|\eta| \leq 1$. As $\eta$
  is non-empty, we have $|\eta| = 1$. This letter may be different from~$a$ or
  not. So we can conclude the overlapping case by showing two hardness results:

  \begin{claim}
    \label{clm:ababab}
    Let $L'$ be any infix-free language containing $aba$ and $bab$. Then
    $\resprobset(L')$ is hard.
  \end{claim}

  \begin{figure}
    \null\hfill
    \begin{subfigure}[b]{0.4\linewidth}
      \centering
      \includegraphics[scale=0.4]{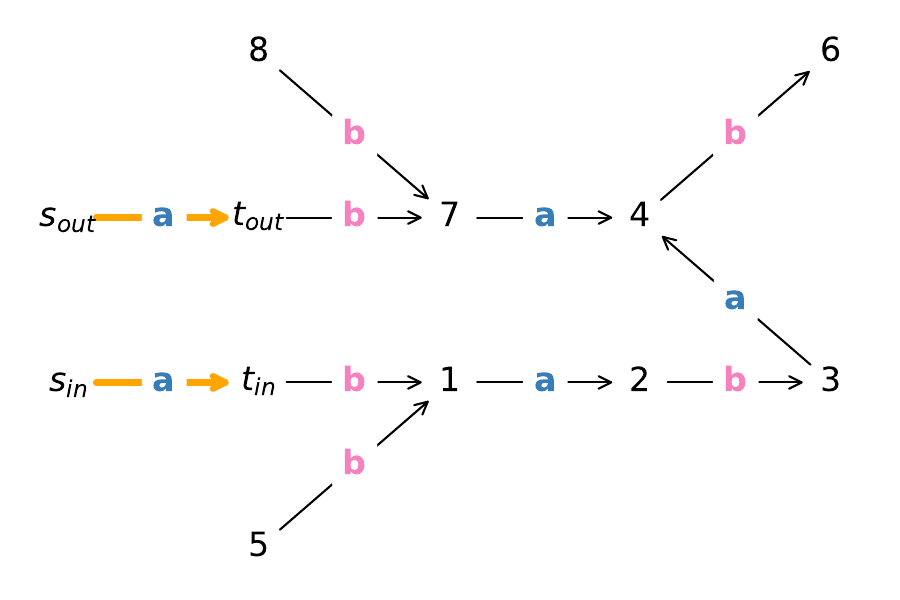}
      \vspace{-3mm}
    \end{subfigure}\hfill
    \begin{subfigure}[b]{0.4\linewidth}
      \centering
      \includegraphics[scale=0.4]{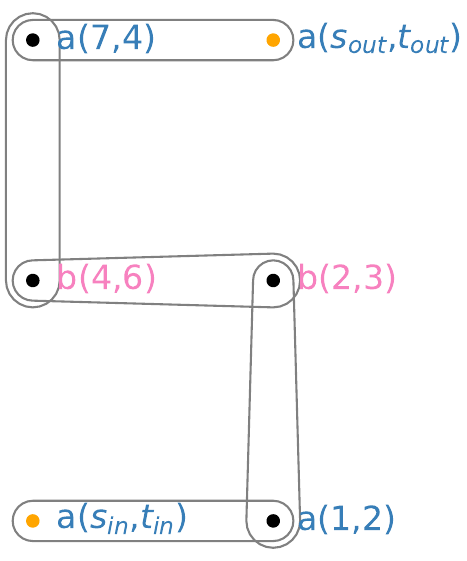}
    \end{subfigure}    \hfill\null
    \vspace{-3mm}
    \caption{Completed gadget used in the proof of \cref{clm:ababab}}
    \label{fig:ababab}
  \end{figure}

  \begin{proof}
    We use the gadget whose completion is pictured on \cref{fig:ababab} (left).
    It is obviously a pre-gadget, and one can check that the hypergraph of
    matches can be reduced to the odd path of \cref{fig:ababab} (right), so
    that it is indeed a gadget and we can conclude by \cref{prp:gadgethard}.
    More precisely, observe that there are three kinds of directed walks in the
    gadget:
    \begin{itemize}
      \item Paths of length 2 or less: all of these form a word which is a
        strict infix of $aba$ or $bab$, so they are not in the language by
        infix-freeness.
      \item Paths of length 3: all of these form the word $aba$ or $bab$, so
        they are in the language.
      \item Paths of length 4 or more: all of these have a path of length 3 as
        infix, and these always form words of the language by the previous
        point, so they are all subsumed by edge-domination.
    \end{itemize}
    Hence, the hypergraph of matches of the language $L'$ on the completion of
    the gadget corresponds precisely to the hyperedges given by the paths $aba$
    and $bab$. One can then obtain the condensation of \cref{fig:ababab} (right)
    by observing that $5 \xrightarrow{b} 1$ is subsumed by the next $a$-fact ($1
    \xrightarrow{a} 2$) and that $8 \xrightarrow{b} 7$ is subsumed by the next $a$-fact ($7
    \xrightarrow{a} 4$). The matches $\{1 \xrightarrow{a} 2, 2 
    \xrightarrow{b} 3\}$ and $\{7 \xrightarrow{a} 4, 4 
    \xrightarrow{b} 6\}$ then eliminate by edge-domination the matches $\{t_u \xrightarrow{b} 1, 1
    \xrightarrow{a} 2, 2 \xrightarrow{b} 3\}$ and $\{1 \xrightarrow{a} 2, 2 
    \xrightarrow{b} 3, 3 \xrightarrow{a} 4\}$ and $\{t_v \xrightarrow{b} 7, 7 \xrightarrow{a} 4, 4 
    \xrightarrow{b} 6\}$, and we can then see that $3 \xrightarrow{a} 4$ is
    eliminated by node-domination (e.g., by $2 \xrightarrow{b} 3$) and that $t_u
    \xrightarrow{b} 1$ is eliminated by node-domination with  $1 \xrightarrow{a}
    2$ and $t_v \xrightarrow{b} 7$ is eliminated by node-domination with $7
    \xrightarrow{a} 4$, leading to the pictured hypergraph. This allows us to
    conclude by \cref{prp:gadgethard}.
  \end{proof}
  
  \begin{claim}
    \label{clm:aaa}
    Let $L$ be any infix-free language containing the word $aaa$. Then
    $\resprobset(L)$ is NP-hard.
  \end{claim}

  \begin{figure}
    \null\hfill
    \begin{subfigure}[b]{0.4\linewidth}
      \centering
      \includegraphics[scale=0.4]{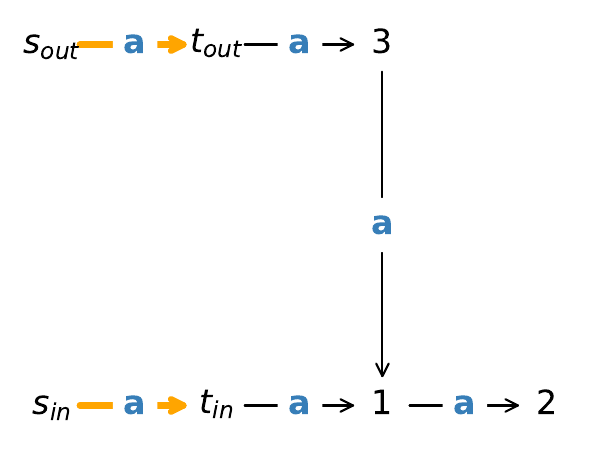}
    \end{subfigure}
    \hfill
    \begin{subfigure}[b]{0.4\linewidth}
      \centering
      \includegraphics[scale=0.4]{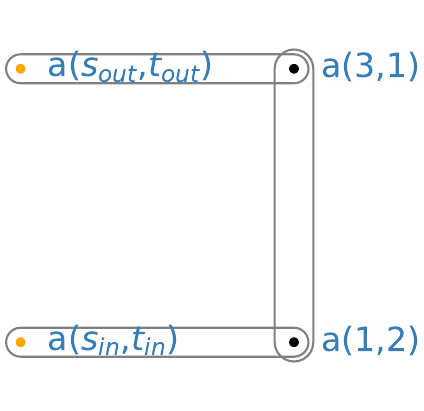}
    \end{subfigure}    
    \hfill\null
    \vspace{-3mm}
    \caption{Completed gadget used in the proof of \cref{clm:aaa}}
    \label{fig:aaa}
  \end{figure}

  \begin{proof}
    We use the gadget whose completion is pictured on \cref{fig:aaa} (left)
    (which incidentally turns out to be identical to that of \cref{fig:aagadget} in
    \cref{prp:aa}). This is indeed a pre-gadget, and the hypergraph of matches
    can be condensed to \cref{fig:aaa} (right), so it is indeed a gadget. Note
    that we use the fact that all directed walks in the gadget form words of the
    form $a^*$: by infix-freeness of the language, the matches precisely correspond
    to the words of length~3.
  \end{proof}

\myparagraph{Showing hardness in the non-overlapping case.}
We conclude the proof of \cref{thm:repeated-letter} by covering the
non-overlapping case. Recall that, in this case, we know that the language $L$
contains a maximal-gap word 
$\alpha = a \gamma_2 \eta \gamma_1 a \in L$ and a word $\alpha' = \gamma_1 a
\gamma_2 \in L$, and recall that we already know that $\gamma_1$ and $\gamma_2$
are both non-empty.

We show:

\begin{claim}
  \label{clm:gammagamma}
  The language $L$ is four-legged unless $\gamma_1$ and $\gamma_2$ both have
  length~1.
\end{claim}

\begin{proof}

  As $\gamma_1$ is non-empty, let us write $\gamma_1 = \chi' x$.
  If $\gamma_1$ has length 2 or more, then $\chi'$ is non-empty, and we have
  $\alpha' = \chi' x (a \gamma_2) \in L$ and
  $\alpha = (a \gamma_2 \eta \chi') x a \in L$.
  Let us take $x$ as body and consider the cross-product word $\alpha'' = \chi'
  x a = \gamma_1 a$. The word $\alpha''$ is a strict infix of
  $\alpha = a \gamma_2 \eta \gamma_1 a
  \in L$. So, as $L$ is infix-free, we have $\alpha'' \notin L$, witnessing that
  $L$ is four-legged. Hence, it suffices to focus on the case where
  $\gamma_1 = x$ has length 1.

  Now, as $\gamma_2$ is non-empty, let us write $\gamma_2 = y \chi''$.
  If $\gamma_2$ has length 2 or more, then $\chi''$ is non-empty, and we have
  $\alpha = a y (\chi'' \eta \gamma_1 a) \in L$ and $\alpha' = (\gamma_1 a)
  y \chi'' \in L$.
  Let us take $y$ as body and consider the cross-product word $\alpha'' = a y
  \chi'' = a \gamma_2$. The word $\alpha''$ is a strict infix of~$\alpha \in L$,
  so as $L$ is infix-free we know that $\alpha'' \notin L$, witnessing that~$L$ is
  four-legged.

  Hence, we conclude that $L$ is four-legged unless $\gamma_1$ and $\gamma_2$
  both have length 1.
\end{proof}

  Thanks to \cref{clm:gammagamma}, we know that $\alpha = a x \eta y a \in L$
  and $\alpha' = y a x \in L$, for some letters $x, y \in \Sigma$ and some word
  $\eta \in \Sigma^*$. 
  So all that remains is to show:

  \begin{claim}
    \label{clm:abcacab}
    Let $L'$ be an infix-free language that contains words $a x \eta y a$ 
    and $y a x$
    for some $a, x, y \in \Sigma$ and $\eta
    \in \Sigma^*$. Then $\resprobset(L')$ is NP-hard.
  \end{claim}

  \begin{figure}
    \null\hfill
    \begin{subfigure}[b]{0.4\linewidth}
      \centering
      \includegraphics[scale=0.4]{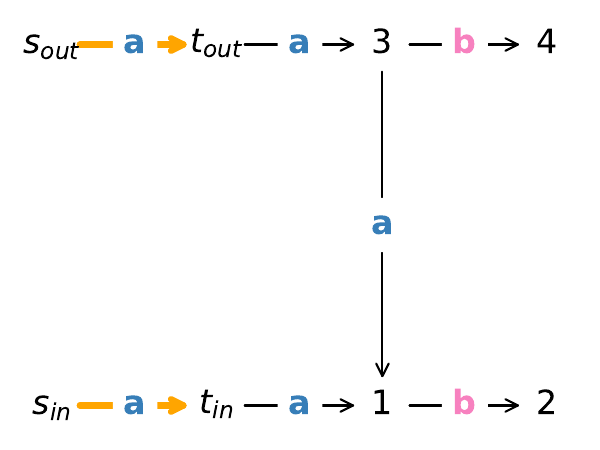}
    \end{subfigure}\hfill
    \begin{subfigure}[b]{0.4\linewidth}
      \centering
      \includegraphics[scale=0.4]{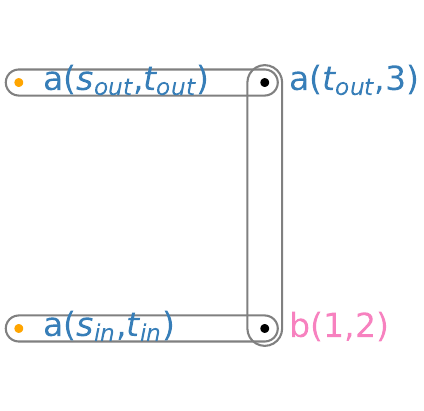}
    \end{subfigure}    \hfill\null
    \vspace{-3mm}
    \caption{Completed gadget used in the proof of \cref{clm:aab}}
    \label{fig:aab}
  \end{figure}

  To show this claim, we make an auxiliary claim:

  \begin{claim}
    \label{clm:aab}
    Let $L'$ be an infix-free language that contains a word of the form $aab$ for
    $a,b \in \Sigma$ with $a \neq b$. Then $\resprobset(L')$ is NP-hard.
  \end{claim}

  \begin{proof}
    We use the gadget whose completion is pictured on \cref{fig:aab} (left).
    Let us now discuss the condensed hypergraph of matches. First observe that
    there are three matches forming the word $aab$, and walks of length at most
    2 must form strict infixes of $aab$ so they do not correspond to words
    of~$L'$.
    Note that, as $ab$ is not in~$L'$ (as a strict infix of~$aab \in L'$), then
    we know that every match using $3 \xrightarrow{a} 1$ must also use the fact
    $t_v \xrightarrow{a} 3$; hence $3 \xrightarrow{a} 1$ can be eliminated by
    node-domination. Likewise, $b$ is not in~$L'$, so every match using $3
    \xrightarrow{b} 4$ must also use $t_v \xrightarrow{a} 3$, and $3
    \xrightarrow{b} 4$ is eliminated by node-domination.
    This ensures that the only walk of length greater than 3
    which does not contain $aab$ as an infix, namely, the $aaa$ path from $s_v$
    to $1$, is subsumed in the condensation in the case where $aaa \in L'$:
    namely, after eliminating $3 \xrightarrow{a} 1$ by node elimination, the
    remaining facts $\{s_v \xrightarrow{a} t_v, t_v \xrightarrow{a} 3 \}$
    correspond to the match $\{s_v \xrightarrow{a} t_v, t_v \xrightarrow{a} 3,
    3 \xrightarrow{b} 4\}$ after we have used node-domination to eliminate $3
    \xrightarrow{b} 4$. Thus, we can build a condensed hypergraph of matches by
    considering only the walks that form the word $aab$, and then applying
    condensation we obtain \cref{fig:aab} (right). This establishes that we have
    a gadget, so we can conclude by \cref{prp:gadgethard}.
  \end{proof}

  We can now conclude the proof of \cref{clm:abcacab}:

  \begin{figure}
    \centering
    \begin{subfigure}[b]{0.7\linewidth}
      \includegraphics[scale=0.4]{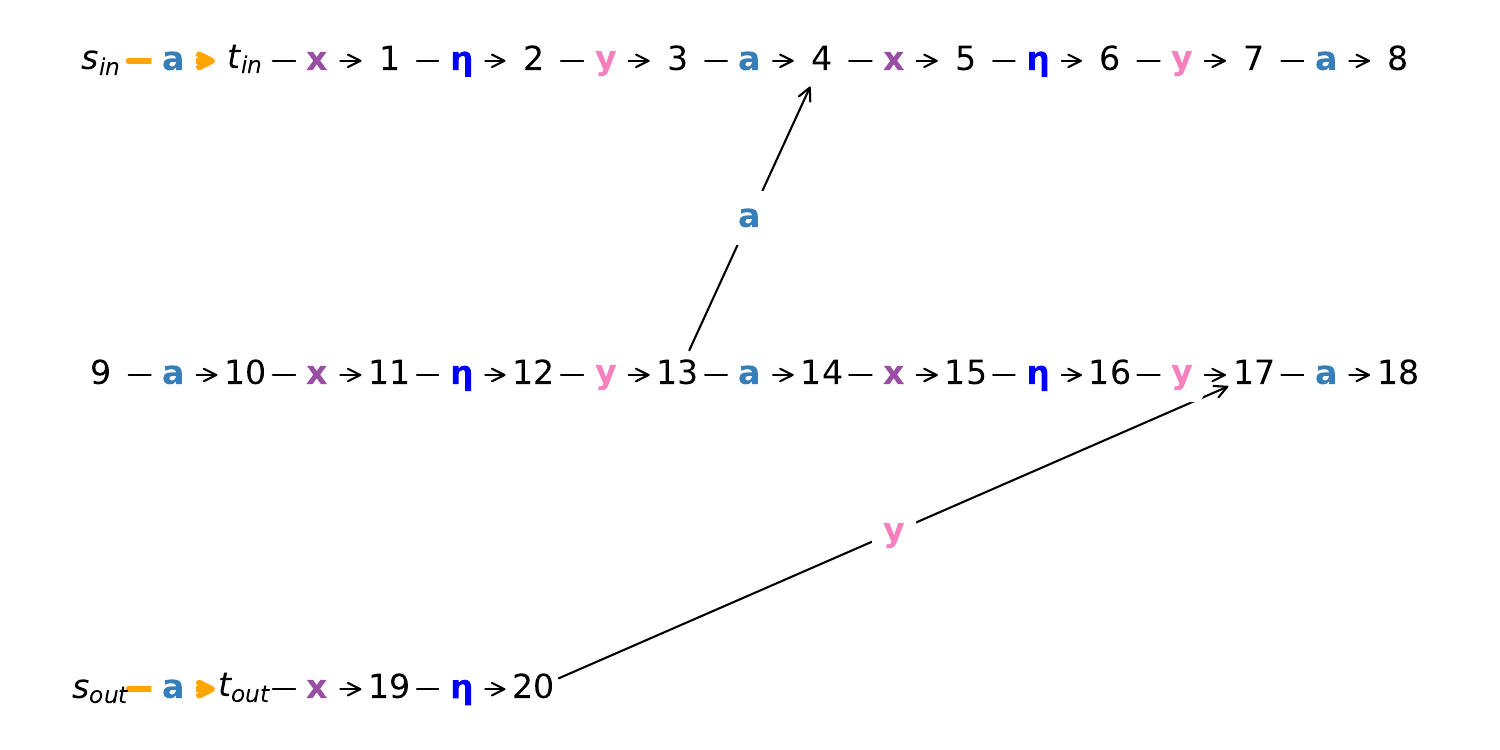}
    \end{subfigure}
    \hfill
    \begin{subfigure}[b]{0.23\linewidth}
      \centering
      \includegraphics[scale=0.4]{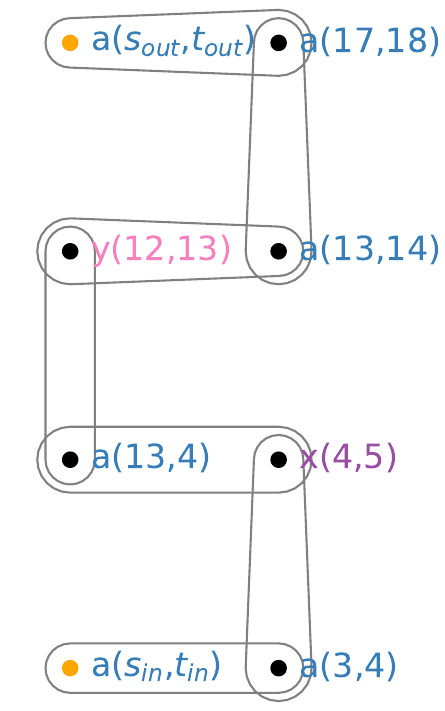}
    \end{subfigure}   \hfill\null
    \vspace{-3mm}
    \caption{Completed gadget used in the proof of \cref{clm:abcacab}}
    \label{fig:abcacab}
  \end{figure}

  \begin{proof}[Proof of \cref{clm:abcacab}]
    We first rule out the case where $y=a$. Then $L'$ contains $yax=aax$, and
    we conclude directly by \cref{clm:aab} or \cref{clm:aaa} depending on
    whether $x=a$ or not.

    We next rule out the case where $x=a$. Then, using \cref{prp:reversal}, let us
    show hardness for $\reversal{L'}$. This language contains $x a y$, and as $x =
    a$ we can apply \cref{clm:aab} or \cref{clm:aaa} again depending on whether $y=a$ or not, and conclude.

    Hence, in the sequel, we assume that $x$ and $y$ are both distinct from $a$.
    (We may have $x=y$ or not.)

    We then use the gadget whose completion is pictured in
    \cref{fig:abcacab} (left), with a condensed hypergraph of matches shown on
    figure \cref{fig:abcacab} (right). The same construction, and the same
    condensed hypergraph of matches, is used no matter whether $x$ and $y$ are
    the same letter or not.

    To see why the condensed hypergraph of matches is as pictured on
    \cref{fig:abcacab}, let us describe the directed walks in the
    completion of the gadget. The ones forming words of the form $ax\eta ya$, and
    $yax$, are guaranteed to be matches: the ones forming strict infixes of
    these are guaranteed not be matches because the language is infix-free; and the
    ones containing one of these as an infix are guaranteed to be eliminated by
    node-domination. In particular, walks that start within an $\eta$-edge, or
    end within an $\eta$-edge, or both, are not relevant. Indeed, observe that
    $\eta$-edges are always followed by (some prefix of) $yax\eta$. So a match
    starting within an $\eta$-edge (possibly using all facts of the edge) must
    include the subsequent $yax$ (because
    $\eta ya$ is a strict infix of $ax \eta ya$ and so is not in the language),
    and the match is then subsumed by the $yax$-match. Likewise, $\eta$-edges
    are always preceded by (some suffix of $\eta yax$), so matches ending within
    an $\eta$-edge (possibly using all facts of the edge) must include the
    preceding $yax$ (because $ax\eta$ is a strict infix of $ax \eta ya$) and
    then they are subsumed by the $yax$ match.

    For this reason, to build the condensed hypergraph of matches, it suffices
    to consider walks that start and end at the edges labeled $a$ and $x$ and
    $y$. One can check that such walks either are infixes of $ax\eta ya$ or
    contain $yax$ as an infix, so their membership to~$L$ is determined by the
    fact that $L$ contains $yax$ and $ax\eta ya$ and is infix-free. Building the
    corresponding condensed hypergraph of matches and further applying
    condensation rules gives \cref{fig:abcacab} (right), which is an odd path,
    so we conclude by \cref{prp:gadgethard}.
  \end{proof}

  Applying \cref{clm:abcacab} with $L' \coloneq  L$, we conclude that $\resprobset(L)$ is
  NP-hard in the remaining cases of the non-overlapping case. This concludes the
  case distinction and concludes the proof of \cref{thm:repeated-letter}.

\section{Additional Complexity Results}
\label{sec:other}
We have shown tractability for local languages (\cref{prp:ro-ptime}),
and hardness for
four-legged languages 
(\cref{prp:four-legged})
and for finite infix-free languages containing
a word with a repeated letter 
(\cref{thm:repeated-letter}). 
These results do not give a full
classification. The purpose of this section is to illustrate two phenomena that any
full classification must account for: non-locality alone does not imply hardness,
and our hardness criteria do not cover all intractable cases.

We first give two tractable classes of languages that are not necessarily local.
First, \emph{bipartite chain languages},
such as \(ab|bc\), are finite languages obtained from a bipartite graph on letter
endpoints; their resilience can be computed by a MinCut reduction that reverses some
matches. Second, \emph{one-dangling languages}, such as \(abc|be\), extend a local language
with one word that consists of two distinct letters such that at least one letter is fresh (i.e., does not
occur in the rest of the language); such languages are tractable by a more
involved flow encoding. We then show additional NP-hard cases not covered by
our
previous hardness criteria.

\subsection{Tractability of Bipartite Chain Languages}
Let us define the
\emph{chain languages}:

\begin{definition}[Chain languages]
  A language $L$ is a \emph{chain language} if:
    \begin{enumerate}
    \item No word in $L$ contains a repeated letter.
    \item Each word of $L$ of length at least~2 may only share its first or last
      letter with other words in~$L$:
        formally,
        for each word $\alpha = a \gamma b$ in~$L$, no letter of~$\gamma$
        occurs in another word of~$L$.
    \end{enumerate}
\end{definition}

Note that chain languages are always finite, because they contain at most
$|\Sigma|$ words of length $>2$: indeed, each word of length 3 or more contains
intermediate letters which cannot occur in another word. We will focus on chain
languages satisfying an extra condition, called being \emph{bipartite}:

\begin{definition}[Endpoint graph, BCL]
  \label{def:BCL}
The \emph{endpoint graph} of a language~$L$ is the undirected graph
$(\Sigma, E)$ having as vertices the letters of~$\Sigma$, and
having an edge $\{a, b\}$ for $a \neq b$ in $\Sigma$ iff $a$ and $b$ are
the endpoints of some word of~$L$ of length at least~$2$ (i.e., $L$ contains a
word of the form $a \alpha b$ or $b \alpha a$).

  A chain language is a \emph{bipartite chain language} (BCL) if its endpoint
  graph is bipartite.
\end{definition}

\begin{example}[\cref{Fig_Bipartite_Chain}]
  \label{ex:bipartite_chain}
  The languages $ab|bc$ (from \cref{expl:four-legged}) and $axyb|bztc|cd|dea$
  are BCLs; the  language $ab|bc|ca$ is a chain language but not a BCL. None of
  these languages are local.
\end{example}

Thus, there are two kinds of chain languages: those that are bipartite, and
those which are not. We conjecture that
non-bipartite ones are intractable:

\begin{conjecture}
  \label{con:chainnonbip}
  Let $L$ be an infix-free chain language. If $L$ is not bipartite, then  $\resprobset(L)$ is NP-hard.
\end{conjecture}

We leave this conjecture open, but show
hardness for the specific non-bipartite chain language of \cref{ex:bipartite_chain}:

\begin{proposition}
  The resilience problem for the language $ab|bc|ca$ is NP-hard.
  \label{prp:abbcca}
\end{proposition}
\begin{proof}

  \begin{figure}
    \null\hfill
    \begin{subfigure}[b]{0.4\linewidth}
      \centering
      \includegraphics[scale=0.4]{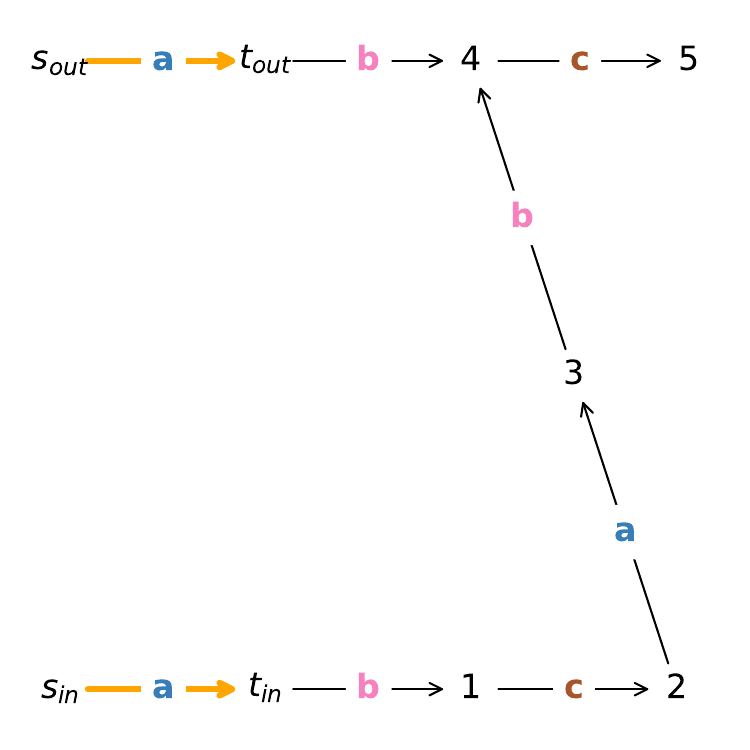}
    \end{subfigure}
    \hfill
    \begin{subfigure}[b]{0.4\linewidth}
      \centering
      \includegraphics[scale=0.4]{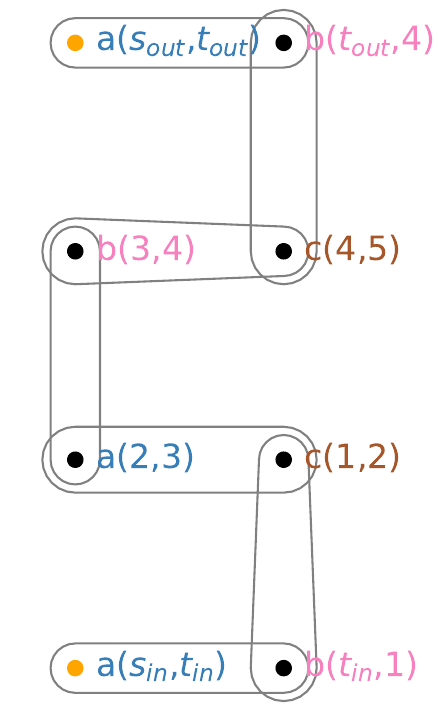}
    \end{subfigure}    
    \hfill\null
    \caption{Completed gadget used in the proof of \cref{prp:abbcca}}
    \label{fig:abbcca}
  \end{figure}

  We claim that the gadget shown in \cref{fig:abbcca} is valid for $ab|bc|ca$.
  Indeed, the gadget has the following properties:
  We can see that it meets all the criteria of being a pre-gadget.
  We can also generate a condensed hypergraph of matches that forms an odd path, and hence this gadget is a valid hardness gadget for $ab|bc|ca$.
\end{proof}

\begin{figure}
\centering
\includegraphics[scale=0.4]{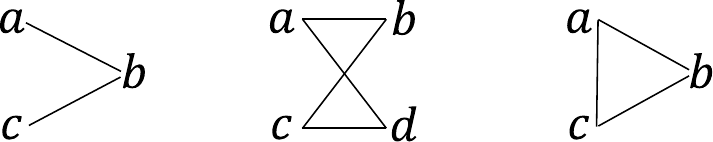}
\caption{Illustrations of the endpoint graphs for the three languages mentioned
  in \cref{ex:bipartite_chain}.}
\label{Fig_Bipartite_Chain}      
\end{figure}

We now focus on chain languages that are bipartite, and we will show
tractability of resilience for such languages in
the rest of this section.
Notice that the definition of BCLs is \emph{symmetric}:
any word $\alpha$ of~$L$ can be replaced by the
reversal $\reversal{\alpha}$
of~$\alpha$ and the resulting language is still a BCL.
Also notice that the class of BCLs is
\emph{incomparable} to local languages (even finite local languages): indeed,
$a x^* b$ and $axb|axc$ are local languages but
not BCLs, and \cref{ex:bipartite_chain} shows BCLs that are not local.
Further note that BCLs containing no word of length less than 2 are always
infix-free.
Further, we have the immediate analogue of \cref{lemma:qwerty}, which we
state here and prove in \cref{apx:redbcl2}:

\begin{toappendix}
  \subsection{Proof of \cref{lem:redbcl22}}
  \label{apx:redbcl2}
\end{toappendix}

\begin{lemmarep}
  \label{lem:redbcl22}
  For any BCL $L$, the language $\red(L)$ is also a BCL.
\end{lemmarep}

\begin{toappendix}
  To prove \cref{lem:redbcl22}, we prove the following more general claim:

\begin{lemma}
  \label{lem:redbcl}
  For any BCL $L$, any $L' \subseteq L$ is also a BCL.
\end{lemma}

To establish this, we  first show that being a chain language is closed under taking subsets:

\begin{lemma}
  \label{lem:redbcl2}
  For any chain language $L$, any $L' \subseteq L$ is also a chain language.
\end{lemma}

\begin{proof}
  Let $L' \subseteq L$. We check the conditions of chain languages:
  \begin{itemize}
      \item $L'$ cannot contain a word with a repeated letter

      \item There
        cannot be a word $\alpha$ of $L'$ sharing an intermediate letter with
        another word $\beta$ of~$L'$ without $\alpha$ and $\beta$ being
        counterexamples for~$L$ being a chain language.
  \end{itemize}

  This establishes that $L'$ is a chain language.
\end{proof}

  Then we can prove \cref{lem:redbcl}:

  \begin{proof}[Proof of \cref{lem:redbcl}]
    Let $L$ be a BCL and take $L' \subseteq L$. By \cref{lem:redbcl2}, we know
    that $L'$ is a chain language. Further, if the endpoint graph
  of~$L$ is bipartite, then the same is true of~$L'$, because the endpoint graph
  of~$L'$ is a subgraph of that of~$L$. Hence, $L'$ is a BCL.
\end{proof}

  This implies in particular \cref{lem:redbcl22}, which concludes.
\end{toappendix}

We can now state our tractability result: the resilience problem is tractable
for all BCLs, both in combined
complexity and in bag semantics:

\begin{proposition}[Bipartite chain languages are tractable]
    \label{prp:rpn}
    If $L$ is a BCL then $\resprobbag(L)$ is in PTIME. Moreover, the problem is also
    tractable in combined complexity: given an
    $\epsilon$NFA $A$ that recognizes a BCL language $L$, and a database $D$, we
    can compute $\resbag(\query_L,D)$ in 
    $(|A| \times |D|^2 \times
    |\Sigma|^2)^{1+o(1)}$.
  \end{proposition}

We present the proof in the rest of this subsection. At a high level, we proceed by partitioning the letters of~$\Sigma$ based on the bipartition of the
  endpoint graph. We then build a flow network like for \cref{prp:ro-ptime}, but
  we reverse some words depending on the sides of the partition in which their
  endpoints fall. Before doing all this, however, in order to establish our precise runtime
  bounds we must show a preliminary claim (\cref{lem:nfachain}):
  the input $\epsilon$NFA recognizing the BCL language can be tractably
  converted to an explicit representation of its language. We state this below,
  and as the proof is routine we defer to \cref{apx:nfachain}:

\begin{toappendix}
  \subsection{Proof of \cref{lem:nfachain}}
  \label{apx:nfachain}
\end{toappendix}

  \begin{lemmarep}
    \label{lem:nfachain}
    Given an $\epsilon$NFA $A$ such that the language $\LL(A)$ recognized
    by~$A$ is a chain language, we can compute
    in time $O(|\Sigma|^2 \times |A|)$ the explicit list of words of the finite
    language $L$.
  \end{lemmarep}

  \begin{toappendix}
  To prove this result, we need an auxiliary claim that intuitively corresponds to the
  setting where we remove the first and last letters from a chain language:

  \begin{claim}
    \label{lem:nfauniq}
    Let $A$ be an $\epsilon$NFA recognizing a finite language where no words
    contain a repeated letter and where no words
    share the same letter, i.e., for every word $\alpha \in \LL(A)$ all
    letters of~$\alpha$ are distinct, and for any two words $\alpha$ and $\beta$
    with $\alpha \neq \beta$ the set of letters of~$\alpha$ and of~$\beta$ are
    disjoint. Then we can compute in $O(|A|)$ the finite
    language $\LL(A)$ as an explicit list of words.
  \end{claim}

  \begin{proof}
    The finiteness of the language is simply because it has cardinality at most
    $|\Sigma|+1$: except for the empty word, every word of the language
    eliminates a letter that cannot be used elsewhere.
    We assume that $A$ is trimmed (see \cref{def:trim}); this can be enforced in
    linear time by \cref{clm:trim}.
    We perform two graph traversals to compute in linear time the set $S_l$ of
    states having a path from some initial state via $\epsilon$-transitions, and
    the set $S_r$ of states having a path to some final state via
    $\epsilon$-transitions.

    We first check if $S_l \cap S_r \neq \emptyset$: if so we add
    $\epsilon$ to the list of recognized words.

    For $s$ a state of~$A$, a \emph{witness word} is a (possibly empty) word
    labeling a path from an initial state of~$A$ to~$s$.
    Now, we claim the following (*): for any state $s$ of~$A$ that has a
    (possibly empty) path via $\epsilon$-transitions to an $x$-transition for
    some $x \in \Sigma$, there is at most one witness word. Indeed, assume by way of
    contradiction that there are two witness words $\alpha \neq \beta$. As $A$ is trimmed, let $\gamma$ be a word
    labeling a run from the target of the $x$-transition to a final state of~$A$. Thus, $A$ accepts $\alpha x \gamma$ and
    $\beta x \gamma$, which are different words because $\alpha \neq \beta$, and
    which share letters (namely~$x$), contradicting the assumption
    on~$\LL(A)$.

    Hence, we compute the list of the non-empty words accepted by~$A$ by doing a
    graph exploration from $S_l$, maintaining for each state one or several
    witness words. For states to which claim (*) applies, we will store
    precisely one witness word, which corresponds to the one of~(*). For states
    $s$ to which (*) does not apply, we may store multiple witness words, but we know that the language accepted
    from~$s$ cannot contain non-empty words, thus it is precisely $\{\epsilon\}$ because $A$ is trimmed (in
    particular $s$ is in~$S_r$), so we do not need to explore further from~$s$. Thus we have the invariant that whenever we reach a state $s$ from another
    state $s'$ then $s'$ had precisely one witness word.

    We start the traversal at states of~$S_l$ with witness word~$\epsilon$. For
    any state $s$ that is reached in the traversal, then either it is reached by
    an $\epsilon$-transition from a state $s'$ and it has a
    witness word which is the same as that of~$s'$, or it is reached by an
    $x$-transition for some letter $x$ from a state~$s'$ and it has a witness word $\alpha x$ where
    $\alpha$ is the witness word of~$s'$. We store the witness words as
    pointers to a trie data structure:
    \begin{itemize}
      \item The root of the trie corresponds to the
    word $\epsilon$ and the nodes of $S_l$ point to the root;
  \item If we traverse a
    transition $s' \xrightarrow{\epsilon} s$ then the pointer for $s$ is the
    same as the pointer for~$s'$;
  \item If we traverse a transition $s'
    \xrightarrow{x} s$ then the pointer for $s$ leads to an $x$-labeled child
        of the node which is the pointer for~$s'$ (creating such a child if it does
        not yet exist).
    \end{itemize}
    Further, when we reach a state $s$ which is in~$S_r$, we add a mark to the node of the tree that
    corresponds to the witness word by which we have reached it. Note that, as $A$ is
    trimmed, all leaves of the trie are marked.

    At the end of this process, the trie has linear size because the graph
    traversal is in linear time. Further, let us show that the trie obeys a very
    strong structural restriction: the least
    common ancestor of any two marked nodes is the root (**). The claim (**)
    immediately implies that the trie has the following structure (***): a union of disjoint
    branches (i.e., the only node with more than one child is the root), with
    the marked nodes being precisely the set of leaves (plus possibly the root).
    Let us now show claim (**). Assume by way of contradiction that there are
    two marked nodes $n_1$ and $n_2$ whose
    lowest common ancestor $n$ is not the root. Then this witnesses that
    there is a non-empty word $\alpha$ (corresponding to the label of the path
    from the root to~$n$) such that $\alpha \beta$ and
    $\alpha \gamma$ are accepted by the automaton, for two distinct $\beta$ and
    $\gamma$ corresponding to paths from $n$ to $n_1$ and $n_2$ respectively.
    However, $\alpha \beta$ and $\alpha \gamma$ are two distinct words which
    share the letters of~$\alpha$, a contradiction again.

    Thanks to the structure of the trie (***), we can finish by computing
    explicitly the labels of the paths from the root of
    the trie to the leaves. Claim (***) implies that this has size linear in the
    size of the trie, hence the overall process runs in
    $O(|A|)$. This concludes the proof.
  \end{proof}

  We can now prove our claim on $\epsilon$NFA recognizing chain languages
  (\cref{lem:nfachain}) from \cref{lem:nfauniq}:

  \begin{proof}[Proof of \cref{lem:nfachain}]
    Let $A$ be the $\epsilon$NFA recognizing a language $L \coloneq \LL(A)$ which is asserted to be
    a chain language. We assume that $A$ is trimmed (see \cref{def:trim}); this can be enforced in
    linear time by \cref{clm:trim}.

  We first deal in a special way with words of length $0$ and $1$ that
    are accepted by~$A$. We first test in linear time as in the proof of
    \cref{lem:nfauniq}
     whether $A$ accepts~$\epsilon$, and if so, add $\epsilon$ to the list of
    accepted words.

  We then compute
  in linear time in~$A$ the set $\Sigma_1$ of letters $a$ such that the
  single-letter word $a$ is accepted by~$L$. We compute these by first computing
  in linear time the states $S_l$ accessible from initial states via
  $\epsilon$-transitions, computing in linear time the states $S_r$ accessible
  from final states by traversing $\epsilon$-transitions backwards, then going
  over all transitions with source in $S_l$ and target in~$S_r$ and letting
  $\Sigma_1$ be the set of letters which label such transitions. 
  We add all the words of $\Sigma_1$ to the list of accepted words.

  Hence, after this preprocessing in $O(|A|)$, we can focus on the words of
    length 2 or more that are accepted by~$A$. We now consider every pair of letters
    $(a,b) \in
    \Sigma^2$ (possibly $a = b$), and compute explicitly the list of words
    accepted by~$A$ that start with~$a$ and end with~$b$. Let $S_{l,a}$ be the set
    of states defined in the
    following way: we take all $a$-labeled
    transitions whose tail is in $S_l$, and let $S_{l,a}$ be the set of all states
    that are heads of such transitions. Further, let $S_{r,b}$  be the set of states
    defined in the following way: we take all $b$-labeled transitions whose head
    is in $S_r$, and let $S_{r,b}$ be the set of all states that are tails of
    such transitions. The sets $S_{l,a}$ and $S_{r,b}$ are computed in
    time $O(|A|)$. We will now compute explicitly the set of words $\alpha$ such
    that $a \alpha b \in L$. This clearly amounts to computing the words
    accepted by the $\epsilon$NFA $A_{a,b}$ obtained from~$A$ by changing the
    initial states to be the states of~$S_{l,a}$ and changing the final states
    to be the states of~$S_{r,b}$.
    Hence, we can compute the list of words
    in time $O(|A|)$ by \cref{lem:nfauniq}. The algorithm has an overall
    complexity of $O(|A| \times |\Sigma|^2)$, which concludes the proof.
  \end{proof}
  \end{toappendix}
  
  We can then prove \cref{prp:rpn} from \cref{lem:nfachain}:

  \begin{inlineproof}[Proof of \cref{prp:rpn}]
    Thanks to \cref{lem:nfachain}, we compute in time $O(|A| \times |\Sigma|^2)$
    the language $L \coloneq \LL(A)$ as an explicit list of words.

    We first process $L$ in linear time to eliminate the words of length
    $0$ or $1$. Namely, if $L$ contains $\epsilon$, then resilience is~$+\infty$; so in
    the sequel we assume that $L$ does not contain $\epsilon$. For words
    of length~$1$, for every letter
    $a \in \Sigma$ such that $L$ contains the single-letter word~$a$, then we
    know that in the input database $D$ we must remove every $a$-fact for the
    query to be false; so we can preprocess $D$ to remove all such facts
    (accounting for their multiplicities in the resilience value that we are
    computing)
    and
    remove the corresponding one-letter words from~$L$. Hence, at the end of
    this preprocessing, in time $O(|A|+|D|)$, we can assume that $L$ consists of
    words of length 2 or more.

    We now point out that $L$ is infix-free. This is by
  definition of a chain language and from the assumption that the words have
  length at least 2. Indeed, if $\alpha \in L$ is a strict infix of $\beta \in L$ then
  $\beta$ must have length at least 3, so there are letters in $\beta$ that are neither
  the first nor the last letter, by definition of $L$ being a chain language these letters
  cannot occur in a word of $L \setminus \{\beta\}$, so they do not occur in
    $\alpha$, so $\alpha$ being an infix
  of~$\beta$ implies that $\alpha$ has length $1$, contradicting the assumption that the
  words of~$L$ have length 2 or more.

  We then build the endpoint graph $G$ of~$L$, simply by going over the explicit
  list and noting which letters occur as the first or last letter of some word.
We now know by assumption that $G$ is bipartite: we compute a witnessing
  bipartition in linear time by the obvious greedy algorithm. 
  Let us assign each letter that occurs in~$L$ as an endpoint (i.e., as the
  first or last letter of a word of length at least~$2$) to either the source or
  target partition. 
  Now, the words of $L$ can be partitioned between the \emph{forward words}, in
  which the first letter is assigned to the source partition and the
  last letter to the target partition;
  and the \emph{backward words}, where the opposite occurs. 

  Then we construct a flow network as follows:
  \begin{itemize}
    \item We add two fresh vertices $t_{\text{source}}$ and $t_{\text{target}}$ which are the source and the target respectively of the network.
    \item For every fact $a(v,v') \in D$ we define two vertices $\start_{a,v,v'}$ and $\eend_{a,v,v'}$.
    \item For every letter $a$, there is an edge with capacity
      $\mult(a(v,v'))$ from $\start_{a,v,v'}$ to
      $\eend_{a,v,v'}$ for every $v,v'$ such that 
      we have $a(v,v')\in D$.

    \item For every forward word containing an infix $ab$, where $a,b$ are letters in $\Sigma$,
      and $a(v,v')$ and $b(v',v'')\in D$, there is an edge with capacity
      $+\infty$ from $\eend_{a,v,v'}$ to $\start_{b,v',v''}$  

    \item For every backward word containing an infix $ab$, where $a,b$ are
      letters in $\Sigma$, and $a(v,v')$ and $b(v',v'')\in D$, there is an edge with capacity $+\infty$ 
      from $\eend_{b,v',v''}$  to
      $\start_{a,v,v'}$.

    \item There is an edge with capacity $+\infty$ from $t_{\text{source}}$ to
      every vertex of the form $\start_{a,v,v'}$ where $a$ is an endpoint
      assigned to the source partition, and from every vertex of the form
      $\eend_{a,v,v'}$ to~$t_{\text{target}}$ where $a$ is an endpoint in the target partition.
  \end{itemize}
  The explicit list of words has size $O(|A| \times |\Sigma|^2)$, and the
  construction runs in time $O(|D|^2 \times |A| \times |\Sigma|^2)$.

  We can see that this construction contains a single edge for every fact,
  whose capacity is the multiplicity of said fact. 
  Every witness of the instance corresponds to a source-to-target walk by
  construction. 
  Indeed, for each forward word, we have a path from $t_{source}$
  to the edge corresponding to its first letter, then following each letter in
  succession until we reach the last letter and then $t_{target}$.
  Further, for each backward word, we have a path from $t_{source}$
  to the edge corresponding to its last letter, then following each letter in
  succession in reverse order, until we reach the first letter and $t_{target}$.

  We must also show that there is no path from source to target in the network that does not correspond to a witness.
  Any path from $t_{source}$ to $t_{target}$ must start with a vertex of the
  form $\start_{a,v,v'}$ where $a$ is assigned to the source partition (in
  particular it is not a letter that occurs in the middle of any word). We
  then move to a vertex $\start_{b,v',v''}$ and there are two cases: either
  $a$ was the first letter of a forward word $\alpha$ and $ab$ occurs as an infix in that word,
  necessarily at the beginning because no middle letters are repeated, or $a$ was the
  last letter of a backward word $\alpha$ and $ba$ occurs in that word, again necessarily at
  the end because no middle letters are repeated. We describe the first case
  (corresponding to forward words), the second (corresponding to backward words)
  is symmetric. From the condition that intermediate letters of words do not
  occur in other words, whenever we reach a vertex, then
  the next vertex must correspond to the next letter of~$\alpha$. Eventually we
  reach the last letter of~$\alpha$ and the corresponding node has no outgoing
  edge in the flow network except the edge to~$t_{\text{target}}$, so we
  terminate.
  Thus, all paths in this network correspond to witnesses. 
  
  Thus, this construction maintains the one-to-one
  correspondence between the contingency sets of~$D$
  for~$\query_L$ and the cuts of the network that
  have finite cost: selecting facts to remove in the
  contingency set corresponds to cutting edges in the
  network, and ensuring that the contingency set does
  not satisfy the query corresponds to ensuring that
  no source to target paths remain. Furthermore, 
  the sum of multiplicities of the facts in a
  contingency set is equal to the cost of the
  corresponding cut, so that minimum cuts correspond
  to minimum contingency sets. We can then, as in
  the case of local languages, simply solve the
  MinCut problem for this network in 
  $(|A| \times |D|^2 \times |\Sigma|^2)^{1+o(1)}$,
  concluding the proof.
\end{inlineproof}

\subsection{One-Dangling Languages}
We now present another example of tractable languages, dubbed
\emph{one-dangling languages}, for which tractability
is shown by a more involved reduction to a flow problem.

\begin{definition}[One-dangling language]
  A \emph{one-dangling language} is a language that can be written as $L \cup
  \{xy\}$ where $L$ is a local language over some alphabet $\Sigma$
  and where $x$ and $y$ are distinct letters
  such that at least one of them is not in $\Sigma$. 
\end{definition}

Note that one-dangling languages cover in particular the more restricted class of languages for which tractability had been shown in
the conference version of this paper
\cite[Proposition~7.7]{amarilli2025resilience}.           %
They also cover the two languages $abcd|be$ and $ax^*b|xd$ for which the
complexity of resilience had been left open 
in~\cite{amarilli2025resilience}.       %
So we can subsume that result by
showing the following:

\begin{proposition}
  \label{prp:submod}
  Let $L \cup \{xy\}$ be a one-dangling language where $L$
  is recognized by an RO-$\epsilon$NFA $A$. 
  The problem $\resprobbag(L\cup \{xy\})$ is solvable in 
      $(|A| \times |D| \times |\Sigma|)^{1+o(1)}$
     time.  
\end{proposition}

In addition to covering more languages, this result is a combined
tractability result instead of a data complexity result, and it
gives a better polynomial exponent than
\cite[Proposition~7.7]{amarilli2025resilience}.           %
This is because our proof will
proceed via a reduction to a
flow problem, 
instead of the submodular-function-minimization argument used in \cite[Proposition~7.7]{amarilli2025resilience}, 
which relies on the general techniques surveyed in ~\cite{MCCORMICK2005321}.

We prove \cref{prp:submod} in the rest of this
subsection, before finally turning to other NP-hard cases in the rest of the
section.
If $\epsilon$ is in~$L$, then resilience is $+\infty$; so in
the sequel we assume that $L$ does not contain $\epsilon$.
We do a slight abuse of notation throughout the proof for convenience: for $D''$ a bag database and $L''$ a
language, we will talk of a \emph{contingency set of~$D''$ for~$L''$} simply to mean a
contingency set of~$D''$ for $\query_{L''}$.

Due to Proposition~\ref{prp:reversal}, if there is a PTIME solution for $\resprobbag(L\cup \{xy\})$,
then there is a PTIME solution for $\resprobbag(L^R\cup \{yx\})$, and the reversal of
a local language is local (as can easily be seen from~\cref{prp:carac}). Hence, we can assume without loss of generality that
$y\not\in\Sigma$. Furthermore, we can also assume without loss of generality that $x$ occurs in some word of $L$, as otherwise $L\cup \{xy\}$ is local and we can conclude by \cref{prp:ro-ptime}. In particular, $x\in \Sigma$.

Let $z$ be a fresh letter which is not in $\Sigma \cup \{y\}$.
Our strategy will be to rewrite the input bag database $D$ over $\Sigma \cup
\{y\}$
to a modified bag database $D'$ over  $\Sigma \cup \{z\}$,
and create a local language $L'$ over $\Sigma \cup \{z\}$
such that $\resbag(L\cup \{xy\}, D) = \resbag(L',D') + \kappa$,
where $\kappa$ is a value that we can tractably compute from $D$. 
We first define the language~$L'$, which is intuitively obtained by
replacing every appearance of the letter $x$ in~$L$ by the word $xz$. 
Formally, let $A = (S,I,F,\Delta)$ be an RO-$\epsilon$NFA for $L$. 
Let $A'$ be the automaton that results from adding one fresh state $s'$ to $S$,
and replacing the one transition $(s,x,t)$ labeled by~$x$ (which exists as $x$ occurs in some word of $L$) by two transitions $(s,x,s')$ and $(s',z,t)$.
Note that the new state $s'$ is not final, so there are no 
words in $\LL(A')$ that end in $x$.
Clearly, the automaton $A'$ is also an RO-$\epsilon$NFA, so that $L' \coloneq
\LL(A')$ is also a local language. As $L$ is a language over $\Sigma$, the
language $L'$ is defined over $\Sigma \cup \{z\}$, and further for any
$m \in \mathbb{N}$ and $\alpha_1, \ldots, \alpha_m \in \Sigma^*$ we have
$\alpha_1 x \alpha_2 x \cdots \alpha_{m-1} x \alpha_m \in L$ iff 
$\alpha_1 xz \alpha_2 xz \cdots \alpha_{m-1} xz \alpha_m \in L'$.

Let us now define how the input bag database $D$ over $\Sigma \cup \{y\}$ is rewritten to a
bag database~$D'$ over $\Sigma \cup \{z\}$.
To do this, we will extend the semantics of bag graph databases to what we
call the \emph{extended bag semantics}, where we allow 
the function $\mult$ to give negative multiplicities, or multiplicity zero,
to edges of the database.
The rest of the semantics
remains unchanged, and the definition of resilience is unchanged as well --
noting that, when computing resilience over database in extended bag semantics,
one can always assume that facts with a non-positive multiplicity are removed.
For this reason, it is easy to see that 
$\resprobbag(L)$ under the extended semantics
reduces in linear time to $\resprobbag(L)$
under the normal bag semantics.
Let us use the notation $\resprobbag^{\sf ex}(L)$ for a fixed language $L$
to denote the resilience
problem when posed over databases in
extended bag semantics.

We can now describe the transformation. Given a bag database $D$, let $V$ be the domain of $D$. For every node $v\in V$, let $D^-_x(v)$ be the set of 
edges of the form $u\xrightarrow{x}v$ for some $u\in V$, and
let $D^+_y(v)$ be the set of edges of the form $v\xrightarrow{y}u$ for some $u\in V$.
We rewrite $D$ by doing the following for each $v\in V$:
\begin{enumerate}
  \item add a fresh node $(v,{\sf in})$ into $V$;
  \item replace all edges of the form $u\xrightarrow{x}v \in D^-_x(v)$ by
    the edge $u\xrightarrow{x}(v,{\sf in})$, keeping the same multiplicity;
\item add an edge $(v,{\sf in}) \xrightarrow{z} v$ into $D$ and define
 $\mult((v,{\sf in}) \xrightarrow{z} v) \coloneq \sum_{e\in D^-_x(v)}\mult(e) -
    \sum_{e\in D^+_y(v)}\mult(e)$ (which may be negative);
 and 
    \item erase any edges in $D^+_y(v)$ from $D$.
\end{enumerate}
 Let us call $D'$ the resulting bag database: it is a bag database in extended
 semantics, and it is over the alphabet $\Sigma \cup \{z\}$ because applying
 the last point of the list above over all $v \in V$ will remove all $y$-facts.
 Let us define 
 the value $\kappa \coloneq \sum_{v\in V}\sum_{e\in D^+_y(v)}\mult(e)$
 (evaluating this expression in $D$).  Note that $\kappa$ can be computed in
 linear time from~$D$ as it is simply the sum of multiplicities of all edges
 labeled $y$ in $D$.

We now claim the following:
\begin{equation}
  \tag{$\dagger$} %
  \resbag(L\cup \{xy\}, D) = \resbag^{\sf ex}(L',D') + \kappa \label{eq:res}
\end{equation}
which suffices to conclude, because we have computed~$\kappa$ and we can compute
the resilience expression in the right-hand side using
\cref{prp:ro-ptime}: the language $L'$ is local, and as we explained the
complexity of computing resilience is unchanged when working over
extended bag semantics databases. To show Equation~\eqref{eq:res}, we will show the
following claim, and explain after its proof how the lemma implies
Equation~\eqref{eq:res}:

\begin{claim}
  \label{clm:flow}
  (i) For any minimum-cost contingency set $C$ of~$D$ for $L\cup\{xy\}$, there
        exists a contingency set $C'$ of~$D'$ for~$L'$ such that $\sum_{e\in C}\mult(e) = \kappa + \sum_{e\in C'}\mult(e)$; and 
  (ii) for any minimum-cost contingency set $C'$ of~$D'$ for~$L'$, there exists
        a contingency set $C$ of~$D$ for $L\cup\{xy\}$ such that  $\sum_{e\in C}\mult(e) = \kappa + \sum_{e\in C'}\mult(e)$.
\end{claim}
\begin{proof}

To prove (i), consider a minimum-cost contingency set $C$ of~$D$ for $L\cup \{xy\}$. 
We will describe a set $C'\subseteq D'$ and prove that it is a contingency set of~$D'$ for~$L'$, and
that $\sum_{e\in C}\mult(e) = \kappa + \sum_{e\in C'}\mult(e)$.

\myparagraph{Constructing $C'$} The set we describe is given by inspecting every node $v\in V$, and noting that 
each satisfies at least one of two cases: (1) either every $e\in D^-_x(v)$ is in $C$, or (2) every $e\in D^+_y(v)$ is in $C$.
If neither holds, then there would be edges $u_1\xrightarrow{x}v$ and
  $v\xrightarrow{y}u_2$ in $D$ and not in $C$, and so the database $D\setminus
  C$ would contain a walk labeled $xy$, contradicting the fact that $C$ is a
  contingency set. 
We also note that if (1) holds for a node $v$, then there is no edge $e\in D^+_y(v)$ in $C$: none of those edges can be part of a walk labeled $xy$ since in $D\setminus C$ there are no remaining edges labeled $x$ that reach $v$, so if any of those edges (which have positive multiplicity) were in $C$, they could be removed and $C$ would still be a contingency set, contradicting the fact that $C$ is minimum-cost.

Then $C'$ is defined as follows: for each $v\in V$ that satisfies (1), the edge $(v,{\sf in})\xrightarrow{z}v$ is added to $C'$; and for each $v$ that satisfies (2), and each $u\in V$, the edge $u\xrightarrow{x}(v,{\sf in})$ is added to $C'$ iff we have $u\xrightarrow{x}v\in C$. Every other edge in $C$ with a label different from $x$ and $y$ is kept into $C'$ as is.

\myparagraph{ $C'$ is a contingency set} Towards a contradiction, let us assume that $C'$ is not a contingency set of~$D'$ for $L'$, so there is a walk $w'$ labeled $\alpha'\in L'$ in $D'\setminus C'$. 
If the walk $w'$ does not visit any vertex of the form 
$(v,{\sf in})$ for some $v\in V$, then $w'$ would be present as is in $D\setminus C$ and form a word of~$L$, contradicting the fact that $C$ is a contingency set. Otherwise, whenever $w'$
  contains $u\xrightarrow{x}(v,{\sf in})$ for some $u\in V$, then it also contains $(v,{\sf in})\xrightarrow{z}v$: indeed these are the only outgoing edges of vertices of the form $(v,{\sf in})$, and $w'$ cannot finish on such a vertex because the transition labeled $x$ in $A'$ does not reach a final state. 
For every such $v$, we obtain that $v$ could not have satisfied (1) in~$D$, because $(v,{\sf in})\xrightarrow{z}v$ would have been added to $C'$ and could not have been traversed by~$w'$, and so $v$ satisfies (2). 
In this case, we have that $u\xrightarrow{x}v \not\in C$, so we
  can form a walk $w$ in $D\setminus C$ by replacing all edges of the form $u\xrightarrow{x}(v,{\sf in})$ and $(v,{\sf in})\xrightarrow{z}v$, for each $v \in V$ such that $(v,{\sf in})$ is traversed by~$w'$, by $u\xrightarrow{x}v$ in~$D \setminus C$. 
  The label of $w$ is obtained from that of~$w'$ by replacing every $xz$ by~$x$ in~$\alpha'$: this yields a word $\alpha \in L$, so that the $\alpha$-walk $w$ in $D \setminus C$ contradicts the fact that $C$ is a contingency set.

\myparagraph{Equivalence} Let us now show that $\sum_{e\in C}\mult(e) = \kappa + \sum_{e\in C'}\mult(e)$. For this, let us inspect the value $\zeta = \sum_{e\in C}\mult(e) - \sum_{e\in C'}\mult(e)$. Recall that every edge not labeled $x$ or $y$ is present in~$C'$ iff it is present in $C$, so we only consider edges labeled $x$ or $y$ in~$C$, 
and edges labeled $x$ or $z$ in~$C'$. We can decompose $\zeta$ by each $v\in V$ by looking at the edges that are added to $C'$ depending on whether $v$ satisfied (1) or (2). First, if $v$ satisfied (1), 
then none of the edges in $D^-_x(v)$ are added to $C'$, and instead we add $(v,{\sf in})\xrightarrow{z}v$. So the contribution of~$v$ to $\zeta$ is $\sum_{e\in D^-_x(v)}\mult(e) - \mult((v,{\sf in})\xrightarrow{z}v)$, which, from the definition of $\mult((v,{\sf in})\xrightarrow{z}v)$, is equal to $\sum_{e\in D^+_y(v)}\mult(e)$.
Second, if $v$ satisfied (2), then the edges from $D^-_x(v)$ that are added to $C'$ are exactly the ones that correspond to those in $C$ so they do not contribute to $\zeta$. 
However, we have that all edges $e\in D^+_y(v)$ are in $C$, 
and they are not added into $C'$, so these do contribute, and so the total contribution of $v$ to $\zeta$ is also $\sum_{e\in D^+_y(v)}\mult(e)$. 
We obtain that $\zeta = \sum_{v\in V}\sum_{e\in D^+_y(v)}\mult(e)$ which matches the definition of $\kappa$, and this proves that $\sum_{e\in C}\mult(e) = \kappa + \sum_{e\in C'}\mult(e)$.

\medskip

Now we move on to prove (ii). Towards this, let $C'$ be a minimum-cost contingency set of~$D'$ for $L'$. We will proceed analogously to (i).

\myparagraph{Constructing $C$} For this, we will start by characterizing $C'$ and noting that the following property must hold for each $v\in V$: if $C'$ contains $(v,{\sf in})\xrightarrow{z}v$ then $C'$ should not contain $u\xrightarrow{x}(v,{\sf in})$ for any $u\in V$. This is because $L'$ does not contain any word ending in $x$, and the only edge that starts at $(v,{\sf in})$ is $(v,{\sf in})\xrightarrow{z}v$. Thus, any edge $u\xrightarrow{x}(v,{\sf in})$ (which has positive multiplicity) could be removed from $C'$, which would contradict the fact that $C'$ has minimum cost.

Now, let us build $C$ from $C'$. For each node $v\in V$, we consider two cases: (a) $C'$ contains $(v,{\sf in})\xrightarrow{z}v$ or (b) it does not. In case (a), we add all edges in $D^-_x(v)$ to $C$, and in case (b) we add all edges in $D^+_y(v)$ to $C$, and also for each edge $u\xrightarrow{x}(v,{\sf in})$ of~$D'$,
  we add $u\xrightarrow{x}v$ to $C$ iff $u\xrightarrow{x}(v,{\sf in})\in C'$. Any other edge in $C'$ not labeled $x$ or $z$ is also added in $C$.

\myparagraph{$C$ is a contingency set}
To see that $C$ is a contingency set of~$D$ for~$L$, we note that since for each $v\in V$, we have added either all edges in $D^-_x(v)$ to $C$ or all edges in $D^+_y(v)$ to $C$, so there is indeed no walk labeled $xy$ in $D\setminus C$. Further, using a reasoning analogous to the above, if any walk in $D\setminus C$ is labeled $\alpha\in L$,
  then a walk 
  could be found in $D'\setminus C'$ whose label is the result of replacing all occurrences of~$x$ in~$\alpha$ by~$xz$, yielding a walk in $D' \setminus C'$ that achieves a word of~$L'$ and contradicting the fact that $C'$ is a contingency set.

\myparagraph{Equivalence} To show $\sum_{e\in C}\mult(e) = \kappa + \sum_{e\in C'}\mult(e)$ we refer to the paragraph {\em Equivalence} from point (i) and note that the reasoning is analogous.
\end{proof}

We end 
the proof of \cref{prp:submod}
by explaining why Equation~\eqref{eq:res}
follows from \cref{clm:flow}. First, suppose that $\resbag(L\cup \{xy\}, D) <
\resbag^{\sf ex}(L',D') + \kappa$. Let $C^*$ be a minimum-cost contingency set of~$D$ for
$L\cup \{xy\}$, and let $C'^*$ be a minimum-cost contingency set of~$D'$ for $L'$. We have that $\sum_{e\in C^*}\mult(e) < \kappa + \sum_{e\in
C'^*}\mult(e)$. However, due to point (i) in \cref{clm:flow}, we know that
there exists a contingency set $C'$ of~$D'$ for $L'$ such that $\sum_{e\in
C^*}\mult(e) = \kappa + \sum_{e\in C'}\mult(e)$, and so $\sum_{e\in C'}\mult(e)
< \sum_{e\in C'^*}\mult(e)$, which contradicts the fact that $C'^*$ has minimum cost.
Now suppose that $\resbag(L\cup \{xy\}, D) > \resbag^{\sf ex}(L',D') + \kappa$.
We reach a contradiction using an analogous argument and using point (ii) in
\cref{clm:flow}.
Thus, we have concluded the proof of \cref{prp:submod}.

\subsection{Additional NP-Hard Cases} 
Already in the case of finite infix-free languages, there are gaps remaining
between the tractable cases (local languages, and languages covered by
\cref{prp:rpn}) and the hard cases (languages containing a word with a repeated letter, and four-legged
languages) that we have identified so far. We show below that an unclassified
language is NP-hard by building a specific gadget:

\begin{proposition}
    \label{prp:otherhard}
    $\resprobset(abcd|be|ef)$ and $\resprobset(abcd|bef)$ are NP-hard.
\end{proposition}
\begin{proof}
  \begin{figure}
    \null\hfill
    \begin{subfigure}[b]{0.4\linewidth}
      \centering
      \includegraphics[scale=0.4]{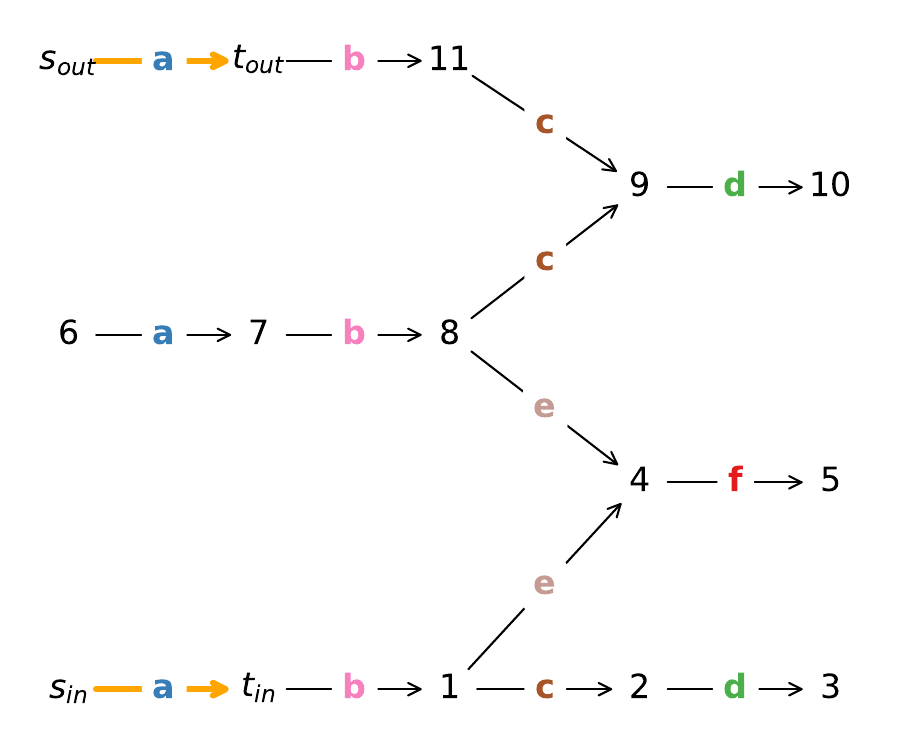}
    \end{subfigure}
    \hfill
    \begin{subfigure}[b]{0.4\linewidth}
      \centering
      \includegraphics[scale=0.4]{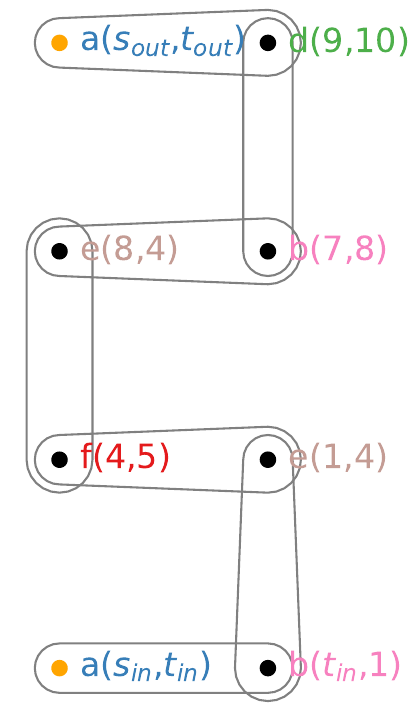}
    \end{subfigure}    \hfill\null
    \caption{Completed gadget used in the proof of \cref{prp:otherhard} regarding $\resprobset(abcd|be|ef)$}
    \label{fig:abcdbeef}
  \end{figure}

  \begin{figure}
    \null\hfill
    \begin{subfigure}[b]{0.4\linewidth}
      \centering
      \includegraphics[scale=0.4]{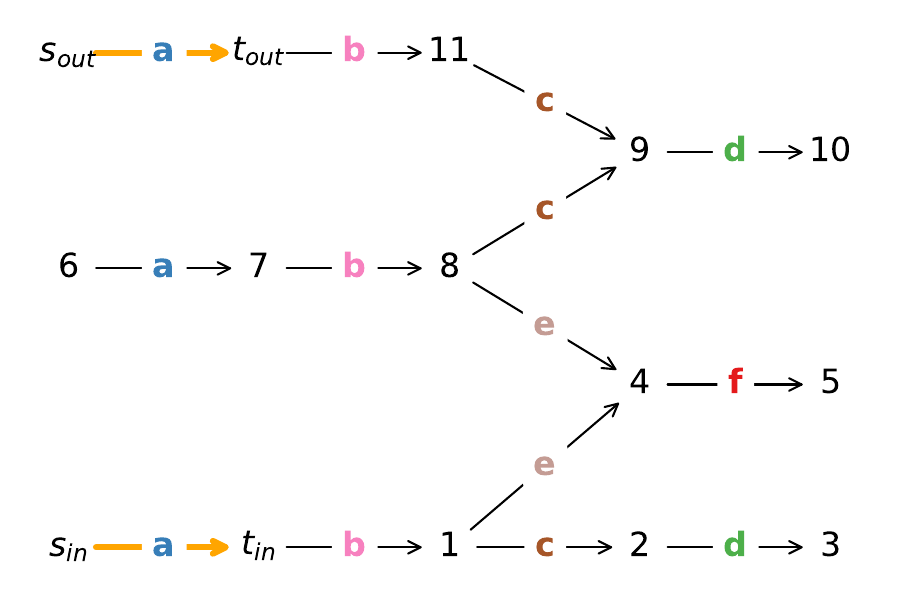}
    \end{subfigure}\hfill
    \begin{subfigure}[b]{0.4\linewidth}
      \centering
      \includegraphics[scale=0.4]{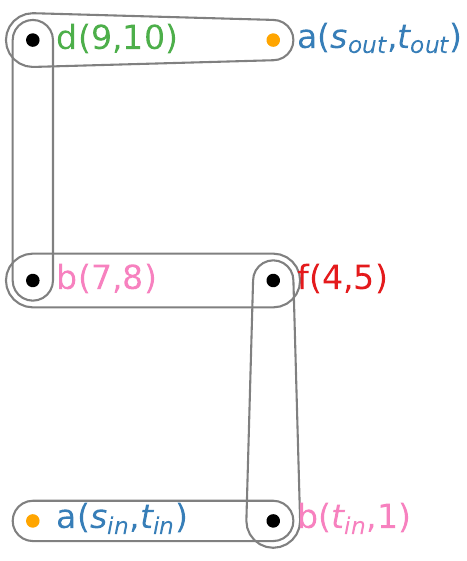}
    \end{subfigure}    \hfill\null
    \caption{Completed gadget used in the proof of \cref{prp:otherhard} regarding $\resprobset(abcd|bef)$}
    \label{fig:abcdbef}
  \end{figure}

  We claim that the gadget shown in \cref{fig:abcdbeef} is valid for $abcd|be|ef$ and the gadget shown in \cref{fig:abcdbef} is valid for $abcd|bef$ (also, they happen to be identical).
  We can see that the gadgets meet all the criteria of a pre-gadget,
  and show that there exists a condensed hypergraph of matches of the
  pre-gadgets for these queries that is an odd path. 
  These properties can be checked by exhaustively considering all walks in the
  gadget, and we conclude that these gadgets are valid hardness gadgets for $abcd|be|ef$ and $abcd|bef$.
\end{proof}

\section{Conclusion and Future Work} 
\label{sec:conclusion}
We have studied the complexity of resilience for RPQs, towards the goal of
characterizing the languages for which the problem is tractable. We study 
bag and set semantics, but
our results so far show no difference between the two settings,
unlike for 
CQs~\cite{makhija2024unified}. 
Refer back to \cref{fig:results} on page~\pageref{fig:results} for a summary of
our results and of some remaining open cases.
We leave open the general question of whether there are
languages $L$, regular or not, for which $\resprobset(L)$ is PTIME but $\resprobbag(L)$
is NP-hard.
We also leave open the question of classifying the complexity of resilience
in intermediate settings, e.g., bag semantics where only weights $1$ and
$\infty$ are allowed; or the setting where some relations are indicated as
\emph{exogenous} in the signature, which amounts to giving a weight of $\infty$
to all their facts. Other settings left open by our work are those where we
delete vertices (or just vertex labels) instead of edges, or where vertices and/or 
edges can be labeled by words or by non-empty words instead of single letters.

At the level of conjectures, 
two directions seem particularly natural. First, for finite
chain languages, we conjecture that the bipartiteness condition of 
\cref{prp:rpn}
is
also necessary for tractability: non-bipartite chain languages should have NP-hard
resilience in set semantics (this is \cref{con:chainnonbip}, with
\cref{prp:abbcca} proving one
representative case).
More generally, at a high level, we still do not fully understand the non-local languages that are not
four-legged, i.e., non-letter-Cartesian languages featuring counterexamples
where some of the ``legs'' are empty.
Second, all tractable classes identified in this paper admit reductions to MinCut or
closely related flow problems (subsuming in particular the
techniques based on submodular optimization used in the earlier
version of this work~\cite{amarilli2025resilience}), and we currently know no tractable example requiring
fundamentally different techniques. It remains open whether flow-based
techniques always suffice whenever the resilience problem is tractable.

Another open question concerns \emph{non-Boolean RPQs}, i.e., where
endpoints are fixed (or given as input).
For instance, the RPQ $aa$ then seems
tractable as possible contingency sets have a simple structure; but
longer RPQs may stay hard. This problem may relate to 
\emph{length-bounded cuts}~\cite{itai1982complexity,10.1145/1868237.1868241},
or
multicuts~\cite{garg1994multiway} and skew multicuts~\cite{kratsch2015fixed}. 
Note that this non-Boolean setting corresponds to the problem of \emph{deletion
propagation with source side effects}~\cite{buneman2002propagation}, which was studied before
resilience: this problem asks about finding a minimum contingency set to remove
a specific query answer. The case of removing multiple tuples has also been
investigated~\cite{kimelfeld2013multituple}. We do not know whether a dichotomy
in one of these settings implies a dichotomy for our problem, or vice-versa.
We also do not understand how our results relate to 
VCSPs~\cite{bodirsky2024complexity}, e.g., whether classes of VCSPs with better
algorithms (in combined complexity, or in terms of the polynomial degree) can
imply tractable algorithms for some RPQ classes; or whether VCSP dichotomies
hold in set semantics (or in bag semantics
with weights restricted to $1$ and $\infty$).

Last, resilience can be studied for 
recursive queries beyond RPQs, in particular two-way RPQs or Datalog, though
this would require new techniques (unlike RPQs, these queries are not
``directional''). 
We note that for two-way RPQs, and more generally for Boolean connected simple
monadic Datalog queries over signatures of arity at most two,
the existence of finite duals~\cite[Theorem~4.6(4)]{bienvenu2014ontology}
yields a PTIME/NP-complete dichotomy for resilience under bag semantics
via the VCSP approach (see~\cite[Lemma~8.14 and Theorem~8.15]
{bodirsky2025complexity}).

\section{Acknowledgments}

The authors are grateful to Charles Paperman for help about local languages,
to Rémi de Pretto for pointing out a minor problem, 
to Blaz Korecic and Martín Andrighetti for insightful discussions 
towards the proof of \cref{prp:submod}, 
and to the reviewers of the conference version for their valuable feedback.
Part of this work was conducted while four of the authors were participating
in the semester-long program 
\href{https://simons.berkeley.edu/programs/logic-algorithms-database-theory-ai}{Logic and Algorithms in Database Theory and AI}
at the Simons Institute for the Theory of Computing.
Amarilli was partially supported by the ANR project EQUUS ANR-19-CE48-0019, by the Deutsche Forschungsgemeinschaft (DFG, German Research Foundation) - 431183758, and by the ANR project ANR-18-CE23-0003-02 (“CQFD”).
Gatterbauer and Makhija were partially supported 
by the National Science Foundation (NSF) under award number IIS-1762268.

\bibliographystyle{plainurl}
\bibliography{BIB/propagation.bib}

@article{DBLP:journals/vldb/GatterbauerS17,
	author = {Wolfgang Gatterbauer and Dan Suciu},
	bibsource = {dblp computer science bibliography, https://dblp.org},
	doi = {10.1007/S00778-016-0434-5},
	journal = {{VLDB} J.},
	note = {Extended version: \url{https://arxiv.org/pdf/1310.6257}},
	number = {1},
	pages = {5--30},
	title = {Dissociation and propagation for approximate lifted inference with standard relational database management systems},
	volume = {26},
	year = {2017}}

@book{aho1974design,
	author = {Aho, Alfred V. and Hopcroft, John E. and Ullman, Jeffrey D.},
	publisher = {Addison-Wesley},
	title = {The design and analysis of computer algorithms},
	year = {1974}}

@book{Goldreich2010PNP,
	author = {Oded Goldreich},
	doi = {10.1017/CBO9780511761355},
	isbn = {9780521192484},
	publisher = {Cambridge University Press},
	title = {P, NP, and NP-Completeness: The Basics of Computational Complexity},
	year = {2010}}

@misc{arxiv:FreireGIM20,
	author = {Cibele Freire and Wolfgang Gatterbauer and Neil Immerman and Alexandra Meliou},
	eprint = {1907.01129v2},
	howpublished = {arXiv},
	note = {Full version with proofs},
	title = {New results for the complexity of resilience for binary conjunctive queries with self-joins},
	year = {2020}}

@inproceedings{bodirsky_et_al:LIPIcs.STACS.2026.15,
	author = {Bodirsky, Manuel and Semani\v{s}inov\'{a}, \v{Z}aneta},
	booktitle = {STACS},
	doi = {10.4230/LIPIcs.STACS.2026.15},
	isbn = {978-3-95977-412-3},
	issn = {1868-8969},
	pages = {15:1--15:20},
	series = {LIPIcs},
	title = {{The complexity of resilience for digraph queries}},
	urn = {urn:nbn:de:0030-drops-255045},
	volume = {364},
	year = {2026}}

@article{DBLP:journals/tods/DayalB82,
	author = {Umeshwar Dayal and Philip A. Bernstein},
	bibsource = {dblp computer science bibliography, https://dblp.org},
	doi = {10.1145/319732.319740},
	journal = {{ACM} Trans. Database Syst.},
	number = {3},
	pages = {381--416},
	title = {On the Correct Translation of Update Operations on Relational Views},
	url = {https://doi.org/10.1145/319732.319740},
	volume = {7},
	year = {1982}}

@inproceedings{DBLP:conf/pods/Kimelfeld12,
	-editor = {Michael Benedikt and Markus Kr{\"{o}}tzsch and Maurizio Lenzerini},
	author = {Benny Kimelfeld},
	bibsource = {dblp computer science bibliography, https://dblp.org},
	biburl = {https://dblp.org/rec/conf/pods/Kimelfeld12.bib},
	booktitle = {PODS},
	date-added = {2026-06-23 22:33:52 -0400},
	date-modified = {2026-06-23 22:33:52 -0400},
	doi = {10.1145/2213556.2213584},
	pages = {191--202},
	publisher = {{ACM}},
	timestamp = {Thu, 19 Feb 2026 16:18:02 +0100},
	title = {A dichotomy in the complexity of deletion propagation with functional dependencies},
	url = {https://doi.org/10.1145/2213556.2213584},
	year = {2012},
	bdsk-url-1 = {https://dblp.org/rec/conf/pods/Kimelfeld12},
	bdsk-url-2 = {https://doi.org/10.1145/2213556.2213584}}

@inproceedings{DBLP:conf/sigmod/MiaoRY19,
	-editor = {Peter Boncz and Stefan Manegold and Anastasia Ailamaki and Amol Deshpande and Tim Kraska},
	author = {Zhengjie Miao and Sudeepa Roy and Jun Yang},
	bibsource = {dblp computer science bibliography, https://dblp.org},
	biburl = {https://dblp.org/rec/conf/sigmod/MiaoRY19.bib},
	booktitle = {SIGMOD},
	date-added = {2026-06-23 22:32:27 -0400},
	date-modified = {2026-06-23 22:32:27 -0400},
	doi = {10.1145/3299869.3319866},
	pages = {503--520},
	publisher = {{ACM}},
	timestamp = {Sun, 02 Nov 2025 21:27:21 +0100},
	title = {Explaining Wrong Queries Using Small Examples},
	url = {https://doi.org/10.1145/3299869.3319866},
	year = {2019},
	bdsk-url-1 = {https://dblp.org/rec/conf/sigmod/MiaoRY19},
	bdsk-url-2 = {https://doi.org/10.1145/3299869.3319866}}

@inproceedings{DBLP:conf/icdt/0005S24,
	-publisher = {Schloss Dagstuhl - Leibniz-Zentrum f{\"{u}}r Informatik},
	-series = {LIPIcs},
	author = {Xiao Hu and Stavros Sintos},
	bibsource = {dblp computer science bibliography, https://dblp.org},
	biburl = {https://dblp.org/rec/conf/icdt/0005S24.bib},
	booktitle = {ICDT},
	date-added = {2026-06-23 22:31:22 -0400},
	date-modified = {2026-10-11 23:40:41 -0400},
	doi = {10.4230/LIPICS.ICDT.2024.24},
	pages = {24:1--24:20},
	timestamp = {Fri, 21 Nov 2025 14:06:54 +0100},
	title = {Finding Smallest Witnesses for Conjunctive Queries},
	volume = {290},
	year = {2024},
	bdsk-url-1 = {https://dblp.org/rec/conf/icdt/0005S24},
	bdsk-url-2 = {https://doi.org/10.4230/LIPIcs.ICDT.2024.24}}

@article{DBLP:journals/pvldb/MakhijaG25,
	author = {Neha Makhija and Wolfgang Gatterbauer},
	bibsource = {dblp computer science bibliography, https://dblp.org},
	doi = {10.14778/3742728.3742756},
	journal = {Proc. {VLDB} Endow.},
	number = {8},
	pages = {2667--2680},
	title = {Is Integer Linear Programming All You Need for Deletion Propagation? {A} Unified and Practical Approach for Generalized Deletion Propagation},
	url = {https://www.vldb.org/pvldb/vol18/p2667-makhija.pdf},
	volume = {18},
	year = {2025}}

@article{DBLP:journals/pvldb/MeliouGMS11,
	author = {Alexandra Meliou and Wolfgang Gatterbauer and Katherine F. Moore and Dan Suciu},
	bibsource = {dblp computer science bibliography, https://dblp.org},
	doi = {10.14778/1880172.1880176},
	journal = {Proc. {VLDB} Endow.},
	number = {1},
	pages = {34--45},
	title = {The Complexity of Causality and Responsibility for Query Answers and non-Answers},
	url = {http://www.vldb.org/pvldb/vol4/p34-meliou.pdf},
	volume = {4},
	year = {2010}}

@article{DBLP:journals/siamcomp/MendelzonW95,
	author = {Alberto O. Mendelzon and Peter T. Wood},
	bibsource = {dblp computer science bibliography, https://dblp.org},
	doi = {10.1137/S009753979122370X},
	journal = {{SIAM} J. Comput.},
	number = {6},
	pages = {1235--1258},
	title = {\href{https://citeseerx.ist.psu.edu/document?repid=rep1&type=pdf&doi=10287acfb08806de28881919978a3da673584b69}{Finding regular simple paths in graph databases}},
	url = {https://doi.org/10.1137/S009753979122370X},
	volume = {24},
	year = {1995}}

@book{Cormen:2009dz,
	author = {Thomas H. Cormen and Charles E. Leiserson and Ronald L. Rivest and Clifford Stein},
	call-number = {QA76.6},
	dewey-call-number = {005.1},
	edition = {3rd ed},
	genre = {Computer programming},
	isbn = {9780262033848 (hardcover : alk. paper)},
	library-id = {2009008593},
	publisher = {MIT Press},
	title = {Introduction to algorithms},
	url = {https://dl.acm.org/doi/book/10.5555/1614191},
	year = {2009}}

@article{10.1145/1868237.1868241,
	-month = dec,
	abstract = {For a given number L, an L-length-bounded edge-cut (node-cut, respectively) in a graph G with source s and sink t is a set C of edges (nodes, respectively) such that no s-t-path of length at most L remains in the graph after removing the edges (nodes, respectively) in C. An L-length-bounded flow is a flow that can be decomposed into flow paths of length at most L. In contrast to classical flow theory, we describe instances for which the minimum L-length-bounded edge-cut (node-cut, respectively) is Θ(n2/3)-times (Θ(√n)-times, respectively) larger than the maximum L-length-bounded flow, where n denotes the number of nodes; this is the worst case. We show that the minimum length-bounded cut problem is NP-hard to approximate within a factor of 1.1377 for L≥ 5 in the case of node-cuts and for L≥ 4 in the case of edge-cuts. We also describe algorithms with approximation ratio O(min{L,n/L}) ⊆ O√n in the node case and O(min {L,n2/L2,√m} ⊆ O2/3 in the edge case, where m denotes the number of edges. Concerning L-length-bounded flows, we show that in graphs with unit-capacities and general edge lengths it is NP-complete to decide whether there is a fractional length-bounded flow of a given value. We analyze the structure of optimal solutions and present further complexity results.},
	address = {New York, NY, USA},
	articleno = {4},
	author = {Baier, Georg and Erlebach, Thomas and Hall, Alexander and K\"{o}hler, Ekkehard and Kolman, Petr and Pangr\'{a}c, Ond\v{r}ej and Schilling, Heiko and Skutella, Martin},
	date-added = {2024-11-28 13:49:48 -0500},
	date-modified = {2024-12-09 21:49:48 -0500},
	doi = {10.1145/1868237.1868241},
	issn = {1549-6325},
	issue_date = {November 2010},
	journal = {ACM Trans. Algorithms},
	keywords = {Approximation algorithms, flows and cuts in graphs, graphs},
	number = {1},
	numpages = {27},
	pages = {1--27},
	publisher = {Association for Computing Machinery},
	title = {\href{https://www.cs.le.ac.uk/people/te17/papers/talg2010.pdf}{Length-bounded cuts and flows}},
	volume = {7},
	year = {2010},
	bdsk-url-1 = {https://doi.org/10.1145/1868237.1868241}}

@article{menger:1927,
	author = {Karl Menger},
	doi = {10.4064/fm-10-1-96-115},
	journal = {Fundamenta Mathematicae},
	pages = {96-115},
	title = {Zur allgemeinen {K}urventheorie},
	volume = {10},
	year = {1927}}

@article{ford1956maximal,
	author = {Ford, Lester Randolph and Fulkerson, Delbert R},
	doi = {https://doi.org/10.4153/cjm-1956-045-5},
	journal = {Canadian journal of Mathematics},
	pages = {399--404},
	publisher = {Cambridge University Press},
	title = {\href{https://web.archive.org/web/20170809102956id_/http://bioinfo.ict.ac.cn/~dbu/AlgorithmCourses/Lectures/FordFulkerson1956.pdf}{Maximal flow through a network}},
	volume = {8},
	year = {1956}}

@article{MeliouGHKMS10,
	-pages = {59-67},
	author = {Alexandra Meliou and Wolfgang Gatterbauer and Joseph Y. Halpern and Christoph Koch and Katherine F. Moore and Dan Suciu},
	bibsource = {DBLP, http://dblp.uni-trier.de},
	date-added = {2024-09-16 09:20:27 -0400},
	date-modified = {2026-10-11 23:38:34 -0400},
	ee = {http://sites.computer.org/debull/A10sept/suciu.pdf},
	journal = {IEEE Data Eng. Bull.},
	number = {3},
	title = {Causality in Databases},
	url = {http://sites.computer.org/debull/A10sept/p59.pdf},
	volume = {33},
	year = {2010}}

@article{DBLP:journals/pvldb/MeliouGS11,
	-pages = {1490--1493},
	author = {Alexandra Meliou and Wolfgang Gatterbauer and Dan Suciu},
	bibsource = {dblp computer science bibliography, https://dblp.org},
	biburl = {https://dblp.org/rec/journals/pvldb/MeliouGS11.bib},
	date-added = {2024-09-16 09:20:16 -0400},
	date-modified = {2026-10-11 23:37:43 -0400},
	doi = {10.14778/3402755.3402803},
	journal = {{PVLDB}},
	number = {12},
	timestamp = {Sat, 25 Apr 2020 13:59:09 +0200},
	title = {Reverse Data Management},
	url = {https://www.vldb.org/pvldb/vol4/p1490-meliou.pdf},
	volume = {4},
	year = {2011},
	bdsk-url-1 = {http://www.vldb.org/pvldb/vol4/p1490-meliou.pdf}}

@article{DBLP:journals/pvldb/FreireGIM15,
	-pages = {180--191},
	author = {Cibele Freire and Wolfgang Gatterbauer and Neil Immerman and Alexandra Meliou},
	bibsource = {dblp computer science bibliography, http://dblp.org},
	biburl = {http://dblp.uni-trier.de/rec/bib/journals/pvldb/FreireGIM15},
	date-added = {2015-12-30 23:30:52 +0000},
	date-modified = {2026-07-15 10:02:17 -0400},
	doi = {10.14778/2850583.2850592},
	journal = {{PVLDB}},
	number = {3},
	pages = {180--191},
	timestamp = {Tue, 24 Nov 2015 14:48:15 +0100},
	title = {The Complexity of Resilience and Responsibility for Self-Join-Free Conjunctive Queries},
	url = {http://www.vldb.org/pvldb/vol9/p180-freire.pdf},
	volume = {9},
	year = {2015},
	bdsk-url-1 = {http://www.vldb.org/pvldb/vol9/p180-freire.pdf}}

@inproceedings{DBLP:conf/pods/FreireGIM20,
	author = {Cibele Freire and Wolfgang Gatterbauer and Neil Immerman and Alexandra Meliou},
	bibsource = {dblp computer science bibliography, https://dblp.org},
	booktitle = {{PODS}},
	doi = {10.1145/3375395.3387647},
	pages = {271--284},
	title = {New results for the complexity of resilience for binary conjunctive queries with self-joins},
	url = {https://doi.org/10.1145/3375395.3387647},
	year = {2020}}

@article{makhija2024unified,
	articleno = {228},
	author = {Makhija, Neha and Gatterbauer, Wolfgang},
	doi = {10.1145/3626715},
	issue_date = {December 2023},
	journal = {{PACMMOD}},
	number = {4},
	numpages = {27},
	pages = {228:1 - 228:27},
	publisher = {ACM},
	title = {A Unified Approach for Resilience and Causal Responsibility with {Integer} {Linear} {Programming} ({ILP}) and {LP} Relaxations},
	url = {https://doi.org/10.1145/3626715},
	volume = {1},
	year = {2023}}

@inproceedings{garg1994multiway,
	author = {Garg, Naveen and Vazirani, Vijay V and Yannakakis, Mihalis},
	booktitle = {ICALP},
	doi = {10.1007/3-540-58201-0_92},
	organization = {Springer},
	pages = {487--498},
	title = {\href{https://citeseerx.ist.psu.edu/document?repid=rep1&type=pdf&doi=ce18a622d6cfc38ff739f850f01a750d33587fd4}{Multiway cuts in directed and node weighted graphs}},
	url = {https://doi.org/10.1007/3-540-58201-0_92},
	year = {1994}}

@article{kratsch2015fixed,
	author = {Kratsch, Stefan and Pilipczuk, Marcin and Pilipczuk, Micha{\l} and Wahlstr\"om, Magnus},
	doi = {10.1137/120904202},
	journal = {SIAM Journal on Discrete Mathematics},
	number = {1},
	pages = {122--144},
	publisher = {SIAM},
	title = {\href{https://citeseerx.ist.psu.edu/document?repid=rep1&type=pdf&doi=4f7b341a8150b2ea893f446616ec851870991112}{Fixed-parameter tractability of multicut in directed acyclic graphs}},
	url = {https://doi.org/10.1137/120904202},
	volume = {29},
	year = {2015}}

@inproceedings{barcelo2019boundedness,
	-organization = {Schloss Dagstuhl -- Leibniz Zentrum f{\"u}r Informatik},
	-series = {Leibniz International Proceedings in Informatics (LIPIcs)},
	author = {Barcel{\'o}, Pablo and Figueira, Diego and Romero, Miguel},
	booktitle = {ICALP},
	date-modified = {2024-12-09 21:50:32 -0500},
	doi = {10.4230/LIPIcs.ICALP.2019.104},
	pages = {104:1--104:15},
	title = {Boundedness of Conjunctive Regular Path Queries},
	url = {https://doi.org/10.4230/LIPIcs.ICALP.2019.104},
	volume = {132},
	year = {2019},
	bdsk-url-1 = {https://doi.org/10.4230/LIPIcs.ICALP.2019.104}}

@book{sakarovitch2009elements,
	author = {Sakarovitch, Jacques},
	doi = {10.1017/CBO9781139195218},
	publisher = {Cambridge University Press},
	title = {Elements of automata theory},
	url = {https://doi.org/10.1017/CBO9781139195218},
	year = {2009}}

@inproceedings{bodirsky2024complexity,
	author = {Bodirsky, Manuel and Semani{\v{s}}inov{\'a}, {\v{Z}}aneta and Lutz, Carsten},
	booktitle = {{LICS}},
	doi = {10.1145/3661814.3662071},
	pages = {14:1--14:14},
	title = {The complexity of resilience problems via valued constraint satisfaction problems},
	url = {https://doi.org/10.1145/3661814.3662071},
	year = {2024}}

@misc{cstheoryepsro,
	author = {Antoine Amarilli},
	howpublished = {Theoretical Computer Science Stack Exchange},
	note = {(version: 2020-11-02)},
	title = {Regular languages accepted by an automaton with at most one transition per letter},
	url = {https://cstheory.stackexchange.com/q/47810}}

@inproceedings{kolmogorov2019testing,
	author = {Kolmogorov, Vladimir},
	booktitle = {{ICALP}},
	doi = {10.4230/LIPIcs.ICALP.2019.77},
	organization = {Schloss-Dagstuhl-Leibniz Zentrum f{\"u}r Informatik},
	pages = {77:1--77:12},
	title = {Testing the Complexity of a Valued {CSP} Language},
	volume = {132},
	year = {2019}}

@book{mpri,
	author = {Jean-\'Eric Pin},
	publisher = {Online book preprint},
	title = {\href{https://www.irif.fr/~jep/PDF/MPRI/MPRI.pdf}{Mathematical foundations of automata theory}},
	url = {https://www.irif.fr/~jep/PDF/MPRI/MPRI.pdf},
	year = {2025}}

@inproceedings{koucky2005bounded,
	author = {Kouck{\`y}, Michal and Pudl{\'a}k, Pavel and Th{\'e}rien, Denis},
	booktitle = {STOC},
	doi = {10.1145/1060590.1060629},
	pages = {257--265},
	title = {\href{https://iuuk.mff.cuni.cz/~koucky/papers/wires.ps}{Bounded-depth circuits: Separating wires from gates}},
	year = {2005}}

@article{barrington2005first,
	-publisher = {Elsevier},
	author = {Barrington, David A Mix and Immerman, Neil and Lautemann, Clemens and Schweikardt, Nicole and Th{\'e}rien, Denis},
	date-modified = {2024-11-28 13:51:46 -0500},
	doi = {10.1016/j.jcss.2004.07.004},
	journal = {Journal of Computer and System Sciences},
	number = {2},
	pages = {101--127},
	title = {First-order expressibility of languages with neutral letters or: The {C}rane {B}each conjecture},
	volume = {70},
	year = {2005},
	bdsk-url-1 = {https://doi.org/10.1016/j.jcss.2004.07.004}}

@incollection{MCCORMICK2005321,
	-editor = {K. Aardal and G.L. Nemhauser and R. Weismantel},
	abstract = {This chapter describes the submodular function minimization problem (SFM); why it is important; techniques for solving it; algorithms by Cunningham, by Schrijver as modified by Fleischer and Iwata, by Iwata, Fleischer and Fujishige, and by Iwata for solving it; and extensions of SFM to more general families of subsets.},
	author = {S. Thomas McCormick},
	booktitle = {Discrete Optimization},
	date-modified = {2024-11-28 13:42:32 -0500},
	doi = {https://doi.org/10.1016/S0927-0507(05)12007-6},
	issn = {0927-0507},
	pages = {321-391},
	publisher = {Elsevier},
	series = {Handbooks in Operations Research and Management Science},
	title = {Submodular Function Minimization},
	volume = {12},
	year = {2005},
	bdsk-url-1 = {https://www.sciencedirect.com/science/article/pii/S0927050705120076},
	bdsk-url-2 = {https://doi.org/10.1016/S0927-0507(05)12007-6}}

@article{itai1982complexity,
	-publisher = {Wiley Online Library},
	author = {Itai, Alon and Perl, Yehoshua and Shiloach, Yossi},
	date-modified = {2024-12-09 21:53:08 -0500},
	doi = {10.1002/net.3230120306},
	journal = {Networks},
	number = {3},
	pages = {277--286},
	title = {The complexity of finding maximum disjoint paths with length constraints},
	volume = {12},
	year = {1982},
	bdsk-url-1 = {https://doi.org/10.1002/net.3230120306}}

@article{amarilli2025approximating,
	author = {Antoine Amarilli and Timothy van Bremen and Octave Gaspard and Kuldeep S. Meel},
	doi = {10.46298/LMCS-21(4:30)2025},
	journal = {Logical Methods in Computer Science},
	number = {4},
	pages = {30:1--30:31},
	title = {\href{https://arxiv.org/abs/2309.13287} {Approximating queries on probabilistic graphs}},
	volume = {21},
	year = {2025}}

@article{hirai2021minimum,
	author = {Hirai, Hiroshi and Mizutani, Ryuhei},
	doi = {10.1016/j.disopt.2021.100642},
	journal = {Discrete Optimization},
	pages = {100642},
	publisher = {Elsevier},
	title = {\href{https://arxiv.org/abs/2006.00153}{Minimum 0-extension problems on directed metrics}},
	volume = {40},
	year = {2021}}

@inproceedings{gottlob2004hypergraph,
	author = {Gottlob, Georg},
	booktitle = {FoIKS},
	doi = {10.1007/978-3-540-24627-5_1},
	organization = {Springer},
	pages = {1--5},
	title = {Hypergraph transversals},
	year = {2004}}

@incollection{karp1972reductibility,
	author = {Karp, Richard M},
	booktitle = {Complexity of Computer Computations. The IBM Research Symposia Series},
	doi = {10.1007/978-1-4684-2001-2_9},
	journal = {Complexity of computer computations},
	pages = {85--103},
	publisher = {Springer},
	title = {Reductibility among combinatorial problems},
	year = {1972}}

@article{ahuja1997computational,
	-publisher = {Elsevier},
	author = {Ahuja, Ravindra K and Kodialam, Murali and Mishra, Ajay K and Orlin, James B},
	date-modified = {2024-12-09 21:49:06 -0500},
	doi = {10.1016/S0377-2217(96)00269-X},
	journal = {European Journal of Operational Research},
	number = {3},
	pages = {509--542},
	title = {\href{https://citeseerx.ist.psu.edu/document?repid=rep1&type=pdf&doi=e5579a220aa778e60c5fea472af2104ab96daf52}{Computational investigations of maximum flow algorithms}},
	volume = {97},
	year = {1997},
	bdsk-url-1 = {https://doi.org/10.1016/S0377-2217(96)00269-X}}

@misc{gadgetrepo,
	author = {Amarilli,	Antoine and Gatterbauer, Wolfgang and Makhija, Neha and Monet, Mika{\"e}l},
	note = {\url{https://osf.io/6kfj3/}},
	title = {Code Repository for Implementation},
	year = {2024}}

@article{kimelfeld2013multituple,
	author = {Kimelfeld, Benny and Vondr{\'a}k, Jan and Woodruff, David P},
	doi = {10.14778/2536258.2536267},
	journal = {PVLDB},
	number = {13},
	pages = {1558--1569},
	title = {Multi-tuple deletion propagation: Approximations and complexity},
	url = {http://www.vldb.org/pvldb/vol6/p1558-kimelfeld.pdf},
	volume = {6},
	year = {2013}}

@inproceedings{buneman2002propagation,
	author = {Buneman, Peter and Khanna, Sanjeev and Tan, Wang-Chiew},
	booktitle = {PODS},
	doi = {10.1145/543613.543633},
	pages = {150--158},
	title = {On propagation of deletions and annotations through views},
	year = {2002}}

@article{amarilli2025resilience,
	articleno = {108},
	author = {Antoine Amarilli and Wolfgang Gatterbauer and Neha Makhija and Mika{\"e}l Monet},
	booktitle = {{PODS}},
	doi = {10.1145/3725245},
	journal = {PACMMOD (PODS)},
	number = {2},
	numpages = {18},
	pages = {108:1--108:18},
	title = {\href{https://dl.acm.org/doi/pdf/10.1145/3725245}{Resilience for regular path queries: Towards a complexity classification}},
	volume = {3},
	year = {2025}}

@article{martens2019dichotomies,
	author = {Martens, Wim and Trautner, Tina},
	doi = {10.1145/3331446},
	journal = {ACM Transactions on Database Systems (TODS)},
	number = {4},
	pages = {1--46},
	publisher = {ACM New York, NY, USA},
	title = {Dichotomies for evaluating simple regular path queries},
	volume = {44},
	year = {2019}}

@article{ito1991outfix,
	author = {Ito, Masami and J{\"u}rgensen, Helmut and Shyr, Huei-Jan and Thierrin, Gabriel},
	doi = {10.1016/0022-0000(91)90026-2},
	journal = {Journal of Computer and System Sciences},
	number = {3},
	pages = {484--508},
	publisher = {Elsevier},
	title = {Outfix and infix codes and related classes of languages},
	volume = {43},
	year = {1991}}

@unpublished{amarilli2025locality,
	author = {Antoine Amarilli and Mika{\"e}l Monet and R{\'e}mi De Pretto},
	note = {Preprint: \url{https://arxiv.org/abs/2511.07361}},
	title = {\href{https://arxiv.org/abs/2511.07361} {Locality testing for NFAs is PSPACE-complete}},
	year = {2025}}

@unpublished{amarilli2025confluence,
	author = {Antoine Amarilli and Mika{\"e}l Monet and R{\'e}mi De Pretto},
	note = {Preprint: \url{https://arxiv.org/abs/2510.09286}},
	title = {\href{https://arxiv.org/abs/2510.09286} {Confluence of the node-domination and edge-domination hypergraph rewrite rules}},
	year = {2025}}

@article{berstel1996local,
	author = {Berstel, Jean and Pin, Jean-Eric},
	doi = {10.1016/0304-3975(95)00104-2},
	journal = {Theoretical Computer Science},
	number = {2},
	pages = {439--446},
	publisher = {Elsevier},
	title = {Local languages and the {B}erry-{S}ethi algorithm},
	volume = {155},
	year = {1996}}

@misc{bodirsky2025complexity,
	author = {Manuel Bodirsky and {\v Z}aneta Semani{\v s}inov{\'a} and Carsten Lutz},
	eprint = {2309.15654v6},
	howpublished = {arXiv},
	note = {Full version with proofs},
	title = {The Complexity of Resilience Problems via Valued Constraint Satisfaction},
	year = {2025}}

@inproceedings{van2023deterministic,
	author = {van den Brand, Jan and Chen, Li and Kyng, Rasmus and Liu, Yang P and Peng, Richard and Gutenberg, Maximilian Probst and Sachdeva, Sushant and Sidford, Aaron},
	booktitle = {FOCS},
	doi = {10.1109/FOCS57990.2023.00037},
	organization = {IEEE},
	pages = {503--514},
	title = {A deterministic almost-linear time algorithm for minimum-cost flow},
	year = {2023}}

@article{bienvenu2014ontology,
	author = {Bienvenu, Meghyn and ten Cate, Balder and Lutz, Carsten and Wolter, Frank},
	doi = {10.1145/2661643},
	journal = {{TODS}},
	month = dec,
	number = {4},
	pages = {33:1--33:44},
	title = {Ontology-Based Data Access: A Study through Disjunctive {Datalog}, {CSP}, and {MMSNP}},
	volume = {39},
	year = {2014}}

\end{document}